%% file: main.tex
\documentclass{amsart}
\usepackage{graphicx} 
 \usepackage{subcaption}
\usepackage{xcolor}
\usepackage{tikz, pgfplots}
\usetikzlibrary{decorations.fractals}
\usetikzlibrary{intersections}
\usepgfplotslibrary{fillbetween} 
\usetikzlibrary{lindenmayersystems}
\usepackage{amsmath, amssymb}
\newtheorem{lemma}{Lemma}
\newtheorem{theorem}{Theorem}
\newtheorem{observation}{Observation}
\newtheorem{definition}{Definition}
\newtheorem{corollary}{Corollary}

\title[A proof of Jordan curve theorem using sweepline algorithm]{A proof of Jordan curve theorem based on the sweepline algorithm for trapezoidal decomposition of a polygon}
\author[Apurva Mudgal]{Apurva Mudgal, \\ Department of Computer Science and Engineering, \\ Indian Institute of Technology Ropar, \\ Rupnagar, Punjab - 140001, \\ Email: {\tt apurva@iitrpr.ac.in}}

\begin{document}

\maketitle
\begin{abstract}
We prove the Jordan curve theorem by generalizing the sweepline algorithm for trapezoidal decomposition of a polygon.
Our proof uses Zorn's lemma (or, equivalently the axiom of choice).

Though several proofs have been given for the Jordan curve theorem by various authors, ours is the {\bf first algorithmic proof} of Jordan curve theorem using computational geometry. Further, with some preparation, the proof can be taught as part of an undergraduate discrete mathematics course, where till now the proof of this theorem was considered inaccessible.
\end{abstract}

\tableofcontents

\input{introduction-001.tex}

\input{preliminaries-001.tex}

\section{Proving the existence of interior}

In this section, we prove the following theorem:

\begin{theorem}
\label{thm-int-jct}
{\bf Interior Jordan Curve Theorem.} There is a bounded, connected and open set $int(J)$ such that $bd(int(J)) = J$.
\end{theorem}

\input{horizontal-sweep-001.tex}
\input{extension-horizontal-sweep-001.tex}

\input{extension-piecewise-vertical-jordan-001.tex}
\input{infinite-ray-001.tex}

\input{sweepline-jordan-001.tex}

\input{zorn-lemma-001.tex}
\input{root-segment-001.tex}

\subsection{Interior Jordan curve theorem}

\input{exterior-jct-001.tex}

\input{jct-001.tex}

\newpage
\appendix
\setcounter{theorem}{0}

\end{document}

%% file: introduction-001.tex
\section{Introduction}

The Jordan curve theorem (JCT) is considered a foundational result in topology, with a non-trivial proof. Since the original proof of Jordan (proved correct in \cite{hales}), many proofs of this theorem have been given (see \cite{compendium1, compendium2} for collections of various proofs of this theorem).

The Jordan curve theorem for the special case of polygons, appears as a basic building block of 
many proofs of JCT. A polygon is a discrete geometric object and can be stored in a computer.
Algorithms on polygons form the basic subroutines of any computer graphics library.

Therefore, several algorithms have been devised in computational geometry to distinguish between the interior and exterior of a polygon such as the ray casting algorithm and the winding number algorithm for the point-in-polygon problem \cite{point-poly}, and the sweepline algorithm using the trapezoidal decomposition of a polygon \cite{overmars}. These algorithms can also be viewed as algorithmic proofs of the Jordan curve theorem for polygons.

However, although the Jordan curve theorem is more than a century old, there is no known algorithmic proof of this theorem for an arbitrary Jordan curve. In this paper, we give the first algorithmic proof of the
Jordan curve theorem, by generalizing the sweepline algorithm for trapezoidal decomposition of a polygon.\\

{\bf Salient features of our proof.} We now list some key features of our proof:

\begin{enumerate}
 \item The proof constitutes a natural generalization of the sweepline algorithm for polygons to
Jordan curves. As such, it is the first {\it algorithmic proof} of the Jordan curve theorem.

\item No discrete approximation is made by our algorithm such as the well-known approximation of a Jordan curve by a converging sequence of polygons.  The sweepline algorithm works directly on the Jordan curve itself.  

\item The sweepline algorithm of this paper can be viewed as one of the first results in “topology from an algorithmic viewpoint”, which we define as an area concerned with algorithmic proofs of topological theorems. This is to be contrasted with the well-known field of “computational topology”, which computes topological properties on discrete approximations of topological objects such as homology groups of cell complexes, unknotting problem, etc. 

\item The proof uses Zorn's lemma, which appears as a natural extension of the induction principle for proving algorithm termination for our case.

\item The proof only assumes the truth of JCT for rectilinear polygons. 

\item The proof does not use the concepts of winding number, Brouwer's fixed point theorem, structure of the convergence of successive polygon approximations, etc., which occur in some form or other in many proofs of the Jordan curve theorem.

\item The recursion tree of the sweepline algorithm is infinite, and further each node can have countably infinite children. Thus, it can be executed only on variants of infinite time Turing machines (see \cite{hamkins-lewis} for one such model).

\item Since the proof is based on the sweepline algorithm, there is the possibility of generalization to higher dimensions using sweeping planes, etc.

\end{enumerate}

{\bf Basic idea of our proof: example of the Koch snowflake.} The first few steps of a sweepline algorithm on the Koch snowflake are illustrated in Figure \ref{koch}. Unlike the sweepline algorithm for a polygon, we do not sweep a single vertical sweepline from left to right of the Jordan curve, as this approach runs into some difficulties. Instead, we first execute a horizontal sweep (the counterpart of trapezoids) and then extend the area swept by successive horizontal sweeps, executing at most countably infinite horizontal sweeps during
the sweepline algorithm.

The basic building block of the sweepline algorithm is a horizontal sweep, which is the counterpart of trapezoids in the sweepline algorithm for a polygon. The horizontal sweep $H(t)$ is specified by a horizontal segment $t$ such that $t \subset \mathbb{R}^2 - J$, where $J$ is the Koch snowflake (or, the given Jordan curve $J$). The region swept by a horizontal sweep consists of all open segment $s_p$, where $p \in t$. The open segment $s_p$ at a point $p$ is the maximal vertical segment containing $p$ such that $int(s_p) \cap J = \phi$. 

The Koch snowflake has been drawn up to a constant number of iterations of the corresponding Lindenmayer system. The region swept by the starting horizontal sweep $H(t_1)$, for this finite iteration version of Koch snowflake, is the region enclosed by a Jordan curve consisting of a finite number of arcs of the finite iteration Koch snowflake and a finite number of vertical segments (drawn in red). Since the Koch snowflake is a fractal curve, the actual area swept by $H(t_1)$ (with respect to the fractal curve, and not the finite iteration approximation illustrated here) will be bounded by a Jordan curve which has countably infinite vertical segments and the remaining points of this curve belong to the Koch snowflake. We call such curves {\it piecewise-vertical Jordan curves} with respect to the Koch snowflake.
(Note that a piecewise-vertical Jordan curve is always defined with respect to a Jordan curve $J$.) 

At any point of time, the region swept by our sweepline algorithm will be bounded by a piecewise-vertical
Jordan curve. The next step of the sweepline algorithm will consist of a horizontal sweep $H(t')$ such that
one endpoint of $t'$ is on a vertical segment of the current boundary and the remaining segment lies in the
exterior of the region swept till now. There can be countably infinite steps in our sweepline algorithm.

With respect to the Koch snowflake, note that $H(t_1)$ is extended by $H(t_2)$, which is further
extended by $H(t_3)$. The next two horizontal sweeps $H(t_4)$ and $H(t_5)$ result by
extending two other vertical segments of the boundary of $(H(t_1), H(t_2), H(t_3))$, the area swept by
the first three horizontal sweeps. The recursion tree of this sweepline algorithm with finite number of steps
has $H(t_1)$ as the root node, $H(t_2)$, $H(t_4)$ and $H(t_5)$ are the children of $H(t_1)$ and 
$H(t_3)$ is the child of $H(t_2)$.

By the symmetry of the Koch snowflake, the reader can imagine that the initial sweep $H(t_1)$ can be
extended by countably infinite horizontal sweeps, corresponding to various interior points of various
vertical segments of the boundary of $H(t_1)$. Each of these children of $H(t_1)$ can again each be extended by countably infinite horizontal sweeps, and this pattern will continue further with grandchildren
of $H(t_1)$, the children of grandchildren of $H(t_1)$, and so on till infinity. Thus, starting from a point
in the interior of the Koch snowflake, the complete interior can be swept by a sweepline algorithm
such that each node of its recursion tree has at most countably infinite children and any downward path from the root of the recursion tree has at most countably infinite nodes. The total number of horizontal sweeps
is countably infinite; however, the total number of downward paths are uncountable.

This concludes an example of the actual run of the sweepline algorithm on the Koch snowflake (a fractal Jordan curve), which completely sweeps its interior. However, note that the existence of the interior required
us to assume that the Jordan curve theorem is true for the Koch snowflake to start with.

After discussing this example, we will now outline the argument of this paper, which uses the sweepline
algorithm with Zorn's lemma, to prove the Jordan curve theorem from first principles.
\\

\begin{figure}
\centering
\begin{tikzpicture}

\node[anchor=south west,inner sep=0] (image) at (0,0) {\includegraphics[width=8cm]{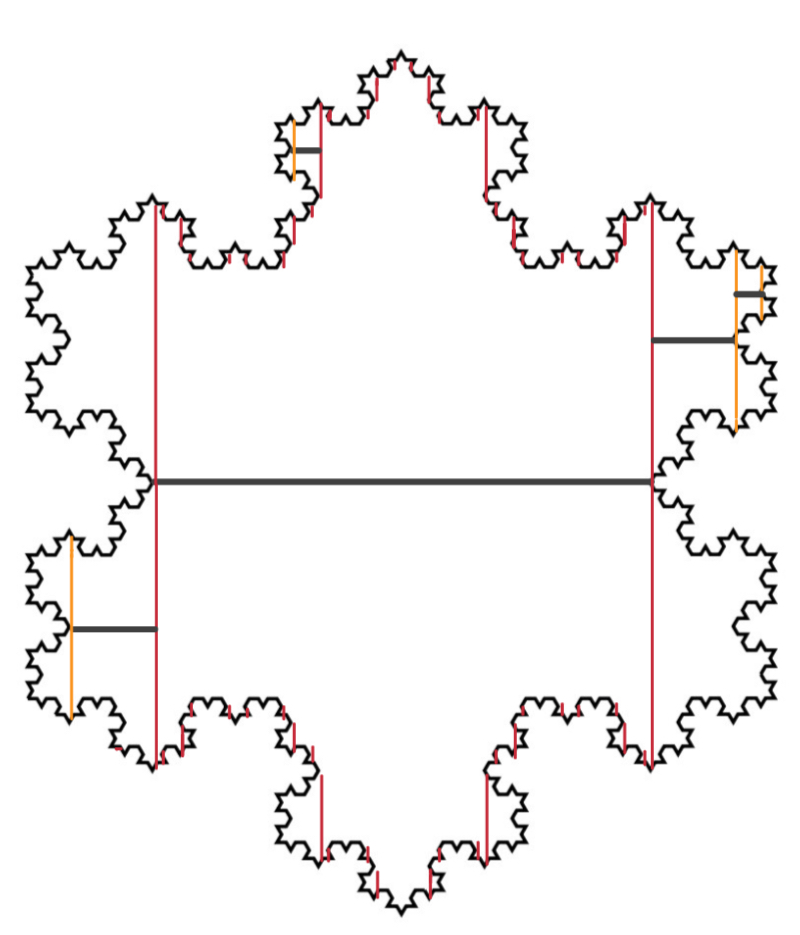}};
\node (A) at (4.2, 4.2) {$t_1$};
\node (B) at (7, 6.3) {$t_2$};
\node (C) at (7.5, 7.2) {$t_3$};
\node (D) at (1, 2.8) {$t_4$};
\node (E) at (2.5, 8) {$t_5$};
\end{tikzpicture}
\caption{First few steps of a sweepline algorithm on the Koch snowflake.}
\label{koch}
\end{figure}
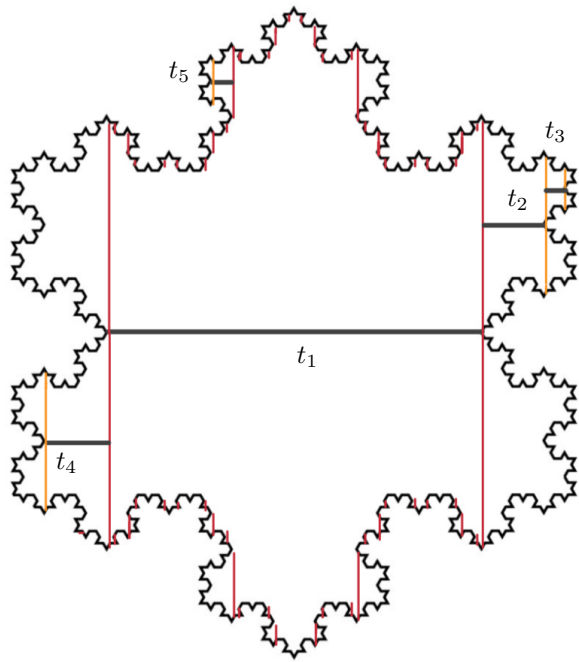

{\bf Outline of the paper.} In Section \ref{sec-preliminaries}, we define the basic terms used in the paper.
In particular, we define the open segment $s_p$ for a point $p \in \mathbb{R}^2$. We also define an equivalence
relation on the set of open segments: {\it two open segments are related to each other if and only if there is
rectilinear path between points in their respective interiors.} For  Sections $3.1$-$3.6$, we make the following assumption: \\

{\bf Assumption $1$.} There exists a point $p^* \in \mathbb{R}^2 - J$ such that its equivalence class $[s_{p^*}]$ has
no open segments of infinite length. \\

We prove that such a point exists in Section $3.7$, after describing our sweepline algorithm and proving its correctness.

Our objective in Sections $3.1$-$3.5$ is to prove that the region swept by any sweepline algorithm (with countably
infinite horizontal sweeps and whose first horizontal sweep $H(t_1)$ has $p^* \in int(t_1)$) is an open, connected set and its boundary is a piecewise-vertical Jordan curve, with respect to the given Jordan curve $J$. The region swept by a sweepline algorithm $\mathcal{S}$ is called its interior and is denoted by $int(\mathcal{S})$. The exterior $ext(\mathcal{S})$ is defined as $\mathcal{R}^2 - cl(\mathcal{S})$, where
$cl(\cdot)$ denotes the operation of taking the closure of the given set.

Though this property is sufficient to prove that the exterior of a sweepline algorithm is an open and unbounded set, it
does not establish that it is connected. We call this restricted statement the {\it interior Jordan curve theorem (JCT)}.

In Sections $3.1$-$3.5$, we prove interior JCT for five successively complex classes of sweepline algorithms:
(i) a single horizontal sweep (section $3.1$), (ii) two successive horizontal sweeps (a sweepline algorithm with two atomic steps; an atomic step is a single horizontal sweep) (section $3.2$), (iii) a finite number of horizontal sweeps (section
$3.3$), (iv) a {\bf ray} i.e., a sweepline algorithm whose recursion tree is a path (section $3.4$), and finally
(v) the most general sweepline algorithm, whose recursion tree is a rooted infinite tree, where each node has at most countably infinite children (section $3.5$).
 
In addition to the above, Section $3.4$ on rays establishes important results about
the structure of rays, which play a critical role in the proof of interior JCT for $J$ (see Section $3.6$).
We first show that every ray $r$ has a limiting segment $s^*_r$. If this limiting segment has positive
length, the ray is called a {\it non-terminating ray}. (Note that all the rays in the recursion tree of the sweepline algorithm described above, which covers the interior of the Koch snowflake, are terminating rays.) The main structural result is Lemma \ref{lem-non-term} which proves that for every non-terminating infinite ray $r$, there exists a finite ray $r'$ such that the region swept by $r'$ is the same as that swept by $r$.

Whereas the previous sections prove that any sweepline algorithm satisfies the interior Jordan curve theorem, it leaves open the question of whether there exists a sweepline algorithm which {\it completely sweeps} the ``interior" of the Jordan curve (the word ``interior" is in double quotes because we have not yet proved its existence). In our framework, we want to prove the existence of a sweepline algorithm $\mathcal{S}_{max}$ such that $bd(int(\mathcal{S}_{max})) = J$.

We now come to this last part of the argument (see Section \ref{sec-zorn-lemma}), which depends on the Zorn's lemma. We consider the set $\mathcal{T}$ of all possible sweepline algorithms starting from $p^*$. We define a partial order $<$ on $\mathcal{T}$. Intuitively, sweepline algorithm $\mathcal{S}_1$ is less than the sweepline algorithm $\mathcal{S}_2$ if and only if $\mathcal{S}_2$ is an extension (or, continuation) of $\mathcal{S}_1$. We prove (using Lemma \ref{lem-non-term}) that every chain $\mathcal{C}$ (with respect to $<$) in $\mathcal{T}$ has a maximal element. By Zorn's lemma, there is a maximal element $\mathcal{S}_{max}$ in $\mathcal{T}$. Since $\mathcal{S}_{max}$ is a sweepline algorithm, the interior Jordan curve theorem is true for $\mathcal{S}_{max}$. It follows from the maximality of $\mathcal{S}_{max}$ that $bd(\mathcal{S}_{max})$ has no vertical segments (if it had a vertical segment $u$, we can extend $\mathcal{S}_{max}$ by a horizontal sweep $H(t')$ with one endpoint of $t'$ in $int(u)$ and thus contradict its maximality with respect to $<$). Hence, we conclude that $bd(\mathcal{S}_{max}) = J$.

The above argument assumes the existence of a point $p^*$ such that $[s_{p^*}]$ has no open segment of infinite length. In Section \ref{sec-root-segment}, we prove the existence of such a point $p^*$, and hence make our proof of interior Jordan curve theorem unconditional. The final theorem on the existence of the interior of $J$ is stated in
Section $3.8$.

Once we have established the existence of $int(J)$, the final step is to prove that $ext(J)$ is a connected
set. This is achieved in two parts. In Section $4$, we prove that there exists one more equivalence class $ext'(J)$ of open segments, which is disjoint from $int(J)$.  The key ingredient of the proof is to apply the inversion map to the Euclidean plane with respect to a point of $int(J)$ and then apply interior JCT to the image of the Jordan curve $J$ under this inversion. 

This still leaves open the possibility that there is a third connected component of $\mathbb{R}^2 - J$, except $int(J)$ and $ext'(J)$. This possibility is ruled out in Section $5$ and the proof of Jordan curve
theorem using the sweepline algorithm is completed.

%% file: preliminaries-001.tex
\section{Preliminaries}
\label{sec-preliminaries}

{\bf Definition of a closed curve.} $\mathbb{S}^1$ is the unit circle $\{ (x,y) ~| ~x^2 + y^2 = 1\}$. $\mathbb{R}^2$ is the Euclidean plane. A closed curve $C$ is
the image of $\mathbb{S}^1$ under a continuous function $\psi: \mathbb{S}^1 \rightarrow \mathbb{R}^2$.\\

{\bf Definition of Jordan curve.} Let $\psi: \mathbb{S}^1 \rightarrow \mathbb{R}^2$ be a one-to-one and continuous function. The Jordan curve $J$ is the image
$\{ \psi(z) ~| ~z \in \mathbb{S}^1\}$ of the unit circle under map $\psi$.

Since $\psi$ is continuous, for every converging sequence $z_1, z_2, \ldots$ of points in $\mathbb{S}^1$, $\lim_{i \rightarrow \infty} \psi(z_i) = \psi( \lim_{i \rightarrow \infty} z_i)$. Since $\psi$ is one-to-one, for any two distinct points $z_1, z_2 \in \mathbb{S}^1$, $\psi(z_1) \neq \psi(z_2)$.

Note that $\mathbb{S}^1$ is a compact set. Thus, every countably infinite sequence $z_1, z_2, \ldots$ of points of
$\mathbb{S}^1$ has a converging subsequence.

Let $\psi_x: \mathbb{S}^1 \rightarrow \mathbb{R}$ be the function defined as follows: for every point $z \in \mathbb{S}^1$, $\psi_x(z)$ is the $x$-coordinate of point $\psi(z)$. Note that $\psi_x$ is a continuous function.
Let $\psi_y: \mathbb{S}^1 \rightarrow \mathbb{R}$ be the function defined as follows: for every point $z \in \mathbb{S}^1$, $\psi_y(z)$ is the $y$-coordinate of point $\psi(z)$. Note that $\psi_y$ is a continuous function.

Since $\psi_x$ and $\psi_y$ are continuous functions on a compact set, they achieve a maximum value and a minimum value.
Thus, there exists a constant $B > 0$ such that all points of
the Jordan curve $J$ lie within the {\it bounding box} $\mathcal{B} = [-B, B] \times [-B, B]$.\\

{\bf Jordan curve theorem (JCT) for rectilinear polygons.} We assume that the Jordan curve theorem is true for all rectilinear polygons. In fact, one can prove the Jordan curve theorem for rectilinear polygons using the sweepline algorithm. \\

{\bf Removing degeneracy.} Let $s$ be a line segment in $\mathbb{R}^2$. We say that $s$ is a line segment of $J$ if and only if $s \subset J$. We further say that $s$ is a {\it maximal} line segment of $J$ if and only if (i) $s \subset J$ and (ii) there is no line segment $s'$ such that $s \subsetneq s' \subset J$. 

Let $M_J$ be the set of all maximal line segments of $J$. We prove that $M_J$ is a countable set

\begin{observation}
The set $M_J$ of maximal line segments is countable.
\end{observation}

\noindent {\bf Proof:} Let $s_1, s_2 \in M_J$ be two distinct segments in $M_J$. Note that $s_1$ and $s_2$ have disjoint interiors. Then, $\psi^{-1}(s_1)$ and $\psi^{-1}(s_2)$ are distinct arcs (i.e., have disjoint interiors) of the unit
circle $\mathbb{S}^1$. Further, for every $s \in M_J$, $\psi^{-1}(s)$ is an arc of positive length.
Since, the unit circle can be partitioned into at most countably infinite arcs of positive length, the
statement is proved. $\blacksquare$\\

Let $0 < \theta < 2 \pi$. 
Let $J_{\theta}$ be the image of $J$ under a rotation of $\mathbb{R}^2$ about the origin by angle $\theta$ in clockwise direction.

\begin{observation}
There exists an angle $0 < \theta^* < 2 \pi$ such that for every vertical line $l'$, $l' \cap J_{\theta^*}$ contains no line segments of positive length. 
\end{observation}

\noindent {\bf Proof:} Let $s_1, s_2, \ldots$ be an enumeration of segments in $M_J$. For
each $i \in \mathbb{N}$, let $\theta_i$ be the counterclockwise angle between the positive $X$-axis and the unique line containing segment $s_i$. Choose $\theta^* \in (0, 2 \pi) - \{ \theta_1, \theta_2, \ldots \}$ to
prove the statement. $\blacksquare$\\

In the following, we use $J$ to denote the rotated curve $J_{\theta^*}$.\\

{\bf Completing the plane $\mathbb{R}^2$.} Let $\mathcal{V}$ be the set of all vertical lines in the plane $\mathbb{R}^2$. 
To each line $l' \in \mathcal{V}$, we add two points with $y$-coordinates equal to
$-\infty$ and $+\infty$, respectively.

We complete $\mathbb{R}^2$ by adding two lines $\omega_{-\infty}$ and $\omega_{+\infty}$ at infinity. The line $\omega_{-\infty}$ consists of all points of lines in $\mathcal{V}$ with $y$-coordinate equal to $-\infty$. Similarly, the line $\omega_{+\infty}$ consists of all points of lines in $\mathcal{V}$ with $y$-coordinate equal to $+\infty$.\\

{\bf Definition of open segment $s_p$.} Let $p \in \mathbb{R}^2 - J$. As before, let $l_p$ denote the vertical line passing through $p$. Consider the subset $S_p = l_p - J$ of line $l_p$. Clearly, $p \in S_p$.

Given two points $q_1, q_2 \in l_p$, we say that $q_1 \leq q_2$ if and only if the $y$-coordinate of $q_1$ is less than or equal to the $y$-coordinate of $q_2$.
Let $U_p = \{ x ~| ~x \in l_p \cap J ~and ~x > p\}$ and $L_p = \{ x ~| ~x \in l_p \cap J ~and ~x < p\}$.
Let $u_p$ be the greatest lower bound (infimum) of $U_p$ and $b_p$ be the least upper bound (supremum) of $L_p$. Note that since $l_p \cap J$ is a closed subset of $l_p$, both $u_p$ and $b_p$ belong to the Jordan curve $J$.

The {\it open segment} $s_p$ at point $p$ is the set $\{ x ~| ~x \in l_p ~and ~b_p < x < u_p\}$. The boundary of $s_p$ consists of the two points $b_p$ and $u_p$. Note that $s_p \subset S_p$.\\

\begin{observation}
Let $p \in \mathbb{R}^2 - J$. Then, $l_p - J$ is the disjoint union of at most countable open
segments.
\end{observation}

\noindent {\bf Proof:} Since $l_p$ and $J$ are both closed subsets, their intersection
$l_p \cap J$ is also a closed subset of $l_p$. The complement $l_p - (l_p \cap J) = l_p - J$ is an open set. An open set of $\mathbb{R}^1$ is a disjoint union of at most countable
open segments.
$\blacksquare$\\

{\it Note.} The set $l_p \cap J$ can be complicated. The {\it Cantor set} is constructed as follows \cite{stromberg}. Start with the
closed unit interval $I = [0,1]$ on the $X$-axis. Remove the open middle third i.e., the interval $(\frac{1}{3}, \frac{2}{3})$ to obtain two disjoint
closed intervals $I_1 = [0, \frac{1}{3}]$ and $I_2 = [\frac{2}{3}, 1]$. Repeat the above operation (removing open middle third interval) recursively on $I_1$ and $I_2$. The recursion is carried out till infinity. Let $\mathcal{I}$ be the set of all removed open intervals. The set $\mathcal{I}$ is countable. However, the Cantor set is uncountable. 

We now construct a Jordan curve from the Cantor set. For every
open interval $I$ removed in any recursive step, we add a semi-circle
connecting its two endpoints such that the semi-circle lies above the $X$-axis. The union of the Cantor set and the added semi-circles
(a countably infinite number of semi-circles are added) constitutes
a Jordan arc $A$ from $(0,0)$ to $(1,0)$. We can construct a Jordan
curve $K$ from $A$ by adding the line segment from $(0,0)$ to $(0,-1)$, the line segment from $(0,-1)$ to $(1, -1)$ and the line
segment from $(1, -1)$ to $(1, 0)$.

We now rotate the plane by $90$ degrees about the origin in anticlockwise direction so that positive $X$-axis goes to the positive $Y$-axis. If we take $p = (0,0)$, $l_p \cap K$ is equal to the Cantor
set. 

One can construct really complicated Jordan curves. In \cite{bishop}, the author proves that there exists a Jordan curve $K$ such that for every line $l$, $l \cap K$ contains a Cantor-type set and hence is uncountable.

Thus, $l_p - J$ has a simpler structure than $l_p \cap J$. We conclude by noting a non-intuitive property of the set $\mathcal{I}$ of open intervals removed during the construction of the Cantor set: {\it for any interval $I \in \mathcal{I}$, there is no immediately next or immediately previous interval in $\mathcal{I}$}. 

{\bf Topologist's sine curve vs $x \sin \frac{1}{x}$.} A key distinguishing property of Jordan curves that
we will use repeatedly in this paper can be illustrated by comparing two curves: (i) the topologist's 
sine curve $\{ \big(x, \sin \big( \frac{1}{x} \big) \big) ~: ~x \in (0,1] \} \cup \{ (0,0) \}$ and (ii)
the damped curve $\{ \big(x, x \cdot \sin \big( \frac{1}{x} \big) \big) ~: ~x \in (0,1] \} \cup \{ (0,0) \}$. The first curve is not a Jordan arc (i.e., portion of a Jordan curve) because the oscillations are not damped as $x$ approaches $0$, whereas the second curve is a Jordan arc because the oscillations are damped by the factor of $x$, as $x$ approaches $0$. We formalize this intuition in the following observation:

\begin{observation}
\label{obs-jordan-arc-limit}
Let $\gamma_1, \gamma_2, \ldots$ be an infinite sequence of mutually disjoint arcs of $J$. For each $i \in \mathbb{N}$,
let $u_i$ and $v_i$ be the two endpoints of arc $\gamma_i$. If both $\lim_{i \rightarrow \infty} u_i$ and $\lim_{i \rightarrow \infty} v_i$ exist, then they are equal. 
\end{observation}

\noindent {\bf Proof:} For each $j \in \mathbb{N}$, let $I_j = \{ \psi^{-1}(p) ~| ~p \in \gamma_j\}$ (since $\psi$ is one-to-one, every point of $J$ has a unique preimage). Note that $I_j$ is a closed interval of
$\mathbb{S}^1$ of positive length. Further, for $j_1 \neq j_2$ ($j_1, j_2 \in \mathbb{N}$), $I_{j_1} \cap I_{j_2} = \phi$ (since Jordan arcs $\gamma_{j_1}$ and $\gamma_{j_2}$ are disjoint).

Consider the sequence of disjoint closed intervals $I_{1}, I_{2}, \ldots$. Note that, by countable additivity of measure

$$\sum_{j=1}^{\infty} |I_{j}| \leq 2 \pi$$

, where $|I_j|$ denotes the arc length of $I_j$ and $2 \pi$ is the perimeter of the unit circle.

Thus, there exists a subsequence $I_{b_{1}}, I_{b_{2}}, \ldots$ ($b_1 < b_2 < \cdots$) of intervals such that $\lim_{j \rightarrow \infty} |I_{b_j} | = 0$. In other words,

$$\lim_{j \rightarrow \infty} \big| \psi^{-1}(u_{b_j}) - \psi^{-1}(v_{b_j}) \big| = 0$$

, where $|y_1 - y_2|$, $y_1, y_2 \in \mathbb{S}^1$, denotes the
length of the smaller of the two arcs formed by $y_1$ and $y_2$ 
in $\mathbb{S}^1$.

Since $\mathbb{S}^1$ is a compact set, the sequence $x_1 = \psi^{-1}(u_{b_1}), x_2 = \psi^{-1}(u_{b_2}), \ldots$ has a convergent 
sequence. Suppose the convergent subsequence is $x_{f_1}, x_{f_2}, \ldots$,
$f_1 < f_2 < \ldots$. Suppose $x^*$ is the limit of this subsequence. By continuity of $\psi$, $\lim_{j \rightarrow \infty} \psi(x_{f_j}) = \psi(x^*)$. Note that $\lim_{j \rightarrow \infty} \psi(x_{f_j}) = \lim_{j \rightarrow \infty} u_{b_{f_j}} = \lim_{j \rightarrow \infty} u_j$ (since we assume that the last limit exists). Let $u^*$ denote the limit $\lim_{j \rightarrow \infty} u_j$. We conclude that $\psi(x^*) = u^*$, or $x^* = \psi^{-1}(u^*)$.

Now consider the sequence $\psi^{-1}(v_{b_{f_1}}), \psi^{-1}(v_{b_{f_2}}), \ldots$. Again, by compactness of $\mathbb{S}^1$, it 
has a convergent subsequence $y_1 = \psi^{-1}(v_{b_{z_1}}), y_2 = \psi^{-1}(v_{b_{z_2}}), \ldots$, $z_1 < z_2 < \cdots$. Suppose $y^*$ is the limit of this subsequence. By continuity of $\psi$, $\lim_{j \rightarrow \infty} \psi(y_j) = \psi(y^*)$. Thus, $\lim_{j \rightarrow \infty} \psi(y_j) = \lim_{j \rightarrow \infty} v_{b_{z_j}} = \lim_{j \rightarrow \infty} v_j$ (since we assume that the last limit exists). Let $v^*$ denote the limit $\lim_{j \rightarrow \infty} v_j$. We conclude that $\psi(y^*) = v^*$, or $y^* = \psi^{-1}(v^*)$.

Thus, there exists sequences $\zeta_1 = \psi^{-1}(u_{b_{z_1}}), \zeta_2 = \psi^{-1}(u_{b_{z_2}}), \ldots$
and $\eta_1 = \psi^{-1}(v_{b_{z_1}}), \eta_2 = \psi^{-1}(v_{b_{z_2}}), \ldots$ converging to $\psi^{-1}(u^*)$ and 
$\psi^{-1}(v^*)$, respectively.

Since for every $j \in \mathbb{N}$,

$$|\psi^{-1}(u^*) - \psi^{-1}(v^*)| \leq |\psi^{-1}(u^*) - \zeta_j|  + |\zeta_j - \eta_j| + |\eta_j - \psi^{-1}(v^*)|$$

and $|\zeta_j - \eta_j| = |I_{b_{z_j}}|$, we conclude that $\psi^{-1}(u^*) = \psi^{-1}(v^*)$. Since $\psi$ is a function, we derive that $u^* = v^*$. Hence, the claim is proved.
$\blacksquare$

%% file: horizontal-sweep-001.tex
\subsection{Horizontal sweeps}
\label{sec-horizontal-sweep}

Let $t$ be a closed horizontal line segment in $\mathbb{R}^2$ such that $int(t) \cap J = \phi$. Let $a$ and $b$ be the left and right endpoints of $t$, respectively.
We use $int(t)$ to denote the interior of $t$.

\begin{definition}
{\bf Horizontal sweep.} The horizontal sweep $H(t)$ using the segment $t$ is the union of all open segments $s_p$, where $p \in int(t)$.
\end{definition}

{\bf Upper and lower functions of $H(t)$.} Let $\mathcal{U}_t: t \rightarrow \mathbb{R}^2$ be the function which maps a point $p \in t$ to the upper endpoint $u(s_p)$ of the open segment at $p$. If $p \in J$ (this can happen only when $p$ is equal to $a$ or $b$), then $\mathcal{U}_t(p) = p$.

Similarly, define $\mathcal{L}_t: t \rightarrow \mathbb{R}^2$ as the function which maps $p \in t$ to the lower endpoint $b(s_p)$ of open segment $s_p$. As above, if $p \in J$, then $\mathcal{L}_t(p) = p$.

We call $\mathcal{U}_t$ and $\mathcal{L}_t$ the {\it upper} and {\it lower} functions of the horizontal sweep $H(t)$, respectively.

\begin{observation}
\label{obs-two-limits}
Suppose there exists a point $p \in int(t) \cup \{ b \}$ such that
the left-hand limit does not exist for $\mathcal{U}_t$. Then, there
exists an infinite sequence of points $z_1, z_2, \ldots$ such that

\begin{enumerate}
\item for each $i \in \mathbb{N}$, $z_i$ lies to the left of $p$,
\item for each $i \in \mathbb{N}$, $z_i$ lies to the left of $z_{i+1}$,
\item the sequence $f(z_1), f(z_3), f(z_5), \ldots$ converges to a finite real number $\alpha$,
\item the sequence $f(z_2), f(z_4), f(z_6), \ldots$ converges to a finite real number $\beta$, and
\item $\alpha \neq \beta$.
\end{enumerate}

One can replace left-hand limit by right-hand limit and 
$\mathcal{U}_t$ by $\mathcal{L}_t$ in the above statement.
\end{observation}

\noindent {\bf Proof:} Let $q_1, q_2, \ldots$ be any sequence of
points of $t$ such that (i) each $q_i$ lies to the left of $p$,
(ii) for each $i \in \mathbb{N}$, $q_i$ lies to the left of
$q_{i+1}$ and (iii) $\lim_{i \rightarrow \infty} q_i = p$.

Since the points $u_i=(q_i, \mathcal{U}_t(q_i))$, $i \in \mathbb{N}$, lie inside the
bounding box $\mathcal{B}$ (because they belong to the Jordan
curve $J$), by the Bolzano-Weierstrass theorem, there exists
a convergent subsequence $S_1 = u_{i_1}, u_{i_2}, \ldots$ ($i_1 < i_2 < \cdots$) of this sequence. Suppose this sequence converges to a point $(p, \alpha)$.

Since the left-hand limit does not exist at $p$, in particular,
point $(p, \alpha)$ is not the left-hand limit at $p$.
Thus, there exists a sequence $p_1, p_2, \ldots$ of points such that
(i) each $p_i$ lies to the left of $p$, (ii) for each $i \in \mathbb{N}$, $p_i$ lies to the left of $p_{i+1}$ and (iii) $\lim_{i \rightarrow \infty} p_i = p$ such that the sequence $w_i = (p_i, \mathcal{U}_t(p_i))$, $i \in \mathbb{N}$, does not converge
to point $(p, \alpha)$.

Thus, there exists an $\epsilon > 0$, for which there exists no natural number $N$ such that all points $w_i$, $i > N$, belong to the open ball $B((p, \alpha), \epsilon)$ of radius $\epsilon$, with its center at point $(p, \alpha)$.
Thus, there exists an infinite subsequence $w_{j_1}, w_{j_2}, \ldots$, $j_1 < j_2 < \cdots$, such that $w_{j_k} \notin B((p, \alpha), \epsilon)$ for all $k \in \mathbb{N}$.

Again, by the Bolzano-Weierstrass theorem, this subsequence has a
convergent subsequence $S_2 = w_{j_{k_1}}, w_{j_{k_2}}, \ldots$,
$k_1 < k_2 < \cdots$. Let $(p, \beta)$ be the limit point of 
this subsequence. Then, $(p, \beta) \notin B((p,\alpha), \epsilon)$ and hence $|\beta - \alpha| \geq \epsilon > 0$.

For simplicity of notation, let $S_1 = u_1, u_2, \ldots$ and
$S_2 = w_1, w_2, \ldots$. Set $z_1 = u_1$. Set $z_2$ to be the first
point of $S_2$ which lies to the right of $z_1$. Set $z_3$ to be
the first point of $S_1$ which lies to the right of $z_2$ and so on.

The resulting sequence $z_1, z_2, \ldots$ satisfies the above
conditions.
$\blacksquare$

\begin{lemma}
\label{lemma-left-limit-exists}
At every point $p \in int(t) \cup \{ b \}$, the left-hand limit exists for both functions $\mathcal{U}_t$ and $\mathcal{L}_t$ and is a finite real number.
Further, at every point $p \in \{ a \} \cup int(t)$, the right-hand limit exists for both functions $\mathcal{U}_t$ and $\mathcal{L}_t$ and is a finite real number.
\end{lemma}

\noindent {\bf Proof:} (see Figure \ref{fig-left-limit-exists}) Let $z_1, z_2, \ldots$ be the sequence of points of $t$ given by the above observation.
For each $i \in \mathbb{N}$, let $w_i$ be the point $(z_i, f(z_i))$. Let $\epsilon = \frac{|\beta - \alpha|}{4}$.
By the above observation, $\epsilon$ is a positive real number. Without loss of generality, assume that $\beta > \alpha$.

Since the sequence $z_1, z_3, \ldots$ converges to $\alpha$ and the sequence $z_2, z_4, \ldots$ 
converges to $\beta$, there exists an even natural number $i_0$ such that for all $i \geq \frac{i_0}{2}$, 
$f(z_{2i}) \in (\beta - \epsilon, \beta + \epsilon)$ and $f(z_{2i+1}) \in (\alpha - \epsilon, \alpha + \epsilon)$.

Let $y_1$ be the point with minimum $y$-coordinate among the two points $w_{i_0}$ and $w_{i_0+2}$. 
Let $R_1$ be the rectangle formed by the horizontal line containing $t$, the horizontal line through $y_1$ and
the vertical lines through $w_{i_0}$ and $w_{i_0+2}$. The point $u_1=w_{i_0+1} \in J$ lies in the interior of rectangle
$R_1$. Further, the point $w_{i_0+3}$ lies in the exterior of rectangle $R_1$. Let $\mu_1$ be an arc of $J$ with
endpoints at $w_{i_0+1}$ and $w_{i_0+3}$. Then, by the Jordan curve theorem for rectangle $R_1$ (a rectangle
is the simplest possible rectilinear polygon and we have assumed that Jordan curve theorem is true for 
all rectilinear polygons), arc $\mu_1$ must intersect rectangle $R_1$. Since the left and right boundaries
of $R_1$ are open segments and the bottom boundary is $int(t)$, $\mu_1$ can intersect $R_1$ only on its
top edge (considered as a closed line segment). Let $v_1$ be the first intersection point of arc $\mu_1$ (considered from $u_1$) with the top edge of $R_1$. Finally, let $\gamma_1$ be the portion of $\mu_1$ from $u_1$ to $v_1$. Note that $\gamma_1$ lies completely inside the closed rectangle $R_1$.

Repeating this process along the sequence, we get a sequence of mutually disjoint arcs $\gamma_1, \gamma_2, \ldots$
of $J$, where $u_i$ and $v_i$ are the endpoints of arc $\gamma_i$, for each $i \in \mathbb{N}$.
Further, for each $i \in \mathbb{N}$, the Euclidean distance $d(u_i, v_i)$ is at least $\frac{|\beta-\alpha|}{2}$.

Further, note that $\lim_{i \rightarrow \infty} u_i = (p, \alpha)$ and $\lim_{i \rightarrow \infty} v_i = (p, \beta)$.
Applying Observation \ref{obs-jordan-arc-limit} to the arcs $\gamma_1, \gamma_2, \ldots$, we conclude that $\alpha = \beta$. We arrive at a contradiction, and the lemma is proved.
$\blacksquare$\\

\begin{center}
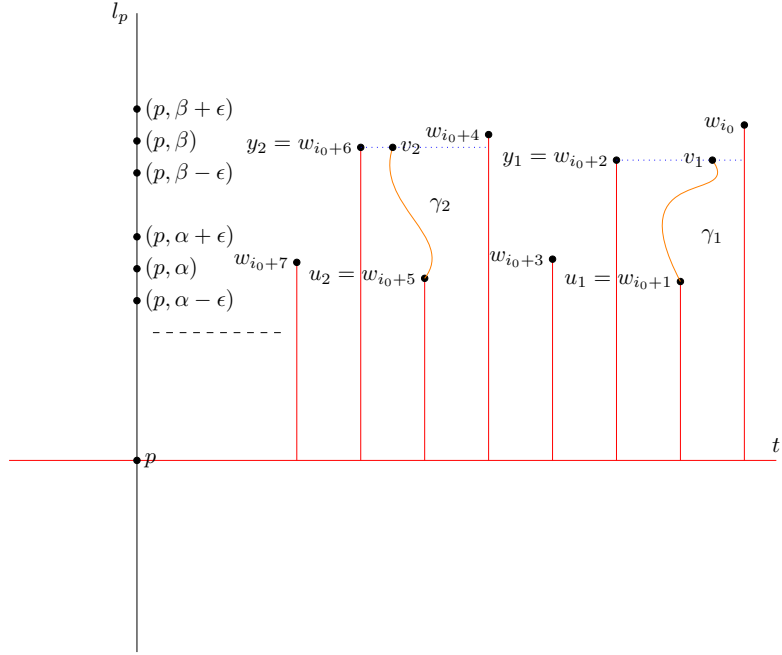
\begin{figure}
\scalebox{0.85}{%
\begin{tikzpicture}
\draw[red] (-2,0) -- (10,0) node[black][anchor=south]{$t$};

\draw[red] (9.5, 0) -- (9.5, 5.25) node[black][anchor=east]{$w_{i_0}$};
\filldraw[black] (9.5, 5.25) circle (0.05) ;
\draw[red] (8.5, 0) -- (8.5, 2.8) node[black][anchor=east]{$u_1=w_{i_0+1}$};
\filldraw[black] (8.5, 2.8) circle (0.05) ;
\draw[red] (7.5, 0) -- (7.5, 4.7) node[black][anchor=east]{$y_1=w_{i_0+2}$};
\filldraw[black] (7.5, 4.7) circle (0.05) ;
\draw[blue,dotted] (7.5, 4.7) -- (9.5, 4.7);
\draw[orange] (8.5, 2.8) .. controls (7.5, 4.7) and (9.5, 4) .. (9, 4.7);
\node[anchor=center] at (9, 3.5) {$\gamma_1$};
\filldraw[black] (9, 4.7) circle (0.05) node[black][anchor=east]{$v_1$};

\draw[red] (6.5, 0) -- (6.5, 3.15) node[black][anchor=east]{$w_{i_0+3}$};
\filldraw[black] (6.5, 3.15) circle (0.05) ;
\draw[red] (5.5, 0) -- (5.5, 5.1) node[black][anchor=east]{$w_{i_0+4}$};
\filldraw[black] (5.5, 5.1) circle (0.05) ;

\draw[red] (4.5, 0) -- (4.5, 2.85) node[black][anchor=east]{$u_2=w_{i_0+5}$};
\filldraw[black] (4.5, 2.85) circle (0.05) ;

\draw[red] (3.5, 0) -- (3.5, 4.9) node[black][anchor=east]{$y_2=w_{i_0+6}$};
\filldraw[black] (3.5, 4.9) circle (0.05) ;
\draw[blue,dotted] (3.5, 4.9) -- (5.5, 4.9);
\draw[orange] (4.5, 2.85) .. controls (5, 3.5) and (3.7, 4) .. (4, 4.9);
\node[anchor=center] at (4.75, 4) {$\gamma_2$};
\filldraw[black] (4, 4.9) circle (0.05) node[black][anchor=west]{$v_2$};

\draw[red] (2.5, 0) -- (2.5, 3.1) node[black][anchor=east]{$w_{i_0+7}$};
\filldraw[black] (2.5, 3.1) circle (0.05) ;
\draw (0,-3) -- (0,7) node[anchor=east]{$l_p$};
\draw[dashed] (2.25, 2) -- (0.25, 2);
\filldraw[black] (0,0) circle (0.05) node[anchor=west]{$p$} ;
\filldraw[black] (0,3) circle (0.05) node[anchor=west]{$(p, \alpha)$} ;
\filldraw[black] (0,3.5) circle (0.05) node[anchor=west]{$(p, \alpha+\epsilon)$} ;
\filldraw[black] (0,2.5) circle (0.05) node[anchor=west]{$(p, \alpha-\epsilon)$} ;
\filldraw[black] (0,5) circle (0.05) node[anchor=west]{$(p, \beta)$} ;
\filldraw[black] (0,4.5) circle (0.05) node[anchor=west]{$(p, \beta-\epsilon)$} ;
\filldraw[black] (0,5.5) circle (0.05) node[anchor=west]{$(p, \beta+\epsilon)$} ;
\end{tikzpicture}
}
\caption{Proof of Lemma \ref{lemma-left-limit-exists}}
\label{fig-left-limit-exists}
\end{figure}
\end{center}

{\it Note.} A function $f: \mathbb{R} \rightarrow \mathbb{R}$ has a {\it simple discontinuity} at point $x \in \mathbb{R}$ if and only if 
(i) the left-hand and right-hand limits exists at $x$ and are both finite real numbers and (ii) $f$ is discontinuous at $x$.

\begin{corollary}
Both functions $\mathcal{U}_t$ and $\mathcal{L}_t$ have a countable number of discontinuities.
\end{corollary}

\noindent {\bf Proof:} Since the left-hand and right-hand limits 
exist at every point $p \in int(t)$ and are finite, every discontinuity of $\mathcal{U}_t(p)$ and $\mathcal{L}_t(p)$ is a simple discontinuity (\cite{stromberg}, Chapter $3$). The number of simple discontinuities of a function from $\mathbb{R}$ to $\mathbb{R}$ is a countable set (\cite{stromberg}, Chapter $3$). $\blacksquare$\\

Though the above lemma establishes that the number of discontinuities
is countable, this does not suffice to prove the Jordan curve theorem
for horizontal sweeps. The set of discontinuities may be countable,
but dense (for example, the set of discontinuities can be the rational numbers).

{\bf An example.} Consider the following well-known example \cite{stromberg}. Let $r_1, r_2, \ldots$ be any enumeration of rational numbers in $[0,1]$. Let $f: [0,1] \rightarrow \mathbb{R}$ be the function defined by $f(x) = \sum_{r_n < x} \frac{1}{2^n}$. The set of discontinuities of $f$ is
the set of rational numbers. The function is monotone and continuous,
with $f(0)=0$ and $f(1) = 1$. The set of discontinuities is 
$[0,1] \cap \mathbb{Q}$, where $\mathbb{Q}$ is the set of rational
numbers. 

We now show that the result of the above lemma is the best possible.
We construct a Jordan curve $W$ from $f$ as follows:

\begin{enumerate}
\item We add three line segments to the graph of $f$: the segment from
$(0,0)$ to $(0,-1)$, the segment from $(0,-1)$ to $(1, -1)$ and the
segment from $(1, -1)$ to $(1, 0)$.

\item Let $r \in \mathbb{Q} \cap [0,1]$. Let $w_r$ be the vertical
segment from $f(r-) = \lim_{z \rightarrow r^-} f(r)$ to $f(r+) = \lim_{z \rightarrow r^+} f(z)$. We connect
the two endpoints of $w_r$ with a semicircle $S_r$, which lies to the
left of segment $w_r$. We apply this modification for every rational
number in $[0,1] \cap \mathbb{Q}$.
\end{enumerate}

This results in a Jordan curve $W$. Let $t$ be the horizontal segment
from $(0,-1)$ to $(1, -1)$. Then, the upper function for the
horizontal sweep $H(t)$ is exactly the function $f$.

We now describe the interior and exterior of $W$. The interior $int(W)$ is
defined as follows:

\begin{enumerate}
\item Let $p \in (0,1) - \mathbb{Q}$ be an irrational number. Then, $f$ is
continuous at point $p$. We add all points in interior of open segment $s_p$
to $int(W)$.

\item Let $p \in (0,1) \cap \mathbb{Q}$ be a rational number. Let $m_p$ denote
the minimum of the left-hand limit $f(p-)$ and the right-hand limit $f(p+)$.
We add all interior points of the vertical line segment with endpoints at $(p, -1)$ and $m_p$ to $int(W)$.
\end{enumerate} 

The exterior $ext(W)$ is defined as the set $\mathbb{R}^2 - ( W \cup int(W))$. The Jordan curve theorem is true for $W$ with $int(W)$ as the bounded, connected open set and $ext(W)$ as the unbounded, connected open set (see Lemma \ref{JCT-horizontal-sweep} for a proof).\\

Let $D(\mathcal{U}_t)$ and $D(\mathcal{L}_t)$ be the sets of all discontinuities of functions $\mathcal{U}_t$ and $\mathcal{L}_t$, respectively.

Suppose $D(\mathcal{U}_t)$ has an infinite number of points.
Let $d_1, d_2, \ldots$ be an enumeration of points in $D(\mathcal{U}_t)$ (such an enumeration is possible for any countable set). 

Construct an interval $t'$ of length $l(t) + 1$, 
where $l(t)$ is the length of horizontal segment $t$.

For each $i \in \mathbb{N}$, if $d_i \in int(t)$, then $\mathcal{U}_t(d_i-)$ and and $\mathcal{U}_t(d_i+)$ denote the left-hand and right-hand limits at point $d_i$. If $d_i$ is the left endpoint of $t$, we set $\mathcal{U}_t(d_i-) = \mathcal{U}_t(d_i)$.
Similarly, if $d_i$ is the right endpoint of $t$, we set 
$\mathcal{U}_t(d_i+) = \mathcal{U}_t(d_i)$. 

For each $k \in \mathbb{N}$, $V_k$ is the vertical
segment formed by the points $(d_k, \mathcal{U}_t(d_k-))$ and
$(d_k, \mathcal{U}_t(d_k+))$. We define $bd(\mathcal{U}_t) = 
\{(z, \mathcal{U}_t(z)) ~| z \in t - D(\mathcal{U}_t) \} 
\cup \{V_k ~| z = d_k, ~k \in \mathbb{N} \}$.

If $z \in t$ is a point where $\mathcal{U}_t$ is continuous, we set
$\zeta(z)$ to the point of $t'$ which is at distance $z + \sum_{d_i < z} \frac{1}{2^i}$ from the left endpoint of $t'$. If $z \in t$ is a discontinuity of $\mathcal{U}_t$ i.e., $z = d_k$ for some $k \in \mathbb{N}$,
$\zeta(z)$ is equal to the closed interval $I_k = [g_k, h_k]$.
Here $g_k$ and $h_k$ are points of $t'$ at distances $z + \sum_{d_i < z} \frac{1}{2^i}$ and $z + \sum_{d_i < z} \frac{1}{2^i} + \frac{1}{2^k}$ from 
the left endpoint of $t'$, respectively.

We now define a map $\eta: t' \rightarrow bd(\mathcal{U}_t)$ as follows:

\begin{enumerate}
\item Let $w \in t'$. If $w = \zeta(z)$, where $z$ is a point of continuity of $\mathcal{U}_t$, we set $\eta(w) = (z, \mathcal{U}_t(z))$.

\item Let $w \in t'$. Suppose $w \in I_k$, for some $k \in \mathbb{N}$. Suppose $w$ partitions $I_k$ into two segments $[c_k, w)$ and $(w, d_k]$ whose lengths are in the ratio $\delta: 1 - \delta$, for some $0 \leq \delta \leq 1$. Then, $\eta(w)$ partitions
the vertical segment $V_k$ of $bd(\mathcal{U}_t)$ in the same
ratio $\delta : 1 - \delta$ (we take the left-hand limit as the left endpoint of $V_k$ and the right-hand limit as the right endpoint of $V_k$).
\end{enumerate}

\begin{observation}
\label{obs-bijection-eta}
$\eta$ is a continuous bijection from $t'$ to $bd(\mathcal{U}_t)$.
\end{observation}

\noindent {\bf Proof:} The fact that $\eta$ is a bijection follows from definition. We now prove
that it is a continuous function.

Consider any point $w \in t'$. Let $\eta(w) = (z, \mathcal{U}_t(z))$. There are three cases:

\begin{enumerate}
\item {\it $z$ is a point of continuity of $\mathcal{U}_t$.} Suppose $w$ is not the left endpoint
of $t'$. Suppose $w'$ tends to $w$ from the left, in segment $t'$. Let $X(\eta(w'))$ denote the
$x$-coordinate of $w'$. By definition of $\eta$, $X(\eta(w')) < X(\eta(w))$. Further, if $w_1 < w_2 < w$,
then $X(\eta(w_1)) \leq X(\eta(w_2)) \leq X(\eta(w))$.

We first show that $X(\eta(w'))$ tends to $X(\eta(w))$ from the left. To prove this, note
that $w - w' = X(\eta(w)) - X(\eta(w')) + \sum_{d_n \in [X(\eta(w')), X(\eta(w))]} \frac{1}{2^n}$.

Thus, $\lim_{w' \rightarrow w-} \eta(w') = \lim_{X(\eta(w')) \rightarrow X(\eta(w))-} \eta(w')
= \lim_{x \rightarrow z=X(\eta(w))-} (x, \mathcal{U}_t(x)) = (z, \mathcal{U}_t(z-)) = (z, \mathcal{U}_t(z))$
(since $z$ is a point of continuity of $\mathcal{U}_t$).

A symmetric argument proves that $\lim_{w' \rightarrow w+} \eta(w') = (z, \mathcal{U}_t(z))$.
Thus, $\eta$ is continuous at $w$.

\item {\it $z$ is a point of discontinuity of $\mathcal{U}_t$.} Suppose $z = d_k$ for some $k \in \mathbb{N}$.
Then, $\eta(w)$ is a point of the vertical $V_k$. If $\eta(w)$ is an interior point of $V_k$, then
continuity at $w$ is the result of the continuity of the linear map from $[c_k, d_k]$ to $V_k$, at any
interior point. 

Now suppose $\eta(w)$ is an endpoint of $V_k$. Thus, $w \in \{c_k, d_k\}$. Without loss of generality, assume that $w = c_k$ i.e., $\eta(w)$ is the left endpoint of $V_k$. Then, $\lim_{w' \rightarrow w+} \eta(w)$ exists and is equal to $\eta(w)$, by the continuity
from the right at point $c_k$, of the linear map from $[c_k, d_k]$ to $V_k$. Further, $\lim_{w \rightarrow w-} \eta(w) = (z, \mathcal{U}_t(z-)) = \eta(c_k) = \eta(w)$ (using the argument for the case when
$X(\eta(w))$ is a point of continuity above).
Thus, in this case also, $\eta$ is continuous at $w$.
\end{enumerate}

$\blacksquare$\\

In a symmetric manner, we construct a continuous bijection $\eta'$ from a segment $t''$ of length $l(t)+1$ and $bd(\mathcal{L}_t)$:

\begin{observation}
\label{obs-bijection-eta-prime}
$\eta'$ is a continuous bijection from $t''$ to $bd(\mathcal{L}_t)$.
\end{observation}

\noindent {\bf Proof:} Similar to Observation \ref{obs-bijection-eta}.
$\blacksquare$

Let $s_{a+}$ be the right-hand limit $\mathcal{U}_t(a+)$ and $s_{b-}$ be the left-hand 
limit $\mathcal{U}_t(b-)$. We can combine $s_{a+}$, a dilation of $t'$, $s_{b-}$ and a dilation of $t''$ to obtain a rectangle $R_t$, which is homeomorphic to a circle.

Let $\psi_t: R \rightarrow s_{a+} \cup bd(\mathcal{U}_t) \cup s_{b-} \cup bd(\mathcal{L}_t)$ be the map which is the identity map on $s_{a+}$,
$\eta$ on $t'$, $\eta'$ on $t''$ and the identity map on $s_{b-}$. We arrive at the following observation:

\begin{observation}
\label{obs-K-t-Jordan-curve}
$\psi_t$ is a continuous and one-to-one function from $R$ to $\mathbb{R}^2$.
\end{observation}

\noindent {\bf Proof:} $bd(\mathcal{U}_t)$ and $bd(\mathcal{L}_t)$ are both $x$-monotone Jordan arcs. $\blacksquare$\\

We denote the Jordan curve given by $\psi_t$ by $K_t$. We define
$int(H(t)) = \cup_{p \in t} s_p - K_t$ and $ext(H(t)) = \mathbb{R}^2 - \big( K_t \cup int(H(t))  \big)$.

\begin{lemma}
\label{JCT-horizontal-sweep}
{\bf JCT for horizontal sweep.} Let $H(t)$ be a horizontal sweep. Then, there exists a Jordan curve $K_t$ such that

\begin{enumerate}
\item the Jordan curve theorem is true for $K_t$,
\item $int(H(t))$ is the bounded, connected open set, out of the two regions in which $K_t$ partitions $\mathbb{R}^2 - K_t$, 
\item $ext(H(t))$ is the unbounded, connected open set, out of the
two regions in which $K_t$ partitions $\mathbb{R}^2 - K_t$, and
\item $K_t$ is the disjoint union of (a) points of $J$ (which can be uncountable) and (b) a countable number of vertical line segments (corresponding to discontinuities of $\mathcal{U}_t$ and $\mathcal{L}_t$) and (c) the two segments $s_{a+}$ and $s_{b-}$.
\end{enumerate}
\end{lemma}

\noindent {\bf Proof:} The fact that $K_t$ is a Jordan curve follows from Observation \ref{obs-K-t-Jordan-curve}.

We now prove that $int(K_t)$ is an open, connected and bounded set. 
We have assumed that all open segments reached from our initial
point have finite length (see Assumption $1$ above). Thus, both upper and lower functions
($\mathcal{U}_t$ and $\mathcal{L}_t$) lie in the bounding box $\mathcal{B}$ of the Jordan curve $J$, and hence the set $int(K_t)$ is bounded.
For any point $p \in int(K_t)$, the vertical line segment $v_p$ from $p$ to the horizontal line segment $t'$ lies in $int(K_t)$. Given two distinct points $p_1, p_2 \in int(K_t)$, we can go from $p_1$ to $p_2$ by first moving on line segment $v_{p_1}$ to the horizontal segment $t'$, then moving towards the $x$-coordinate of $p_2$ on the segment $t'$, and finally moving on line segment $v_{p_2}$ to reach the destination. This rectilinear path with two turns lies completely in $int(K_t)$. Hence, $int(K_t)$ is connected.

Consider any point $w \in int(K_t)$. Without loss of generality, assume that
$w$ lies above $t$. Let $w'$ be the perpendicular projection of $w$ on the segment $t$. By definition of $int(K_t)$, $w' \in int(t)$. There are two cases: (i) $w'$ is a point of continuity of $\mathcal{U}_t$, or (ii) $w'$ is a point of discontinuity
of $\mathcal{U}_t$. 

First consider case (i). Let $\epsilon_1 > 0$ be the distance of $w$ from the set $\{w', (w',\mathcal{U}_t(w'))\}$. Let $\epsilon_2 = \min( |w'-a|, |b-w'|)$, where
$a$ and $b$ are the left and right endpoints of $t$. 
Choose $\epsilon_3 = \frac{\min(\epsilon_1, \epsilon_2)}{2}$. Since $\mathcal{U}_t$ is continuous at $w'$, there exists a $\delta > 0$ such that for all $|x-w'| < \delta$ ($x \in int(t)$), $|\mathcal{U}_t(x) - \mathcal{U}_t(w')| < \frac{\epsilon_3}{4}$. Choose $\epsilon = \min(\frac{\epsilon_3}{8}, \frac{\delta}{2})$. Then, the open ball $B(w, \epsilon)$ lies completely in $int(K_t)$.

Now let us consider case (ii). Let $\epsilon_1 > 0$ be the distance of $w$ from the set $\{w', (w',m_{w'})\}$, where $m_{w'}$ is the minimum of $\mathcal{U}_t(w'-)$ and $\mathcal{U}_t(w'+)$). Let $\epsilon_2 = \min( |w'-a|, |b-w'|)$. 
Choose $\epsilon_3 = \frac{\min(\epsilon_1, \epsilon_2)}{2}$. Since the left-hand limit $\mathcal{U}_t(w'-)$ and the right-hand
limit $\mathcal{U}_t(w'+)$ exist at $w'$ and $m_{w'} = \min( \mathcal{U}_t(w'-), \mathcal{U}_t(w'+) )$, there exists a $\delta > 0$ such that for all $|x-w'| < \delta$ ($x \in int(t)$), $\mathcal{U}_t(x) - m_{w'} > -\frac{\epsilon_3}{4}$. Choose $\epsilon = \min(\frac{\epsilon_3}{8}, \frac{\delta}{2})$. Then, the open ball $B(w, \epsilon)$ lies completely in $int(K_t)$.

Since the point $w$ was arbitrary, we conclude that $int(K_t)$ is an open set.

We now prove that $ext(K_t)$ is an unbounded, open and connected set. Let $w$ be any point in $ext(K_t)$. If $l_w \cap K_t$ is empty, then the open segment $s_w$ at $w$ is equal to line $l_w$. If $l_w$ intersects the horizontal segment $t$, then $s_w$ is a half-line. An argument similar to that for $int(K_t)$ proves that $ext(K_t)$ is an open and connected set.

We finally show that $cl(int(K_t)) = int(K_t) \cup K_t$. First, note that $\mathbb{R}^2$ is the disjoint union of $int(K_t)$, $K_t$ and $ext(K_t)$. 

We now prove that for any point $p \in K_t$, every open ball with center at $p$ contains a point of $int(K_t)$ and
a point of $ext(K_t)$. First assume that $p$ does not belong to $s_{a+} \cup s_{b-}$. Let $p'$ be the perpendicular projection of $p$ on $t$. Note that $p' \in int(t)$.
There are two cases: (i) $p'$ is a point of continuity of both $\mathcal{U}_t$ and
$\mathcal{L}_t$ and (ii) $p' \in D(\mathcal{U}_t) \cup D(\mathcal{L}_t)$.

In case (i), $l_p \cap W$ consists of only two points: $(p', \mathcal{U}_t(p'))$ and $(p', \mathcal{L}_t(p'))$. Thus, any open ball with
center at $p$ contains a point of $int(K_t)$ just below $p$ and a point of $ext(K_t)$ just above $p$.

Now consider case (ii). Without loss of generality, assume that $p$ lies above
$t$ and $p' \in D(\mathcal{U}_t)$. Then, $p$ is a point on the vertical segment from $(p', \mathcal{U}_t(p'-))$ and $(p', \mathcal{U}_t(p'+))$.
Without loss of generality, assume that $\mathcal{U}_t(p'-) < \mathcal{U}_t(p'+)$. 

Take any $\epsilon > 0$. Then, since the right-hand
limit exists at point $p'$, there exists a $\delta > 0$ such that for all $p' < x < p'+\delta < b$, $|\mathcal{U}_t(x) - \mathcal{U}_t(p'+)| < \frac{\epsilon}{2}$. Take any point $r$ of continuity of $\mathcal{U}_t$ in the open interval $(x, x+\min(\frac{\epsilon}{2}, \delta))$. Then, the interior of the open segment $s_r$ contains a point of $B(p, \epsilon)$. Thus, $B(p,\epsilon) \cap int(K_t) \neq \phi$.
Applying the same argument to the left-hand limit, we conclude that $B(p, \epsilon) \cap ext(K_t) \neq \phi$.

Thus, $K_t \subset bd(int(K_t))$ and $K_t \subset bd(ext(K_t))$. Any point $p \notin K_t$ belongs to one of the two sets $int(K_t)$ or $ext(K_t)$. Without loss of generality assume that $p \in int(K_t)$. There is an open ball $B(p, \epsilon) \subset int(K_t)$ with center at $p$, for some $\epsilon > 0$. Thus, $p \notin K_t \cup ext(K_t)$. Hence, $K_t = bd(int(K_t)) = bd(ext(K_t))$.

This establishes the Jordan curve theorem for $K_t$.
$\blacksquare$\\

We now prove a property of vertical segments of $K_t$, which we will use in the next section:

\begin{lemma}
\label{lem-vertical-segment-crossing}
Let $w$ be a vertical segment of $K_t$ such that the vertical
line containing $w$ intersects $int(t)$ at a point $q$. Let $a'$ and $b'$ be the upper and lower endpoints of $w$ respectively. Without loss of generality, assume that $a' = \mathcal{U}_t(q+)$ 
and $b' = \mathcal{U}_t(q-)$ (other cases are symmetric). 

Then, for any point $w' \in cl(w) - \{ a' \}$, there exists an $\epsilon > 0$ ($\epsilon$ depends on $w$) such that every
point in the set $\{ \psi(x) ~| ~x \in \mathbb{S}^1 ~and ~\psi^{-1}(w') - \epsilon \leq x \leq \psi^{-1}(w') + \epsilon\}$ lies in 
the closed half-plane to the left of $l_q$.
\end{lemma}

\noindent {\bf Proof:} Suppose, for the sake of contradiction, there is no such $\epsilon > 0$.
For any $\epsilon > 0$, let $U_{w', \epsilon}$ denote the set $\{ \psi(x) ~| ~x \in \mathbb{S}^1 ~and ~\psi^{-1}(w') - \epsilon \leq x \leq \psi^{-1}(w') + \epsilon\}$.
Take $\epsilon_1 = \frac{1}{2}$. Then, there exists $x_1 \in U_{w', \epsilon_1}$ such that $\psi(x_1)$ lies strictly to the right of $l_q$.
Take $\epsilon_2 = \frac{|x_1 - \psi^{-1}(w')|}{2}$. Again, there exists $x_2 \in U_{w', \epsilon_1}$ such that $\psi(x_2)$ lies strictly to the right of $l_q$.
Repeating this construction, we get a sequence of points $x_1, x_2, \ldots$ in $\mathbb{S}^1$ such that $\lim_{i \rightarrow \infty} |x_i - \psi^{-1}(w')| = 0$
and for each $i \in \mathbb{N}$, $\psi(x_i)$ lies strictly to the right of $l_q$.

Since $\psi$ is a continuous function, we conclude that $\lim_{i \rightarrow \infty} |\psi(x_i) - w'| = 0$. Since for each $i \in \mathbb{N}$, 
$\psi(x_i)$ lies strictly to the right of $l_q$, we can find a subsequence $x_{j_1}, x_{j_2}, \ldots$ ($j_1 < j_2 < \ldots$) of this sequence
such that for each $k \in \mathbb{N}$, $x$-coordinate of $\psi(x_{j_{k+1}})$ is strictly less than the $x$-coordinate of $\psi(x_{j_{k}})$.
This subsequence can be constructed as follows: take $x_{j_1} = x_1$. Suppose we have already constructed the first $k$ terms of this sequence.
Since $\psi(x_1), \psi(x_2), \ldots$ converges to $w'$, there exist infinitely many points of this sequence which lie to the left of the
vertical line through $x_{j_k}$. We pick any such point as $x_{j_{k+1}}$ ($x_{j_{k+1}}$ must be different from all points picked till now) and repeat.

Let $q_{1}, q_{2}, \ldots$ be the points of $t$ with the same $x$-coordinate as $\psi(x_{j_1}), \psi(x_{j_2}), \ldots$, respectively.
Thus, for each $k \in \mathbb{N}$, the upper endpoint of $s_{q_k}$ has $y$-coordinate less than or equal to the $y$-coordinate of $\psi(x_{j_k})$. 
Let $m$ be the midpoint of line segment joining $w'$ and $a'$. Then, there exists a natural number $N$, such that for all $k > N$, the upper endpoint of $s_{q_k}$ has $y$-coordinate strictly less than $m$. Since $s_{q_k}$ are finite segments and hence lie inside the bounding box $\mathcal{B}$, there exists a convergent
subsequence of $s_{q_1}, s_{q_2}, \ldots$ which converges to an open segment $s^*$ whose upper endpoint does not have $y$-coordinate higher than $m$.
Thus, $\mathcal{U}_t(q+)$ has value less than or equal to the $y$-coordinate of $m$. This contradicts our assumption that $\mathcal{U}_t(q+) = a'$.
$\blacksquare$\\

The structure of $K_t$ is central to our approach - all sweepline algorithms
considered in this paper sweep a region whose boundary is a Jordan curve
with similar structure to $K_t$ (the boundary for a single horizontal sweep).
Before proceeding further, we define a subclass of Jordan curves, which 
we call {\it piecewise-vertical Jordan curves}. All sweepline algorithms considered by us sweep a region whose boundary is a piecewise-vertical Jordan curve.

\begin{definition}
{\bf Piecewise-vertical Jordan curves.} A piecewise-vertical Jordan curve $K$, with respect to a given Jordan curve $J$, is a Jordan curve which is the disjoint union of (a) a subset $\mathcal{P}(K)$ of points of $J$ and (b) a countable set $\mathcal{V}(K)$ of vertical line segments. Further, for each
vertical segment $s \in \mathcal{V}$, there exists a horizontal line segment $u_s$, $int(u_s) \subset \mathbb{R}^2 - J$ such that
$s$ is a vertical segment of the boundary $K_{u_s}$ of the horizontal sweep $H(u_s)$.
\end{definition}

{\it Note.} A piecewise-vertical Jordan curve $K$ can contain vertical
segments belonging to the Jordan curve $J$. In the following, whenever we
refer to the vertical segment of such a curve, we mean a vertical segment in
$\mathcal{V}(K)$.

We now define an equivalence relation $\sim$ on the set of all horizontal sweeps.

\begin{definition}
{\bf Relation on horizontal sweeps.} Let $H(t_1)$ and $H(t_2)$ be two horizontal sweeps. Then, $H(t_1) \sim H(t_2)$ if and only if they have the same set of open segments. In other words, for every open segment $s_p$, $p \in t_1$, there exists a point $q \in t_2$ such that $s_q = s_p$ and vice versa.
\end{definition}

One can verify that $\sim$ is an equivalence relation on the set of all horizontal sweeps.
Let $\mathcal{H}$ be the set of all equivalence classes of horizontal sweeps, according to
the above relation $\sim$.

%% file: extension-horizontal-sweep-001.tex
\subsection{Extension of a horizontal sweep}
\label{sec-extension-horizontal-sweep}

We next define the extension of a horizontal sweep, the building block of
the sweepline algorithm for Jordan curves.

\begin{definition}
{\bf Extension of a horizontal sweep.} Let $H(t)$ be a horizontal sweep. Let $K_t$ be the piecewise-vertical Jordan curve separating $int(H(t))$ and $ext(H(t))$. Let $u_1$ be any vertical segment of $K_t$. Let $p$ be any point of $int(u_1)$. Let $t'$ be a horizontal segment such that (i) one endpoint of $t'$ is $p$ and (ii) $t' - \{p\} \subset ext(K_t)$. 
Then, we say that the horizontal sweep $H(t)$ can be extended by the horizontal sweep $H(t')$.
\end{definition}

Let $E=(H(t), H(t')$ denote the extension of $H(t)$ by $H(t')$.
Without loss of generality, assume that $u_1$ lies above line segment $t$ and $p$ is the left endpoint of $t'$. Further, let $q \in t$ be the point such that $u_1 \subset l_q$. Let $\mathcal{U}_t, \mathcal{U}_{t'}$ denote the upper functions of $H(t)$ and
$H(t')$, respectively. Let $\mathcal{L}_t, \mathcal{L}_{t'}$ denote the lower functions of $H(t)$ and $H(t')$, respectively.

For a point $z \in \mathbb{R}^2 - J$, let $h_z$ be the
horizontal line through point $z$. 
Let $s_{z+}$ be the limiting segment $\lim_{z' \rightarrow z+, ~z' \in h_z} s_{z'}$ and $s_{z-}$ be the limiting segment $\lim_{z' \rightarrow z-, ~z' \in h_z} s_{z'}$.

Let $a_1$ and $b_1$ be the upper and lower endpoints of $u_1$, respectively. We now prove that the upper endpoint of $u_1$ is the upper endpoint of $s_{q-}$ (i.e., $a_1=\mathcal{U}_t(q-)$) and the lower endpoint of $u_1$ is the upper endpoint of $s_{q+}$ (i.e., $b_1=\mathcal{U}_t(q+)$). Suppose the converse is true i.e.,
$a_1 = \mathcal{U}_t(q+)$ and $b_1 = \mathcal{U}_t(q-)$. Then,
for any $\epsilon > 0$, there exists a point $z \in t$, arbitrarily close to $q$ and to the right of $z$, such that $z$ is a point of continuity of $\mathcal{U}_t$ and $|\mathcal{U}_t(z) - \mathcal{U}_t(q)| < \epsilon$. For sufficiently small $\epsilon > 0$, $s_{z}$ intersects $int(t')$ at a point $y$. Thus, $y \in int(H(t))$, a contradiction.

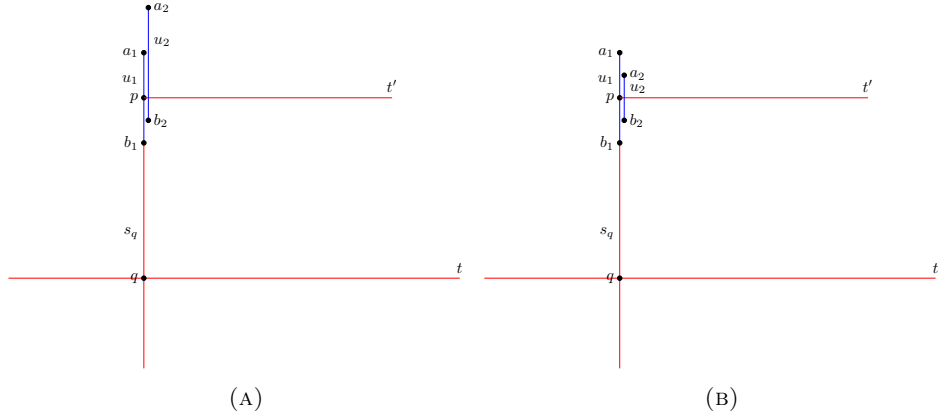
\begin{figure}
\begin{subfigure}{0.5\textwidth}
%    \centering
\scalebox{0.6}{%
\begin{tikzpicture}
    \draw[red] (-3, 0) -- (7, 0) node[black][anchor=south]{$t$};
    \draw[red] (0, -2) -- (0, 3) node[black][pos=0.6][anchor=east]{$s_q$};
    \draw[blue] (0, 3) -- (0, 5) node[black][pos=0.7][anchor=east]{$u_1$};
    \filldraw[black] (0, 5) circle (0.05) node[black][anchor=east]{$a_1$};
    \filldraw[black] (0, 3) circle (0.05) node [black][anchor=east]{$b_1$};
    \filldraw[black] (0, 0) circle (0.05) node [black][anchor=east]{$q$};
    \filldraw[black] (0, 4) circle (0.05) node[black][anchor=east]{$p$};
    \draw[red] (0, 4) -- (5.5,4) node[black][anchor=south]{$t'$};
    \draw[blue] (0.1, 3.5) -- (0.1, 6) node[black][pos=0.7][anchor=west]{$u_2$};
    \filldraw[black] (0.1, 3.5) circle (0.05) node[black][anchor=west]{$b_2$};
    \filldraw[black] (0.1, 6) circle (0.05) node[black][anchor=west]{$a_2$};
\end{tikzpicture}
}
\caption{}
\end{subfigure}%
\begin{subfigure}{0.5\textwidth}
 %   \centering
\scalebox{0.6}{%
\begin{tikzpicture}
    \draw[red] (-3, 0) -- (7, 0) node[black][anchor=south]{$t$};
    \draw[red] (0, -2) -- (0, 3) node[black][pos=0.6][anchor=east]{$s_q$};
    \draw[blue] (0, 3) -- (0, 5) node[black][pos=0.7][anchor=east]{$u_1$};
    \filldraw[black] (0, 5) circle (0.05) node[black][anchor=east]{$a_1$};
    \filldraw[black] (0, 3) circle (0.05) node [black][anchor=east]{$b_1$};
    \filldraw[black] (0, 0) circle (0.05) node [black][anchor=east]{$q$};
    \filldraw[black] (0, 4) circle (0.05) node[black][anchor=east]{$p$};
    \draw[red] (0, 4) -- (5.5,4) node[black][anchor=south]{$t'$};
    \draw[blue] (0.1, 3.5) -- (0.1, 4.5) node[black][pos=0.7][anchor=west]{$u_2$};
    \filldraw[black] (0.1, 3.5) circle (0.05) node[black][anchor=west]{$b_2$};
    \filldraw[black] (0.1, 4.5) circle (0.05) node[black][anchor=west]{$a_2$};
\end{tikzpicture}
}
\caption{}
\end{subfigure}
\caption{Constructing $K_{t,t'}$. (a) The case where $a_1 \leq a_2 < b_1 < b_2$, $s_{p+} = \overline{a_2b_2}$, $s_{p-} = s_{q+}$ and $s_p = \overline{a_1b_2}$. (b) The case where $a_1 \leq a_2 \leq b_2 \leq b_1$, $s_{p-} = s_{q+}$ and $s_p = s_{p+} = \overline{a_2b_2}$. }
\label{fig-constructing-K-t-t'}
\end{figure}

Consider the left vertical boundary $s_{p+}$ of $H(t')$. Let $a_2$
and $b_2$ be the upper and lower endpoints of $s_{p+}$.

We next prove that $b_1 \leq b_2 < a_2$. Suppose, for the sake of
contradiction, that $b_2 < b_1$. Since $b_1 = \mathcal{U}_t(q+)$, 
there exists a point $z \in t$, arbitrarily close to $q$ and
to the right of $z$, such that $z$ is a point of continuity of $\mathcal{U}_t$ and $|\mathcal{U}_t(z) - \mathcal{U}_t(q)| < \epsilon$. For sufficiently small $\epsilon > 0$, $b_2 \in int(s_{z})$. We arrive at a contradiction, since $b_2 \in J$.

Note that the open segment $s_p$ at point $p$ is a subset of $u_2=s_{p+}$. Further, $s_p$ is also a subset of $u_1$, since 
$a_1, b_1 \in J$.

We next show that $s_{p-} = s_{q-}$. If we take any point $z \in t$,
arbitrarily close to $q$ and to the left of $q$, then the line $l_z$ intersects $h_p$ at a point $z'$. The closed line segment 
$\overline{zz'}$ has no point of $J$, and hence $s_{z} = s_{z'}$.
Since $z'$ can be arbitrarily close to $q$, the result follows.

Further, since $s_{p-} \cap J = \phi$ and $s_{p+} \cap J = \phi$
(by the definition of open segment), there is no point of $J$ in the intersection $I(t,t') = s_{q-} \cap s_{p+} = s_{p-} \cap s_{p+}$.

We further note that $s_{p-} \cap s_{p+} = s_p$ and hence $I(t,t')=s_p$.

{\bf Construction of $K_{t, t'}$.} We remove $s_p$ from $K_{t}$ and replace it
with $K_{t'} - s_p$, to obtain a new curve $K_{t,t'}$ (see Figure \ref{fig-constructing-K-t-t'}). We define
$int(E) = int(H(t)) \cup int(H(t')) \cup \{ s_p \}$ and $ext(E) = \mathbb{R}^2 - \{ K_{t,t'} \cup int(E) \}$.

We now formalize the construction of $K_{t,t'}$ in terms of maps. 
Let $\phi_{t}: \mathbb{S}^1 \rightarrow \mathbb{R}^2$ and $\phi_{t'}: \mathbb{S}^1 \rightarrow \mathbb{R}^2$ be continuous and one-to-one maps such that their
images are the Jordan curves $K_{t}$ and $K_{t'}$, respectively.

Let $I_{p,t} = \{ \phi_t^{-1}(z) ~| ~z \in cl(s_p) \}$. $I_{p,t}$ is a closed interval
of $\mathbb{S}^1$, with length strictly less than $2 \pi$.

Let $I_{p,t'} = \{ \phi_{t'}^{-1}(z) ~| ~z \in cl(s_p) \}$. $I_{p,t'}$ is a closed interval of $\mathbb{S}^1$, with length strictly less than $2 \pi$.

Let $a'$ and $b'$ be the lower and upper endpoints of $cl(s_p)$.
Let $\omega: I_{p,t} \rightarrow I_{p,t'}$ be a continuous bijection between these two closed intervals, such that $\omega(\phi_t^{-1}(a')) = \phi_{t'}^{-1}(a')$ and
$\omega(\phi_t^{-1}(b')) = \phi_{t'}^{-1}(b')$.

Define $\phi_{t,t'}: \mathbb{S}^1 \rightarrow \mathbb{R}^2$ as follows:

\begin{enumerate}
\item for $z \in I_{p,t}$, $\phi_{t,t'}(z) = \phi_{t'}(\omega(z))$, and
\item for $z \notin I_{p_t}$, $\phi_{t,t'}(z) = \phi_t(z)$.
\end{enumerate}

\begin{observation}
\label{obs-K-t-t'}
$K_{t,t'}$ is a Jordan curve.
\end{observation}

\begin{figure}
\begin{subfigure}{0.5\textwidth}
\scalebox{0.6}{%
\begin{tikzpicture}
    \draw[red] (-3, 0) -- (7, 0) node[black][anchor=south]{$t$};
    \draw[red] (0, -2) -- (0, 3) node[black][pos=0.6][anchor=east]{$s_q$};
    \draw[blue] (0, 3) -- (0, 5) node[black][pos=0.7][anchor=east]{$u_1$};
    \filldraw[black] (0, 5) circle (0.05) node[black][anchor=east]{$a_1$};
    \filldraw[black] (0, 3) circle (0.05) node [black][anchor=east]{$b_1$};
    \filldraw[black] (0, 0) circle (0.05) node [black][anchor=east]{$q$};
    \filldraw[black] (0, 4) circle (0.05) node[black][anchor=east]{$p$};
    \draw[red] (0, 4) -- (5.5,4) node[black][anchor=south]{$t'$};
    \filldraw[black] (3, 0) circle (0.05) node[black][above, left]{$x_1$};
    \filldraw[black] (3, 4) circle (0.05) node[black][above, left]{$x_2$};
    \draw[red] (3, 5) -- (3,3) node[black][pos=0.3][left]{$s_{x_2}$};
    \filldraw[black] (3, 5) circle (0.05) node[black][above, left]{$x''$};
    \draw[red] (3, -1) -- (3,1) node[black][pos=0.2][left]{$s_{x_1}$};
    \draw[blue] (3, 3) -- (3, 2) node[black][pos=0.7][anchor=east]{$u'_1$};
    \draw[blue] (3, 1) -- (3, 2) node[black][pos=0.7][anchor=east]{$u'_2$};
    \filldraw[black] (3, 3) circle (0.05) node[black][above, right]{$a'_1$};
    \filldraw[black] (3, 1) circle (0.05) node[black][above, right]{$a'_2$};
    \filldraw[black] (3, 2) circle (0.05) node[black][above, right]{$y$};
    \filldraw[black] (2.5, 1) circle (0.05) node[black][above, left]{$x'$};
    
\end{tikzpicture}
}
\caption{{\bf Case I}}
\end{subfigure}
\begin{subfigure}{0.5\textwidth}
\scalebox{0.6}{%
\begin{tikzpicture}
    \draw[red] (-3, 0) -- (7, 0) node[black][anchor=south]{$t$};
    \draw[red] (0, -2) -- (0, 3) node[black][pos=0.6][anchor=east]{$s_q$};
    \draw[blue] (0, 3) -- (0, 5) node[black][pos=0.7][anchor=east]{$u_1$};
    \filldraw[black] (0, 5) circle (0.05) node[black][anchor=east]{$a_1$};
    \filldraw[black] (0, 3) circle (0.05) node [black][anchor=east]{$b_1$};
    \filldraw[black] (0, 0) circle (0.05) node [black][anchor=east]{$q$};
    \filldraw[black] (0, 4) circle (0.05) node[black][anchor=east]{$p$};
    \draw[red] (0, 4) -- (5.5,4) node[black][anchor=south]{$t'$};
    \filldraw[black] (3, 0) circle (0.05) node[black][above, left]{$x_1$};
    \filldraw[black] (3, 4) circle (0.05) node[black][above, left]{$x_2$};
    \draw[red] (3, 5) -- (3,3) node[black][pos=0.3][left]{$s_{x_2}$};
    \filldraw[black] (3, 5) circle (0.05) node[black][above, left]{$x''$};
    \draw[red] (3, -1) -- (3,1) node[black][pos=0.2][left]{$s_{x_1}$};
    \draw[blue] (3, 3) -- (3, 1.5) node[black][pos=0.3][anchor=east]{$u'_1$};
    \draw[blue] (3, 1) -- (3, 2.5) node[black][pos=0.3][anchor=east]{$u'_2$};
    \filldraw[black] (3, 3) circle (0.05) node[black][above, right]{$a'_1$};
    \filldraw[black] (3, 1) circle (0.05) node[black][above, right]{$a'_2$};
    \filldraw[black] (3, 1.5) circle (0.05) node[black][above, right]{$b'_1$};
    \filldraw[black] (3, 2.5) circle (0.05) node[black][above, right]{$b'_2$};
    \filldraw[black] (3, 2) circle (0.05) node[black][above, right]{$y$};
    \filldraw[black] (2.5, 1) circle (0.05) node[black][above, left]{$x'$};
    
\end{tikzpicture}
}
\caption{{\bf Case II}}
\end{subfigure}
\caption{$\phi_{t,t'}$ is one-to-one.}
\label{fig-proof-one-to-one}
\end{figure}
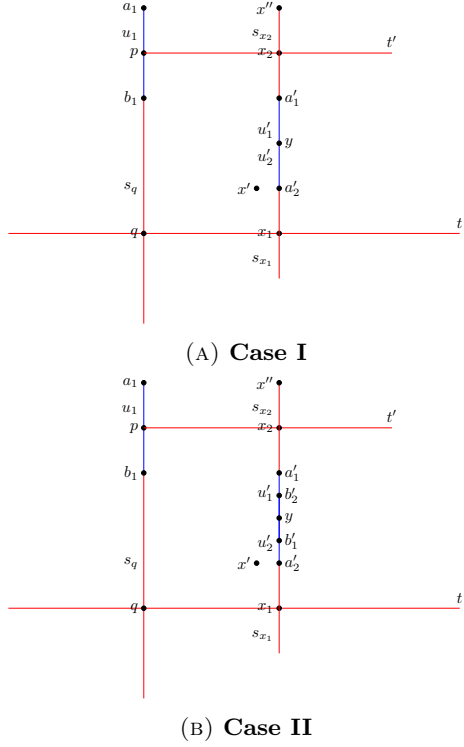

\noindent {\bf Proof:} We first prove that $\phi_{t,t'}$ is a continuous function. Let $y_1 = \phi_t^{-1}(a')$ and $y_2 = \phi_t^{-1}(b')$ be the two endpoints of $I_{p,t}$. Since
$\phi_t$, $\phi_{t'}$ and $\omega$ are continuous functions, we have to prove continuity of $\phi_{t,t'}$ at only these two points. Note that
$\lim_{z \rightarrow y_1, ~z \in I_{p,t}} \phi_{t,t'}(z) = \lim_{z \rightarrow y_1, ~z \in I_{p,t}} \phi_{t'}(\omega(z)) = \phi_{t'}(\omega(y_1)) = a'$, by continuity
of $\phi_{t'}$ and $\omega$. By a similar argument, $\lim_{z \rightarrow y_1, ~z \notin I_{p,t}} \phi_{t,t'}(z) = \phi_t(y_1) = a'$. Thus, $\phi_{t,t'}$ is continuous at $y_1$ and by a similar argument, at $y_2$. This concludes the proof of
continuity $\phi_{t,t'}$.

We now prove that $\phi_{t,t'}$ is one-to-one. Suppose, for the sake of
contradiction, that $\phi_{t,t'}$ is not one-to-one. Then, there exist
two distinct points $z_1, z_2 \in \mathbb{S}^1$ such that $\phi_{t,t'}(z_1) = \phi_{t,t'}(z_2)$. Since $\phi_t$, $\phi_{t'}$ and $\omega$ are one-to-one
functions, this is possible if and only if $z_1 \notin I_{p,t}$ and $z_2 \in I_{p,t}$
or vice versa. Without loss of generality assume that $z_1 \notin I_{p,t}$ and
$z_2 \in I_{p,t}$. Thus, $\phi_{t,t'}(z_1) = \phi_t(z_1)$ and $\phi_{t,t'}(z_2) = \phi_{t'}(\omega(z_2))$.

Note that any point of $bd(H(t))$ belongs to $cl(s_{w+}) \cup cl(s_{w-})$ for some point $w \in cl(t)$.
Similarly, any point of $bd(H(t'))$ belongs to $cl(s_{w+}) \cup cl(s_{w-})$ for some point $w \in cl(t')$.

Let $y = \phi_{t,t'}(z_1) = \phi_{t,t'}(z_2)$ (see Figure \ref{fig-proof-one-to-one}). Let $x_1$ be the point of $t - \{ q \}$ such that $y \in cl(s_{x_1+}) \cup cl(s_{x_1-})$. Let $x_2$ be the point of $t' - \{ p \}$ such that $y \in cl(s_{x_2+}) \cup cl(s_{x_2-})$.
Let $u'_1$ be the vertical segment of $bd(H(t'))$ below $x_2$. Let $u'_2$ be the vertical segment of $bd(H(t))$ above $x_1$. Note that
$u'_1$ and $u'_2$ may consist of a single point.

Let $R$ denote the rectangle formed by the four points $p$, $q$, $x_1$ and $x_2$.
Note that $\mathcal{U}_t(x_1-) \in u'_2$. Thus, there exists a point $x' \in t$, strictly between the two points $q$ and $x_1$, such that
(i) $x'$ is a point of continuity of $\mathcal{U}_t$ and (ii) $\mathcal{U}_t(x') \in J$ lies strictly in the interior of rectangle $R$.
Further, the upper endpoint $x'' \in J$ of $s_{x_2}$ lies strictly outside the rectangle $R$.

We break the proof into two cases:

\begin{enumerate}
\item {\bf Case I.} $y$ is the lower endpoint of $u'_1$ and the upper endpoint of $u'_2$. 

\item {\bf Case II.} One of the following two statements is true: (i) $y$ is not the lower endpoint of $u'_1$ and (ii) $y$ is not the upper endpoint of $u'_2$. 
\end{enumerate}

By the Jordan curve theorem for a rectangle (a rectangle is the simplest rectilinear polygon and we have assumed the truth of Jordan curve
theorem for all rectilinear polygons), $J$ crosses rectangle $R$ at least two times.

Fist, consider {\bf Case I} (see Figure \ref{fig-proof-one-to-one}(a)). Let $a'_1$ be the upper endpoint of $u'_1$ and $a'_2$ be the lower endpoint of $u'_2$.
By Lemma \ref{lem-vertical-segment-crossing}, $J$ cannot cross rectangle $R$ at any point in $u'_1 - \{y\}$, $u'_2 - \{ y \}$ and
$(u_1 - \{ a_1 \})\cap \overline{pq}$. Further, $J$ cannot cross $s_{x_2}$, $int(t)$, $int(t')$ and $s_q$, since by definition these 
open segments belong to $\mathbb{R}^2 - J$. Thus, $J$ can cross rectangle $R$ only at point $y$. Since $\psi$ is a one-to-one function, this 
contradicts the above conclusion that $J$ crosses $R$ at least two times.

Now, we consider {\bf Case II} (see Figure \ref{fig-proof-one-to-one}(b)). Let $a'_1$ and $b'_1$ be the upper and lower endpoints of $u'_1$. Let $b'_2$ and $a'_2$ be the upper and lower endpoints of
$u'_2$.

First, note that $s_{x_2-} \cap s_{x_1} = \phi$ and $s_{x_2+} \cap s_{x_1} = \phi$.
Suppose, for the sake of contradiction, there exists a point $y' \in s_{x_2-} \cap s_{x_1}$. Since $\lim_{z \rightarrow x_2-} s_z = s_{x_2-}$ and $y' \in s_{x_2-}$, there exists a point $z \in t'$ to the left of $x_2$, such that $h_{y'}$ intersects
$s_z$ and the the horizontal segment $s'$ from $y'$ to $s_{z}$ does not intersect $J$. Let $Q$ be the path from $y'$ to $x_2$ which first follows $s'$, then the open segment $s_{z}$ and finally the horizontal segment $t'$. Since $Q \subset \mathbb{R}^2 - J$, $y' \in int(H(t))$ and $x_2 \in int(t') \subset ext(H(t))$, this contradicts the Jordan curve theorem for the horizontal sweep $H(t)$ proved above. A symmetric argument proves that $s_{x_1-} \cap s_{x_2} = \phi$ and $s_{x_1+} \cap s_{x_2} = \phi$.

Since $u'_1 \subset cl(s_{x_1-}) \cup cl(s_{x_1+})$, we conclude that $a'_2 \leq b'_1$ ($\leq$ is the ordering of points on the line $l_{x_1} = l_{x_2}$ in increasing order of their $y$-coordinates). Similarly, since $u'_2 \subset cl(s_{x_1-}) \cup cl(s_{x_1+})$, we conclude that $b'_2 \leq a'_1$. 

Thus, in this case, $u'_1$ and $u'_2$ overlap and $y \in u'_1 \cap u'_2$. The proof is same as the proof of {\bf Case I} above (the only difference is that $J$ may not be able to cross $R$ even at point $y$).

This concludes the proof that the image $K_{t,t'}$ of $\phi_{t,t'}$ is a Jordan
curve.
$\blacksquare$

\begin{lemma}
\label{interior-jct-horizontal-sweep}
{\bf Interior JCT for extension of a horizontal sweep.} Suppose horizontal sweep $H(t')$ can be extended to $E=(H(t), H(t'))$ by the sweepline algorithm.
Then, there exists a piecewise-vertical Jordan curve $K_{t, t'}$ such that $int(E)$ is a bounded, connected open set with $bd(int(E)) = K_t$ and $ext(E)$ is an unbounded open set.

\end{lemma}

\noindent {\bf Proof:} By Observation \ref{obs-K-t-t'}, $K_{t,t'}$ is a Jordan curve.

We next prove that $int(E)$ is an open set. In the previous section ,we proved that both $int(H(t))$ and $int(H(t'))$ are open sets. Since $int(E) = int(H(t)) \cup int(H(t')) \cup s_p$, it suffices to prove that for every point $z \in s_p$, there exists an $\epsilon > 0$ ($\epsilon$ depends on $z$) such that $B(z,\epsilon) \subset int(E)$. This follows from observing that $s_{p-} \cap s_{p+} = s_p$ and
$s_{p-} = s_{q-}$.

By Assumption $1$ above, all open segments in $H(t)$ and $H(t')$ are finite and hence
lie within the bounding box $\mathcal{B}$. Thus, $int(E)$ is a bounded set.
For any two points $p_1 \in int(t)$ and $p_2 \in int(t')$, we construct a rectilinear
path $Q \subset \mathbb{R}^2 - J$ as follows. Since $s_{q-} = s_{p-}$, there exists
a point $z \in int(t)$ such that (i) $z$ lies to the left of $q$, (ii) line $h_p$ intersects $s_z$ and (iii) the horizontal segment $s'$ from $p$ to $s_z$ does not
contain any point of $J$. The path $Q$ is constructed by first following segment
$t$ from $p_1$ to $z$, then $s_z$ from $z$ to the endpoint of $s'$ on $s_z$, then
following $s'$ to $p$ and finally following $t'$ from $p$ to $p_2$. This proves that
$int(E)$ is a connected set.

Thus, $int(E)$ is a bounded, open and connected set, as claimed.

We now prove that $bd(int(E)) = K_{t,t'}$. Take any convergent sequence
$q_1, q_2, \ldots$ of points of $int(E)$. Let $q^*$ be the limit point of this sequence. Assume that $q^* \notin int(E)$. Then, since $int(E) = int(H(t)) \cup int(H(t')) \cup s_p$, one of the following three cases occurs: 

\begin{enumerate}
\item An infinite subsequence $q_{i_1}, q_{i_2}, \ldots$ of points lies in $int(H(t))$. Then, $q^* \in bd(H(t)) - s_p \in K_{t,t'}$. 

\item An infinite subsequence $q_{i_1}, q_{i_2}, \ldots$ lies in $int(H(t'))$.
Then, $q^* \in bd(H(t')) - s_p \in K_{t,t'}$.

\item An infinite subsequence $q_{i_1}, q_{i_2}, \ldots$ lies in $s_p$.
Then, $q^* \in bd(s_p) \in K_{t,t'}$. 
\end{enumerate}

Thus, we conclude that $bd(int(E)) \subset K_{t,t'}$. Note that $K_{t,t'} = K_t \cup K_{t'} - \{ s_p \} = (K_t - \{s_p\}) \cup (K_{t'}-\{s_p\})$ and $K_{t,t'} \cap int(E) = \phi$. Thus, for every
point $z \in K_{t,t'}$, there exists a convergent sequence $q_1, q_2, \ldots$ such
that (i) for each $i \in \mathbb{N}$, $q_i \in int(E)$ and (ii) $\lim_{i \rightarrow \infty} q_i = z$. Thus, $K_{t,t'} \subset bd(int(E))$. We conclude that
$bd(int(E)) = K_{t,t'}$.

Thus, $cl(int(E)) = int(E) \cup K_{t,t'}$ and hence $ext(E) = \mathbb{R}^2 - cl(int(E))$ is an open set. Since the boundary of 
a set is equal to the boundary of its complement, we conclude that
$bd(ext(E)) = bd(int(E)) = K_{t,t'}$. Since $cl(E) \subset \mathcal{B}$, $ext(E)$ is unbounded.
$\blacksquare$\\

The extension of a horizontal sweep corresponds to two steps of the sweepline algorithm.

%% file: extension-piecewise-vertical-jordan-001.tex
\subsection{Extension of a piecewise-vertical Jordan curve}

We now define the extension of a piecewise-vertical Jordan curve and prove interior Jordan curve theorem for the extension, by generalizing the arguments for extension of a horizontal sweep above. 
This generalization will be used in the definition of a general sweepline algorithm for a Jordan curve (see Section \ref{sec-sweepline-algorithm}).

\begin{definition}
{\bf Extension of a piecewise-vertical Jordan curve.} Let $K$ be a piecewise-vertical Jordan curve, with respect to $J$. Suppose further that the interior Jordan curve theorem for piecewise-vertical Jordan curves (see Lemma \ref{piecewise-vertical-JCT}) has been established for curve $K$. Let $int(K)$ and $ext(K)$ be the interior and exterior of $K$ (since we assume the truth of interior Jordan curve theorem for $K$, $int(K)$ is a bounded, connected and open set, whereas $ext(K)$ is an unbounded open set).
Let $u_1$ be any vertical segment of $\mathcal{V}(K)$. Let $p$ be any point of $int(w)$. Let $t'$ be a horizontal segment such that $t' - \{p\} \subset ext(K)$. 
Then, we say that the piecewise-vertical Jordan curve $K$ can be extended by the horizontal sweep $H(t')$.
\end{definition}

Let $E=(K, H(t'))$ denote the extension of $K$ by $H(t')$.
Let $a_1$ and $b_1$ be the upper and lower endpoints of $u_1$, respectively. Consider the left vertical boundary $s_{p+}$ of $H(t')$. Let $a_2$ and $b_2$ be the upper and lower endpoints of $s_{p+}$.
We next prove that $b_1 \leq b_2 < a_2$.
Note that the open segment $s_p$ at point $p$ is a subset of $u_2=s_{p+}$. Further, $s_p$ is also a subset of $u_1$, since $a_1, b_1 \in J$.

{\bf Construction of $K'$.} We remove $s_p$ from $K$ and replace it
with $K_{t'} - s_p$, to obtain a new curve $K'$ (see Figure \ref{fig-constructing-K-t-t'}). We define
$int(E) = int(K) \cup int(H(t')) \cup \{ s_p \}$ and $ext(E) = \mathbb{R}^2 - \{ K' \cup int(E) \}$.
We now formalize the construction of $K'$ in terms of maps. 
Let $\phi_{K}: \mathbb{S}^1 \rightarrow \mathbb{R}^2$ and $\phi_{t'}: \mathbb{S}^1 \rightarrow \mathbb{R}^2$ be continuous and one-to-one maps such that their
images are the Jordan curves $K$ and $K_{t'}$, respectively.

Let $I_{p,K} = \{ \phi_K^{-1}(z) ~| ~z \in cl(s_p) \}$. $I_{p,K}$ is a closed interval
of $\mathbb{S}^1$, with length strictly less than $2 \pi$. 
Let $I_{p,t'} = \{ \phi_{t'}^{-1}(z) ~| ~z \in cl(s_p) \}$. $I_{p,t'}$ is a closed interval of $\mathbb{S}^1$, with length strictly less than $2 \pi$.

Let $a'$ and $b'$ be the lower and upper endpoints of $cl(s_p)$.
Let $\omega: I_{p,K} \rightarrow I_{p,t'}$ be a continuous bijection between these two closed intervals, such that $\omega(\phi_K^{-1}(a')) = \phi_{t'}^{-1}(a')$ and
$\omega(\phi_K^{-1}(b')) = \phi_{t'}^{-1}(b')$.

Define $\phi_{K'}: \mathbb{S}^1 \rightarrow \mathbb{R}^2$ as follows: (i) for $z \in I_{p,K}$, $\phi_{K'}(z) = \phi_{t'}(\omega(z))$, and (ii) for $z \notin I_{p,K}$, $\phi_{K'}(z) = \phi_K(z)$.
Without loss of generality, assume that $p$ is the left endpoint of $w$.
We set $int(E) = int(K')$ and $ext(E) = ext(K')$.

In the proof of the following observation, we use the Jordan curve theorem for rectilinear polygons, instead of the Jordan curve theorem for rectangles.

\begin{observation}
\label{obs-K'}
$K'$ is a Jordan curve.
\end{observation}

\begin{figure}
\begin{subfigure}{0.5\textwidth}
\scalebox{0.6}{%
\begin{tikzpicture}
    \draw[red] (2, 0) -- (7, 0) node[black][anchor=south]{$t$};
    \draw[very thick,orange] (2,0) -- (2,2.5) -- (0, 2.5) -- (0, 1.5) -- (-2, 1.5) -- (-2, 2.5) node[black][left]{$Q$} -- (-2, 5) -- (-1,5) -- (-1, 4) -- (0,4);
    %\draw[red] (0, -2) -- (0, 3) node[black][pos=0.6][anchor=east]{$s_q$};
    \draw[blue] (0, 3) -- (0, 5) node[black][pos=0.7][anchor=east]{$u_1$};
    \filldraw[black] (0, 5) circle (0.05) node[black][anchor=east]{$a_1$};
    \filldraw[black] (0, 3) circle (0.05) node [black][anchor=east]{$b_1$};
    %\filldraw[black] (0, 0) circle (0.05) node [black][anchor=east]{$q$};
    \filldraw[black] (0, 4) circle (0.05) node[black][anchor=east]{$p$};
    \draw[red] (0, 4) -- (5.5,4) node[black][anchor=south]{$t'$};
    \filldraw[black] (3, 0) circle (0.05) node[black][above, left]{$x_1$};
    \filldraw[black] (3, 4) circle (0.05) node[black][above, left]{$x_2$};
    \draw[red] (3, 5) -- (3,3) node[black][pos=0.3][left]{$s_{x_2}$};
    \filldraw[black] (3, 5) circle (0.05) node[black][above, left]{$x''$};
    \draw[red] (3, -1) -- (3,1) node[black][pos=0.2][left]{$s_{x_1}$};
    \draw[blue] (3, 3) -- (3, 2) node[black][pos=0.7][anchor=east]{$u'_1$};
    \draw[blue] (3, 1) -- (3, 2) node[black][pos=0.7][anchor=east]{$u'_2$};
    \filldraw[black] (3, 3) circle (0.05) node[black][above, right]{$a'_1$};
    \filldraw[black] (3, 1) circle (0.05) node[black][above, right]{$a'_2$};
    \filldraw[black] (3, 2) circle (0.05) node[black][above, right]{$y$};
    \filldraw[black] (2.5, 1) circle (0.05) node[black][above, left]{$x'$};
    
\end{tikzpicture}
}
\caption{{\bf Case I}}
\end{subfigure}
\begin{subfigure}{0.5\textwidth}
\scalebox{0.6}{%
\begin{tikzpicture}
    \draw[red] (2, 0) -- (7, 0) node[black][anchor=south]{$t$};
    \draw[very thick,orange] (2,0) -- (2,2.5) -- (0, 2.5) -- (0, 1.5) -- (-2, 1.5) -- (-2, 2.5) node[black][left]{$Q$} -- (-2, 5) -- (-1,5) -- (-1, 4) -- (0,4);
  
    %\draw[red] (-3, 0) -- (7, 0) node[black][anchor=south]{$t$};
    %\draw[red] (0, -2) -- (0, 3) node[black][pos=0.6][anchor=east]{$s_q$};
    \draw[blue] (0, 3) -- (0, 5) node[black][pos=0.7][anchor=east]{$u_1$};
    \filldraw[black] (0, 5) circle (0.05) node[black][anchor=east]{$a_1$};
    \filldraw[black] (0, 3) circle (0.05) node [black][anchor=east]{$b_1$};
    %\filldraw[black] (0, 0) circle (0.05) node [black][anchor=east]{$q$};
    \filldraw[black] (0, 4) circle (0.05) node[black][anchor=east]{$p$};
    \draw[red] (0, 4) -- (5.5,4) node[black][anchor=south]{$t'$};
    \filldraw[black] (3, 0) circle (0.05) node[black][above, left]{$x_1$};
    \filldraw[black] (3, 4) circle (0.05) node[black][above, left]{$x_2$};
    \draw[red] (3, 5) -- (3,3) node[black][pos=0.3][left]{$s_{x_2}$};
    \filldraw[black] (3, 5) circle (0.05) node[black][above, left]{$x''$};
    \draw[red] (3, -1) -- (3,1) node[black][pos=0.2][left]{$s_{x_1}$};
    \draw[blue] (3, 3) -- (3, 1.5) node[black][pos=0.3][anchor=east]{$u'_1$};
    \draw[blue] (3, 1) -- (3, 2.5) node[black][pos=0.3][anchor=east]{$u'_2$};
    \filldraw[black] (3, 3) circle (0.05) node[black][above, right]{$a'_1$};
    \filldraw[black] (3, 1) circle (0.05) node[black][above, right]{$a'_2$};
    \filldraw[black] (3, 1.5) circle (0.05) node[black][above, right]{$b'_1$};
    \filldraw[black] (3, 2.5) circle (0.05) node[black][above, right]{$b'_2$};
    \filldraw[black] (3, 2) circle (0.05) node[black][above, right]{$y$};
    \filldraw[black] (2.5, 1) circle (0.05) node[black][above, left]{$x'$};
  
\end{tikzpicture}
}
\caption{{\bf Case II}}
\end{subfigure}
\caption{$\phi_{K'}$ is one-to-one.}
\label{fig-proof-K'-jordan-curve}
\end{figure}
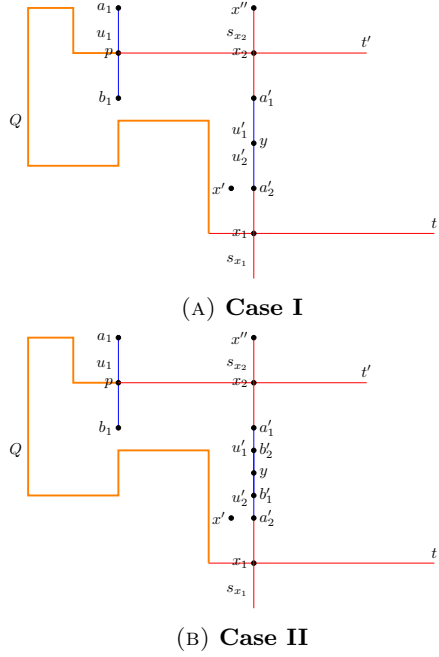

\noindent {\bf Proof:} The proof is on the same lines as the proof of Observation \ref{obs-K-t-t'} (see Figure \ref{fig-proof-K'-jordan-curve}). The main differences are the following. Since the interior Jordan curve theorem is true for $K$, (i) there exists a
horizontal segment $t$ such that $u_2$ is a vertical segment on the boundary of the horizontal sweep $H(t)$ and (ii) there is rectilinear path $Q \subset \mathbb{R}^2 - J$ from $x_2$ to $p$ with a finite number of turns. Instead of rectangle $R$, we consider the rectilinear polygon formed by path $Q$, the line segment $\overline{px_2}$ and the line segment ${x_2x_1}$. In the proof, we use the Jordan curve theorem for rectilinear polygons, instead of the Jordan curve theorem for rectangles.
$\blacksquare$\\

\begin{lemma}
\label{piecewise-vertical-JCT}
{\bf Interior JCT for extension of a piecewise-vertical Jordan curve.} If a piecewise-vertical Jordan curve $K$ can be extended by horizontal sweep $H(t')$ {\bf and} interior JCT (there is recursion here; we derive the truth of interior JCT for $K'$ from the truth of interior JCT for $K$) is true for $K$, then there exists a piecewise-vertical Jordan curve $K'$ such that

\begin{enumerate}
    \item $int(E)$ is a bounded, connected open set with $bd(int(E)) = K'$ and $ext(E)$ is an unbounded open set,
    \item the set $int(E)$ is the union of a collection of open segments $s_p$, $p \in I_{K'} \subset \mathbb{R}^2 - J$. Further, for any two open segments $s_{p_1}$ and $s_{p_2}$ ($p_1, p_2 \in I_{K'}$), there
    is a rectilinear path with a finite number of turns from a
    point $w_1 \in s_{p_1}$ to a point $w_2 \in s_{p_2}$, and
    \item for every vertical segment $u \in \mathcal{V}_{K'}$, there
    exists a horizontal segment $g_u$ such that (i) $int(g(u)) \subset int(K')$ and (ii) $u$ is a vertical segment of the
    Jordan curve $K_{g_u}$ forming the boundary of horizontal sweep $H(g_u)$.
\end{enumerate}
\end{lemma}

\noindent {\bf Proof:} The proof is on the same lines as the proof of Lemma \ref{interior-jct-horizontal-sweep}. $\blacksquare$

Note that the above lemma is sufficient to prove interior JCT for any finite sequence of horizontal sweeps. As an example, consider the five consecutive horizontal sweeps $(H(t_1), H(t_2), \ldots, H(t_5))$ in Figure \ref{koch}. By Lemma \ref{JCT-horizontal-sweep}, interior JCT is true for the initial horizontal sweep $H(t_1)$ (base case).
By the above lemma, for each $1 \leq i \leq 4$, if interior JCT is true till the first $i$ horizontal sweeps, then it is also true for the first $i+1$ horizontal sweeps (the inductive step). 

%% file: infinite-ray-001.tex
\subsection{Infinite rays of horizontal sweeps}

\label{sec-infinite-ray}

\begin{definition}
An {\bf infinite ray} $r$ is a countably infinite sequence $H(t_1), H(t_2), \ldots$ of horizontal sweeps such that

\begin{enumerate}
    \item $H(t_1)$ can be extended by $H(t_2)$ from a vertical segment of $bd(H(t_1))$. Let $E_2 = (H(t_1), H(t_2))$. Let $K_2$ be the piecewise-vertical Jordan curve bounding $int(E_2)$.

    \item $K_2$ can be extended by $H(t_3)$ from a vertical segment of $bd(H(t_2))$. Let $K_3$ be the piecewise-vertical Jordan curve bounding $E_3 = (H(t_1), H(t_2), H(t_3))$.
    
    \item $K_3$ can be extended by $H(t_4)$ from a vertical segment of $bd(H(t_3))$. Let $K_4$ be the piecewise-vertical Jordan curve bounding $(H(t_1), H(t_2), H(t_3), H(t_4))$.

    \item and so on till infinity.
\end{enumerate}
\end{definition}

{\it Note.} The existence of $K_2, K_3, \ldots$ as well as the truth of JCT for them, follows from Lemma \ref{JCT-horizontal-sweep} (base case on $H(t_1)$) and Lemma \ref{piecewise-vertical-JCT} (the inductive step). An infinite ray can be viewed as a special case of the sweepline algorithm for Jordan curves, where the recursion tree is an infinite path (see Section \ref{sec-sweepline-algorithm}).
In Figure \ref{koch}, $(H(t_1), H(t_2), H(t_3))$ is a finite ray, but $(H(t_1), H(t_2), H(t_4))$ is not a finite ray.

\subsubsection{Structure of terminating and non-terminating rays}
We now study the structure of infinite rays. The first lemma, called the {\it tunnel lemma}, places
restrictions on the structure of region swept by an infinite ray. 

\begin{lemma}
{\bf Tunnel lemma.} Let $s_{i_1}, s_{i_2}, \ldots$ ($i_1 < i_2 < \cdots$) be an infinite sequence of open segments such that (i) for each $k \in \mathbb{N}$, $s_{i_k} \cap t_{i_k} \neq \phi$ and (ii) $\lim_{k \rightarrow \infty} s_{i_k}$ exists. Let $s^*$ be the limit of this
sequence of open segments. Assume that $s^*$ has positive length. 

Then, there exists a natural number $N_0$ and a horizontal sweep $H(t_{s^*})$ such that (i) $s^* \cap int(t_{s^*}) \neq \phi$ and (ii) for all natural numbers $j > N_0$, $s_{i_j} \cap int(t_{s^*}) \neq \phi$.
\end{lemma}

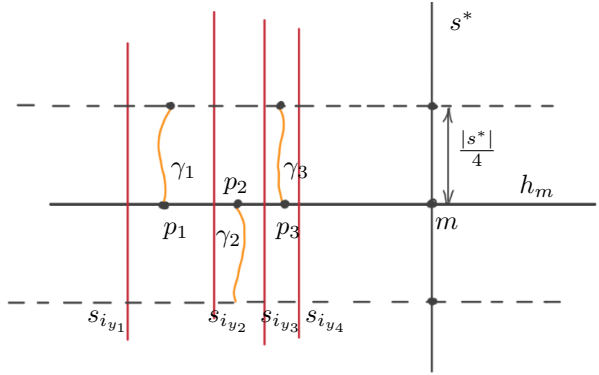
\begin{figure}
\centering
\begin{tikzpicture}
\node[anchor=south west,inner sep=0] (image) at (0,0) {\includegraphics[width=9cm]{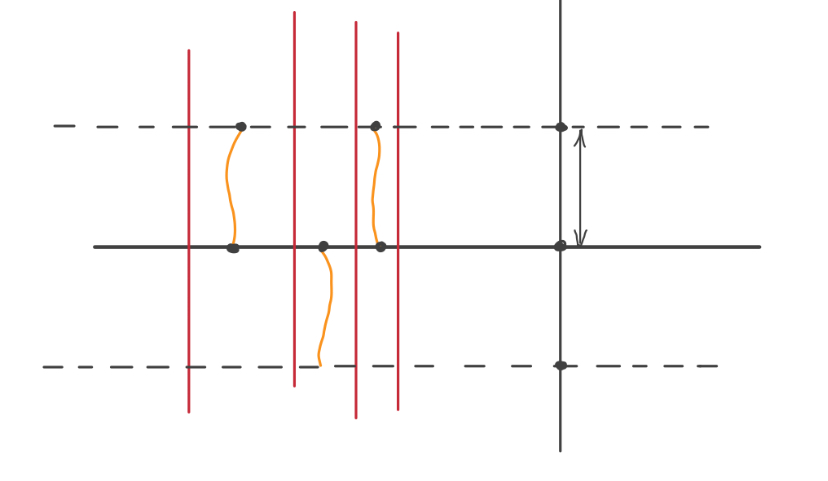}};
\node (A) at (2.7, 2.2) {$p_1$};
\node (D) at (2.8, 3) {$\gamma_1$};
\node (G) at (1.8, 1) {$s_{i_{y_1}}$};
\node (B) at (3.5, 2.8) {$p_2$};
\node (E) at (3.4, 2.1) {$\gamma_2$};
\node (C) at (4.2, 2.2) {$p_3$};
\node (F) at (4.3, 3) {$\gamma_3$};
\node (H) at (4.7, 1) {$s_{i_{y_4}}$};
\node (H1) at (4.1, 1) {$s_{i_{y_3}}$};
\node (H2) at (3.4, 1) {$s_{i_{y_2}}$};
\node (I) at (6.5, 5) {$s^*$};
\node (J) at (6.7, 3.3) {$\frac{|s^*|}{4}$};
\node (H) at (6.3, 2.3) {$m$};
\node (K) at (7.5, 2.8) {$h_m$};

%\foreach \y in {1, 2, ..., 5}
%{
%\draw (0,\y) -- (7.5,\y);
%\node at (-0.5,\y){\y};
%}

%\foreach \x in {1,2,...,7}
%{
%\draw (\x,0) -- (\x,5.5);
%\node at (\x,-0.5){\x};
%}

\end{tikzpicture}
\caption{Tunnel lemma: {\bf Case I}}
\label{fig-tunnel-lemma-1}
\end{figure}

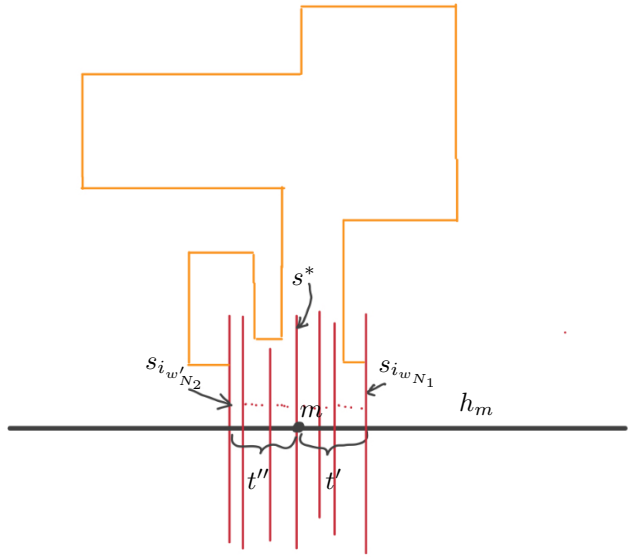
\begin{figure}
\centering
\begin{tikzpicture}
\node[anchor=south west,inner sep=0] (image) at (0,0) {\includegraphics[width=9cm]{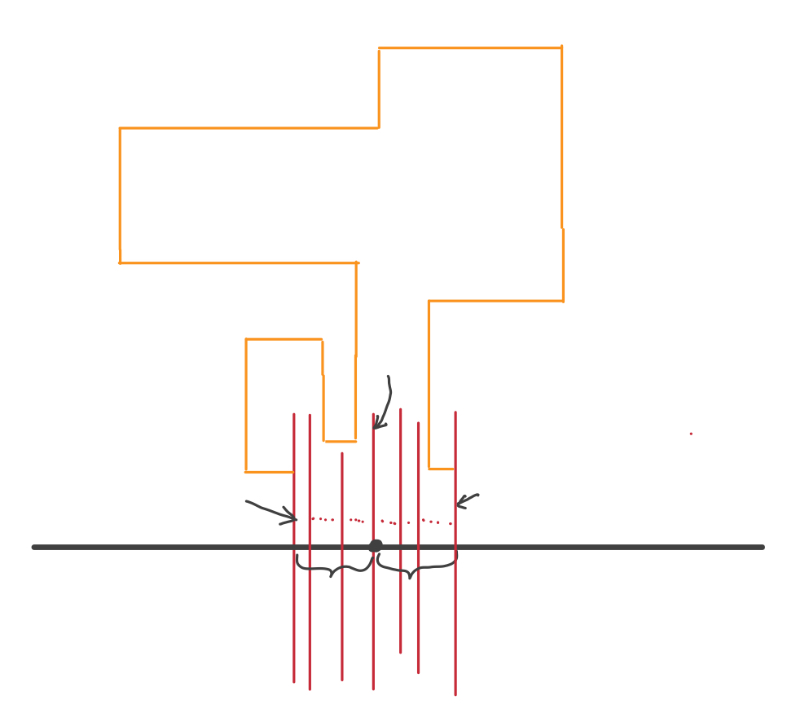}};
\node (A) at (2.6, 2.7) {$s_{i_{w'_{N_2}}}$};
\node (B) at (5.7, 2.7) {$s_{i_{w_{N_1}}}$};
\node (C) at (4.3, 4) {$s^*$};
\node (D) at (3.7, 1.3) {$t''$};
\node (E) at (4.7, 1.3) {$t'$};
\node (F) at (6.2, 8) {$Q$: path from $s_{i_{w_{N_1}}}$ to $s_{i_{w'_{N_2}}}$ in ray $r$};
\node (G) at (6.6, 2.3) {$h_m$};
\node (H) at (4.4, 2.2) {$m$};
\end{tikzpicture}
\caption{Tunnel lemma: {\bf Case II}}
\label{fig-tunnel-lemma-2}
\end{figure}

\noindent {\bf Proof:} We consider two cases:

\begin{enumerate}
\item {\bf Case I.} There exists a natural number $N$ such that for all $j > N$, $s_{i_j}$ and $s_{i_{j+1}}$
lie on the same side of the vertical line containing $s^*$.

\item {\bf Case II.} Case I does not hold.
\end{enumerate}

First, let us consider {\bf Case I}. We can extract a subsequence $s_{i_{j_1}}, s_{i_{j_2}}, \ldots$
($j_1 < j_2 < \cdots$) such that for each natural number $k \in \mathbb{N}$,  $s_{i_{j_{k+1}}}$ lies between $s_{i_{j_{k}}}$ and the vertical line containing $s^*$.

Let $m$ be the midpoint of $s^*$. Let $h_m$ be the horizontal line passing through $m$. 
Since $\lim_{k \rightarrow \infty} s_{i_{j_k}} = s^*$, there exists a natural number $N_1$ such that
for all $k > N_1$, $s_{i_{j_k}}$ intersects line $h_m$, (ii) the distance of the upper endpoint of $s_{i_{j_k}}$ from $h_m$ is at least $\frac{|s^*|}{4}$ and (iii) the distance of the lower endpoint
of $s_{i_{j_k}}$ from $h_m$ is also at least $\frac{|s^*|}{4}$. 

Let $u_{j_k}$ be the line segment of $h_m$ bounded by $s^*$ and $s_{i_{j_k}}$ on either side.
A segment $s_{i_{j_k}}$ is {\it bad} if and only if $u_{j_k}$ contains at least one point of the
Jordan curve $J$.

Suppose there are an infinite number of bad segments. Then, there exists an infinite sequence $y_1, y_2, \ldots$ of natural numbers such that (i) $y_1, y_2, \ldots$ is a subsequence of $j_1, j_2, \ldots$ and (ii) for each $w \in \mathbb{N}$, there is a point $p_w$ of the Jordan curve $J$ in the interior of the line segment of $h_m$ bounded by $s_{i_{y_w}}$ and $s_{i_{y_{w+1}}}$.

By an argument on the same lies as that used in the proof of Lemma \ref{lemma-left-limit-exists} (see Figure \ref{fig-tunnel-lemma-1}), we arrive at a contradiction. In Figure \ref{fig-tunnel-lemma-1},
$p_1, p_2, \ldots$ are points of $J$. The two dashed horizontal lines pass through points of $s^*$
at distance $\frac{|s^*|}{4}$ on either side of its midpoint $m$. Since the upper endpoint of
$s^*$ lies outside the horizontal strip formed by the two dashed horizontal lines, we can find
disjoint Jordan arcs $\gamma_1, \gamma_2, \ldots$ of $J$, as depicted in the figure. There is
an infinite subsequence $\gamma_{i_1}, \gamma_{i_2}, \ldots$ ($1 \leq i_1 < i_2 < \cdots$) such that
the two endpoints of these arcs converge to $m$ and a point at distance $\frac{|s^*|}{4}$ from $m$
on $s^*$. We arrive at a contradiction by applying Observation \ref{obs-jordan-arc-limit} to $\gamma_{i_1}, \gamma_{i_2}, \ldots$. 

Thus, there are a finite number of bad segments.

Let $N'$ be a natural number such that $s_{i_{j_{N'}}}$ is not a bad segment. Let $t'$ be the
horizontal segment $u_{j_{N'}}$. Let $N_0$ be a natural number such that for all $j > N_0$,
$s_{i_j} \cap int(t') \neq \phi$. Set $t_{s^*} = t'$ and conclude the statement of the lemma.

We now consider {\bf Case II}. In this case, there are an infinite number of segments from the sequence $s_{i_1}, s_{i_2}, \ldots$ on both sides of the vertical line containing $s^*$. 
Let $A_1$ be the subsequence consisting of all segments to the left of $s^*$ and $A_2$ be the
subsequence consisting of all segments to the right of $s^*$. Let $A_1$ be equal to $s_{i_{w_1}}, s_{i_{w_2}}, \ldots$ and $A_2$ be equal to $s_{i_{w'_1}}, s_{i_{w'_2}}, \ldots$, where $\{i_{w_j} ~| ~j \in \mathbb{N}\} \cup 
\{ i_{w'_j} ~| ~j \in \mathbb{N}\} = \mathbb{N}$ and $\cup$ denotes disjoint union.

We apply the argument of {\bf Case I} to $A_1$ to obtain a natural number $N_1$ such that for all
$s_{i_{w_j}}$, $j \geq N_1$, the portion $t'$ of $h_m$ between $s_{i_{w_j}}$ and $s^*$ has no points of $J$.
Similarly, we apply the argument of {\bf Case I} to $A_2$ to obtain a natural number $N_2$ such that for all
$s_{i_{w'_j}}$, $j \geq N_2$, the portion $t''$ of $h_m$ between $s_{i_{w'_j}}$ and $s^*$ has no points of $J$.

Without loss of generality, assume that $N_1 < N_2$. Let $Q \subset \mathbb{R}^2 - J$ be the rectilinear path traced by ray $r$ from $s_{i_{w_{N_1}}}$ to $s_{i_{w'_{N_2}}}$ (see Figure \ref{fig-tunnel-lemma-2}). 

Since $Q$ is generated by a finite path of ray $r$, by interior JCT for sweepline algorithms
with a finite number of steps, $Q$ does not contain more than one point from any
open segment. 

We next prove that $Q$ does not contain any point of $s^*$. Suppose $Q$
contains at least one point of $s^*$. Let $u$ be the portion of a horizontal segment of $Q$
such that one endpoint of $u$ belongs to $s^*$. Without loss of generality, assume that
$u$ lies to the left of the vertical line containing $s^*$. Then, there exists a natural number
$N_3$ such that for all $j \geq N_3$, $s_{i_{w_{j}}}$ intersects $u$. Let $Q' \subset \mathbb{R}^2 - J$ be the rectilinear path generated by ray $r$ from $s_{i_{w_{N_1}}}$ to $s_{i_{w_{N_4}}}$, where
$w_{N_4} > \max(w'_{N_2}, w_{N_1}, w_{N_3})$. Then, $Q'$ contains at least two points from an open segment $s_x$, for some point $x \in u$. This contradicts the fact that $Q'$ is generated by a finite path of ray $r$.

Consider the rectilinear polygon $R$ formed by $s_{i_{w_{N_1}}}$, $s_{i_{w'_{N_2}}}$, $t'$, $t''$, and the path $Q$.
Since we have assumed the truth of the Jordan curve theorem for rectilinear polygons, interior and
exterior with respect to $R$ are well-defined. 

Thus, there is an endpoint of $s^*$ in the interior of $R$ and an endpoint of $s^*$ is the 
exterior of $R$. Since the endpoints belong to $J$, the Jordan curve $J$ must 
intersect $R$. Hence, we arrive at a contradiction.
Thus, {\bf Case II} cannot occur.

Hence, the lemma is proved.
$\blacksquare$

Let $s_1, s_2, \ldots$ be an infinite sequence of open segments such that for each $i \in \mathbb{N}$, $s_i \in int(H(t_i))$.
Since all open segments lie within the bounding box $\mathcal{B}$, there exists a convergent subsequence
$s_{i_1}, s_{i_2}, \ldots$ ($i_1 < i_2 < \cdots$). Let $s^*$ be the limit of this subsequence. We say that the closure 
$cl(s^*)$ is a {\it limiting segment} of ray $r$.

Let $W(r)$ be the set of all limiting segments of ray $r$.

\begin{lemma}
\label{lem-W-r}
$|W(r)| = 1$.
\end{lemma}

\begin{figure}
\begin{tikzpicture}
%\begin{center}
\node at (3,3) {\includegraphics[width=9cm]{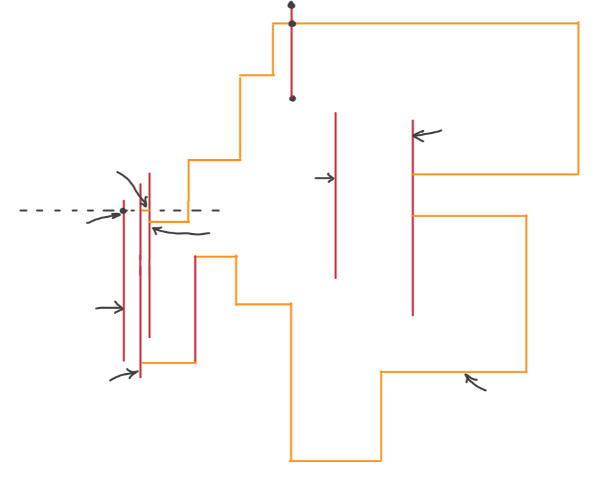}};
%\includegraphics[width=9cm]{W-r-fig1};
%\foreach \y in {1, 2, ..., 7}
%{
%\draw (0,\y) -- (7.5,\y);
%\node at (-0.5,\y){\y};
%};

%\foreach \x in {1,2,...,7}
%{
%\draw (\x,0) -- (\x,7.5);
%\node at (\x,-0.5){\x};
%};
\node (A) at (5.7, 0.7) {$Q'$};
\node (B) at (5.1, 4.6) {$s_{i'_{N_2}}$};
\node (C) at (2.8, 3.7) {$s^*_2$};
\node (D) at (1.7, 3.2) {$s_{i_{N_1}}$};
\node (E) at (-0.1, 0.8) {$s_{i_{N_3}}$};
\node (F) at (-0.3, 2) {$s^*_1$};
\node (G) at (-0.5, 3.1) {$q_1$};
\node (H) at (0, 4.1) {$t'_1$};
\node (I) at (2.5, 6.4) {$s_{p_2}$};
%\node (C) at (5.7, 0.7) {$Q'$};

\end{tikzpicture}
\caption{Geometric cases for the proof that $|W(r)|=1$.}
\label{fig1-W-r}
%\end{center}
\end{figure}
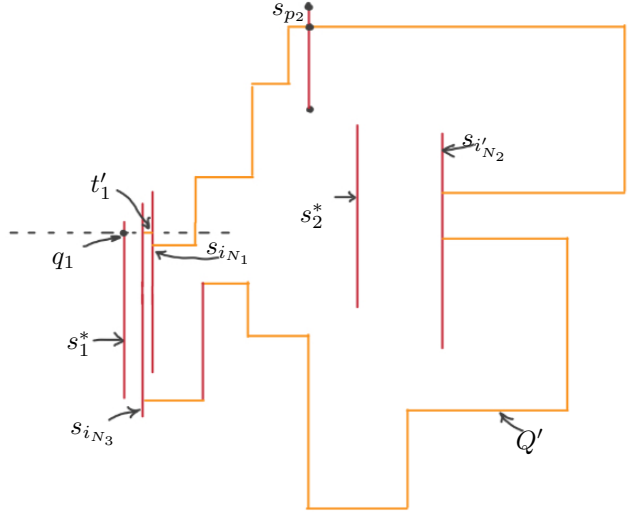

\begin{figure}
%\begin{center}
\begin{tikzpicture}

\node at (3,3) {\includegraphics[width=9cm]{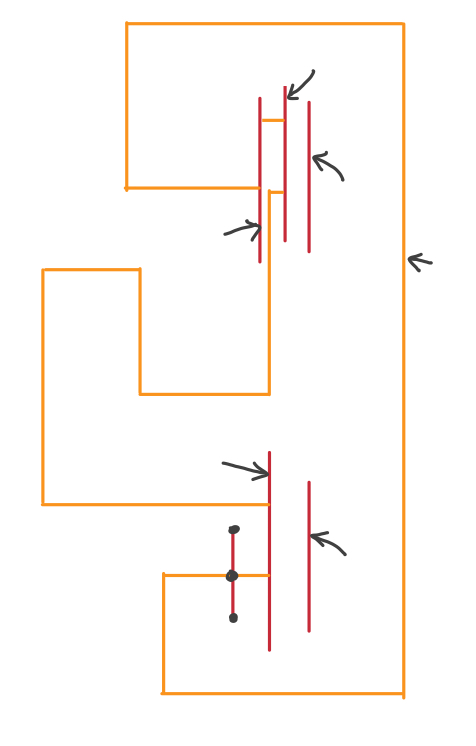}};

%\includegraphics[width=6cm]{W-r-fig2}
%\includegraphics[width=9cm]{W-r-fig1};
%\foreach \y in {-1, 0, 1, 2, ..., 7, 8, 9}
%{
%\draw (0,\y) -- (7.5,\y);
%\node at (-0.5,\y){\y};
%};

%\foreach \x in {1,2,...,7}
%{
%\draw (\x,-1) -- (\x,9.5);
%\node at (\x,-0.5){\x};
%};
\node (A) at (5.1, -0.8) {$s^*_2$};
\node (B) at (2.5, 1.5) {$s_{i'_{N_2}}$};
\node (C) at (1.9, -0.6) {$u$};
\node (D) at (2.6, -0.6) {$s_{p_2}$};
\node (E) at (6.8, 5) {$Q'$};
\node (F) at (5.1, 6.4) {$s^*_1$};
\node (G) at (4.6, 8.9) {$s_{i_{N_3}}$};
\node (H) at (2.5, 5.5) {$s_{i_{N_1}}$};
%\node (A) at (5.7, 0.7) {$Q'$};

\end{tikzpicture}
\caption{Geometric cases for the proof that $|W(r)|=1$ (contd.)}
\label{fig2-W-r}
%\end{center}
\end{figure}
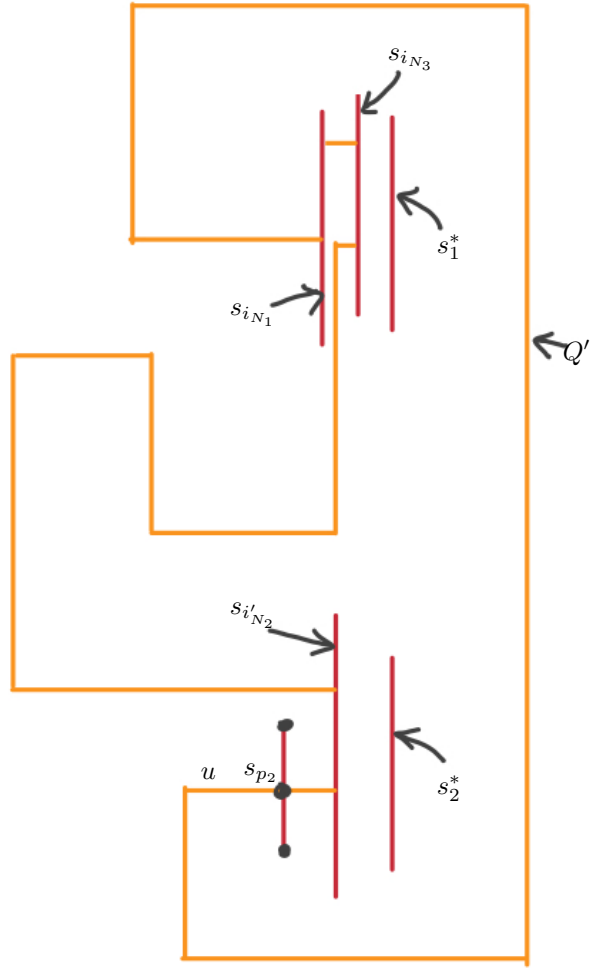

\noindent {\bf Proof:} The main argument of the proof is on the same lines as the impossibility of {\bf Case II}
in the proof of tunnel lemma above. 

Clearly, $|W(r)| \geq 1$. Suppose, for the sake of contradiction, that $|W(r)| > 1$. Let $s^*_1$ and 
$s^*_2$ be two distinct limiting open segments in $W(r)$. 

Let $A_1 = s_{i_1}, s_{i_2}, \ldots$ ($i_1 < i_2 < \cdots$ and $\forall k \in \mathbb{N}, ~s_{i_k} \cap t_{i_k} \neq \phi$) be a sequence of open segments converging to $s^*_1$. Let $A_2 = s_{i'_1}, s_{i'_2}, \ldots$ ($i'_1 < i'_2 < \cdots$ and $\forall k \in \mathbb{N}, ~s_{i'_k} \cap t_{i'_k} \neq \phi$) be a sequence of open segments converging to $s^*_2$.

Let $q_1$ be an arbitrary point in $int(s^*_1)$ and $q_2$ be an arbitrary point in $int(s^*_2)$. Let $h_1$ and
$h_2$ be the horizontal lines through $q_1$ and $q_2$, respectively.
We apply tunnel lemma to $A_1$ to obtain a natural number $N_1$ such that for all
$s_{i_j}$, $j \geq N_1$, the portion $t'$ of $h_1$ between $s_{i_j}$ and $s^*_1$ has no points of $J$ (though the tunnel lemma has been proved for $h_m$, the proof also applies to line $h_1$ through any point $q_1 \in int(s^*_1)$).
Similarly, we apply the tunnel lemma to $A_2$ to obtain a natural number $N_2$ such that for all
$s_{i'_j}$, $j \geq N_2$, the portion $t''$ of $h_2$ between $s_{i'_j}$ and $s^*_2$ has no points of $J$.

Without loss of generality, assume that $i_{N_1} < i'_{N_2}$. Let $N_3$ be a natural number such that (i) $i_{N_3} > i'_{N_2}$ and (ii) there exists a real number $\delta > 0$ such that the $x$-coordinate of any
point of $s_{i'_{N_2}}$ and the $x$-coordinate of any point of $s_{i_{N_1}} \cup s_{i_{N_3}}$ differ in absolute value by at least $\delta$ (both $s_{i_{N_1}}$ and $s_{i_{N_3}}$ lie on the same side of the vertical line containing $s_{i'_{N_2}}$. Clearly, $N_3 > N_1$.

We consider two different geometric cases: (i) $s^*_1$ and $s^*_2$ are on different vertical lines (see Figure \ref{fig1-W-r}), and
(ii) $s^*_1$ and $s^*_2$ are on the same vertical line (see Figure \ref{fig2-W-r}).

Let $Q \subset \mathbb{R}^2 - J$ be the rectilinear path $Q$ given by ray $r$ from $s_{i_{N_1}}$ to $s_{i_{N_3}}$. $Q$ also has a point of open segment $s_{i'_{N_2}}$. Note that $Q$ does not have two points from any open segment. Let $t'_1$ be the portion of $t'$ between $s_{i_{N_1}}$ and $s_{i_{N_3}}$. 

We start from the point $p_1 \in Q \cap s_{i'_{N_2}}$ and walk outwards on $Q$ in both directions till 
we either reach $s_{i_{N_1}}$ or $s_{i_{N_3}}$, or we reach a point of the horizontal segment $t'_1$.
Let $Q'$ be the subpath of $Q$ obtained as a result of this walk. We construct a rectilinear polygon $R \subset \mathbb{R}^2 - J$ using $t'_1$, $s_{i_{N_1}}$, $s_{i_{N_3}}$ and the rectilinear path $Q'$.

Let $p_2$ be a point in the interior of a horizontal segment of $Q'$ such that $x$-coordinate of $p_2$ is
different from any point in $t'_1 \cup \{ s_{i'_{N_2}} \}$. Then, the open segment $s_{p_2}$ cannot cross
rectangle $R$ (if $s_{p_2}$ crosses $R$, it must cross it at a point $p'_2$ of $Q'$ and then $Q'$ will
contain two points from the same open segment, a contradiction). 

Thus, there is an endpoint of $s_{p_2}$ in the interior of $R$ and an endpoint of $s_{p_2}$ in the exterior of $R$ (the interior and exterior of $R$ exist because of JCT for rectilinear polygons). We conclude that $J$ crosses $R$, and hence arrive at a contradiction.

Let us now consider case (ii). Without loss of generality, assume that $s^*_2$ lies below $s^*_1$, on the
common vertical line containing both of these open segments. In addition to ensuring the above properties, 
we choose $N_1$, $N_2$ and $N_3$ such that $s_{i'_{N_2}}$ does not intersect the horizontal line $h_1$.

The construction of $R$ is the same as case (i). If $R$ crosses $s_{i'_{N_2}}$, one endpoint of $s_{i'_{N_2}}$
will lie in the interior of $R$ and the other endpoint will lie in the exterior of $R$. On the other hand, if $R$ 
does not cross $s_{i'_{N_2}}$, let $u$ be the horizontal segment of $R$ with the lowest $y$-coordinate
among all horizontal segments of $R$ with one endpoint at $s_{i'_{N_2}}$. Then, one can find a point
$p_2 \in int(u)$ such that one endpoint of $s_{p_2}$ lies in the interior of $R$ and the other endpoint
lies in the exterior of $R$. Thus, in both cases, the Jordan curve $J$ will intersect $R$, and we again
arrive at a contradiction.

Hence, the lemma is proved. $\blacksquare$

We denote the unique limiting segment in $W(r)$ by $s^*_r$.

\begin{definition}
{\bf Terminating and non-terminating infinite rays.} Ray $r$ is a terminating ray iff $s^*_r$ has length $0$ (i.e.,
$cl(s^*_r)$ consists of a single point of the Jordan curve $J$). Ray $r$ is a non-terminating ray iff $s^*_r$ has 
positive length.
\end{definition}

\begin{definition}
{\bf Interior of ray $r$.} The interior $int(r)$ of ray $r$ is equal to $\cup_{i=2}^{\infty} int(E_i)$.
\end{definition}

\begin{definition}
{\bf Relation $\sim$ on rays of horizontal sweeps.} A ray $r_1$ is related to ray $r_2$ ($r_1 \sim r_2$) iff 
$int(r_1) = int(r_2)$.
\end{definition}

\begin{observation}
$\sim$ is an equivalence relation on the set of all rays of horizontal sweeps.
\end{observation}

We denote the equivalence class of ray $r$ by $[r]$.

\begin{lemma}
\label{lem-non-term}
Let $r$ be a non-terminating infinite ray of horizontal sweeps. Then, there exists a finite ray $r'$ of horizontal sweeps
such that $r \sim r'$ (i.e., $r' \in [r]$).
\end{lemma}

\begin{figure}
%\begin{center}
\begin{tikzpicture}
\node at (3,3) {\includegraphics[width=9cm]{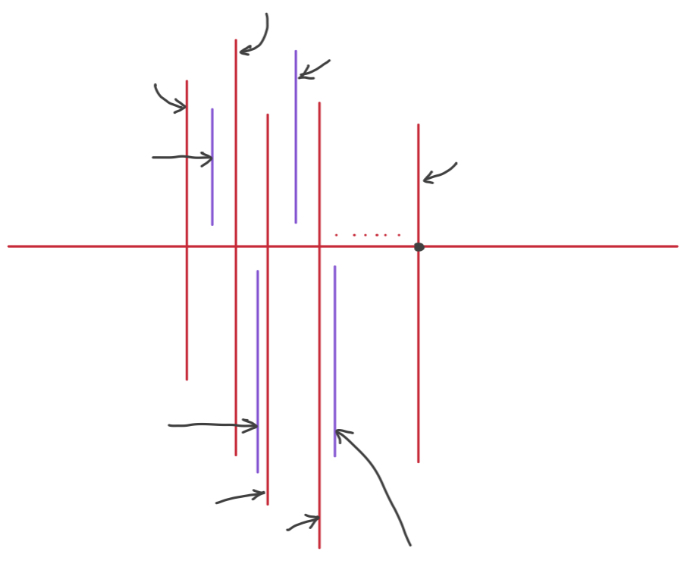}};

%\includegraphics[width=9cm]{fig-case-ii}
%\foreach \y in {-1, 0, 1, 2, ..., 7, 8, 9}
%{
%\draw (0,\y) -- (7.5,\y);
%\node at (-0.5,\y){\y};
%};

%\foreach \x in {1,2,...,7}
%{
%\draw (\x,-1) -- (\x,9.5);
%\node at (\x,-0.5){\x};
%};
\node (A) at (6, 3.8) {$t_{s^*_r}$};
\node (B) at (4.1, 3.3) {$m$};
\node (C) at (4.6, 4.5) {$s^*_r$};
\node (D) at (0.3, 5.5) {$s_{i_N}$};
\node (E) at (2.1, 6.6) {$s_{i_{N+1}}$};
\node (F) at (1, 0) {$s_{i_{N+2}}$};
\node (G) at (2, -0.2) {$s_{i_{N+3}}$};
\node (H) at (0.3, 4.5) {$u_{i'_{j_1}}$};
\node (I) at (0.3, 1.2) {$u_{i'_{j_2}}$};
\node (J) at (2.8, 6) {$u_{i'_{j_3}}$};
\node (K) at (4, -0.8) {$u_{i'_{j_4}}$};
%\node (A) at (5.1, -0.8) {$s^*_2$};
\end{tikzpicture}
\caption{Proof that $u^* \neq s^*_r$, for case (ii). An infinite subsequence of $u_{i'_{j_1}}$, $u_{i'_{j_2}}$,
$\ldots$ lies above (or below) $t_{s^*_r}$ and hence $u^*$ lies above (or below) the horizontal line $h_m$.}
\label{fig-case-ii}
%\end{center}
\end{figure}
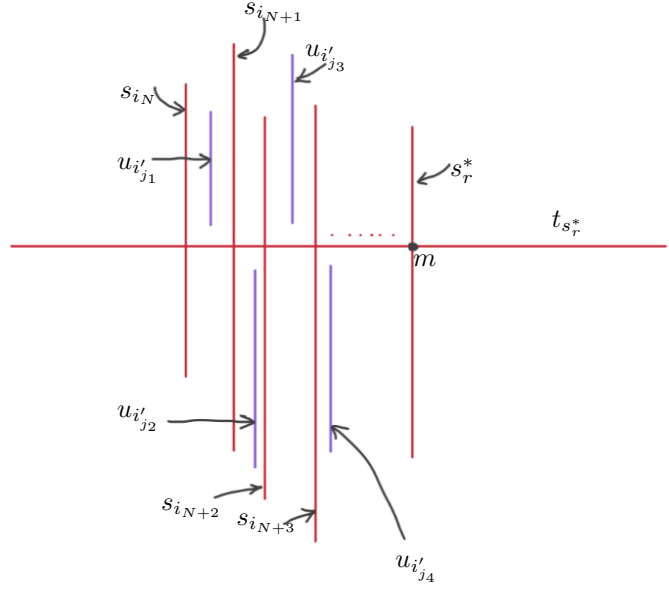

\noindent {\bf Proof:} Since $s^*_r \in W(r)$, there exists an infinite sequence
$s_{i_1}, s_{i_2}, \ldots$ ($i_1 < i_2 < \cdots$) of open segments converging to $s^*_r$ such that
for each $k \in \mathbb{N}$, $s_{i_k} \cap t_{i_k} \neq \phi$.

Applying the tunnel lemma to this convergent sequence, we obtain a natural number $N$
and a horizontal segment $t_{s^*_r}$ such that (i) $s^*_r \cap t_{s^*_r} \neq \phi$ and (ii)
for all natural number $k \geq N$, $s_{i_k} \cap t_{s^*_r} \neq \phi$.

Let $U$ be the set of all open segments $u$ in $int(r)$ such that (i) $u$ is generated by the sweepline algorithm corresponding to ray $r$ after
$s_{i_N}$ and (ii) $u$ does not intersect $t_{s^*_r}$. 

We prove that there exists a natural number $M > 0$ such that each $u \in U$ belongs to
$int(E_M)$, where $E_M = (H(t_1), H(t_2), \ldots, H(t_M))$. Suppose, for the sake of contradiction, this is not true. Then, there exists an
infinite sequence $A = u_{i'_1}, u_{i'_2}, \ldots$ ($i'_1 < i'_2 < \cdots$) of open segments in $U$ 
such that for each $k \in \mathbb{N}$, $u_{i'_k} \in int(H(t_{i'_k}))$. 
Since all segments of sequence $A$ lie within the bounding box $\mathcal{B}$, this sequence
has a convergent subsequence $A' = u_{i'_{j_1}}, u_{i'_{j_2}}, \ldots$, where $1 \leq j_1 < j_2 < \cdots$. Let $u^*$ be the limit of this sequence. By definition, $u^* \in W(r)$.
If $u^* \neq s^*_r$, $|W(r)| > 1$, and we arrive at a contradiction. 

Now assume that $u^* = s^*_r$. There are two cases: (i) an infinite number of open segments in sequence $A'$ lie on the opposite side of $s^*_r$ to the open segments in sequence $s_{i_N}, s_{i_{N+1}}, \ldots$, or (ii) only a finite number of open segments in $A'$ lie on the opposite side of $s^*_r$ to the open segments in sequence $s_{i_N}, s_{i_{N+1}}, \ldots$. Case (i) is not possible; the proof is similar to the proof of impossibility of {\bf Case II} of tunnel lemma above. In Case (ii), since none of the open segments in sequence $A'$ intersects $t_{s^*_r}$, we conclude that $u^* \neq s^*_r$ (see Figure \ref{fig-case-ii}). Thus, $|W(r)| > 1$ and we again arrive at a contradiction.

Let $N'$ be a natural number such that $N' > N$ and for all $k \geq N'$, $i_k > M$.
Construct a new ray $r'$ by replacing the portion of ray $r$ generated after open segment
$s_{i_{N'}}$ by the horizontal sweep $H(t')$, where $t'$ is the portion of $t_{s^*_r}$ between
$s_{i_{N'}}$ and $s^*_r$. Thus, $r' \sim r$ and $r$ is a finite ray. Hence, the lemma is proved.
$\blacksquare$

\subsubsection{Interior JCT for rays}

{\bf Construction of $K_r$.} Let $\phi_{t_1}, \phi_{t_2}, \ldots$ be the one-to-one and continuous functions 
from $\mathbb{S}^1$ to $\mathbb{R}^2$, corresponding to Jordan curves $K_{t_1}, K_{t_2}, \ldots$, respectively (for each $i \in \mathbb{N}$, $K_{t_i}$ is the boundary of $H(t_i)$). Let $\phi_1 = \phi_{t_1}$, $\phi_2$, $\ldots$ be
the one-to-one and continuous functions from $\mathbb{S}^1$ to $\mathbb{R}^2$, corresponding
to Jordan curves $K_{t_1}$, $K_2$, $\ldots$. For each $j \geq 2$, $K_j$ is obtained
from $K_{j-1}$ by the extension of a piecewise-vertical Jordan curve (see the construction
of $K'$ in the previous section). Thus, for each $j \geq 2$, there exist intervals
$I_j, I'_j \in \mathbb{S}^1$ and a continuous bijection $\omega_j: I_j \rightarrow I'_j$ such that (i) for $z \in I_j$, $\phi_j(z) = \phi_{t_j}(\omega_j(z)))$ and (ii) for $z \notin I_j$,
$\phi_j(z) = \phi_{j-1}(z)$.

Further, by the definition of ray $r$, $I_1 \supset I_2 \supset I_3 \supset \cdots$. By Cantor's intersection theorem, the nested sequence of closed intervals $cl(I_1) \supset cl(I_2) \supset cl(I_3) \supset \cdots$
has a non-empty intersection. Let $I^* = \cap_{j=1}^{\infty} cl(I_j)$.

Note that $r$ is a terminating ray if and only if $I^*$ consists of a single point. 

If $r$ is a non-terminating ray, there is a finite ray $r'$ such that $r \sim r'$. A finite number of
applications of interior JCT for a piecewise-vertical Jordan curve (see Lemma \ref{interior-jct-horizontal-sweep}) established interior JCT for $r'$. Let $\phi_{r'}: \mathbb{S}^1 \rightarrow \mathbb{R}^2$ be the one-to-one and continuous function defining the boundary Jordan curve of $int(r')$. We define $\phi_r(z)$ as equal to $\phi_{r'}(z)$.

On the other hand, if $r$ is a terminating ray, $K_{r}$ (given by map $\phi_r$) is defined as the limit of $K_{t_1}, K_2, \ldots$ as follows. Take any point $z \in \mathbb{S}^1$. If $z \notin I^*$, there exists a natural number $N_z > 0$ such that $\phi_{i_1}(z) = \phi_{i_2}(z)$ for all $i_1, i_2 \geq N_z$. In this case, we define $\phi_r(z) = \phi_{N_z}(z)$. Now consider the case when $z \in I^*$. Since $r$ is a terminating ray, $I^*$ consists of a single point. We set $\phi(z) = s^*_r$, as the limiting segment $s^*_r$ consists of a single point.

We now prove the following:

\begin{observation}
\label{obs-phi-r}
$K_r$ is a Jordan curve.
\end{observation}

\begin{figure}
%\begin{center}
\begin{tikzpicture}
\node at (3,3) {\includegraphics[width=9cm]{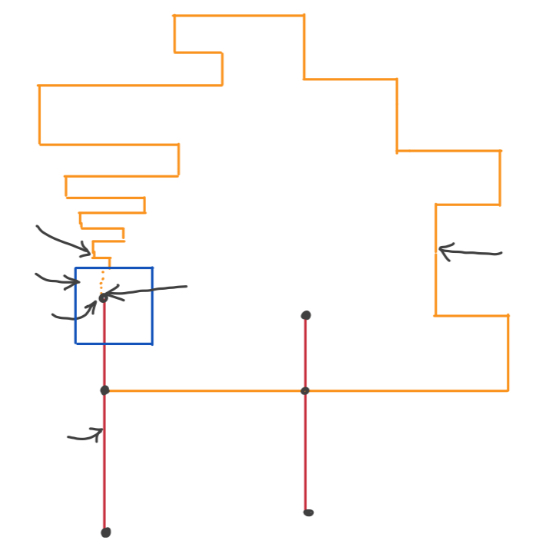}};
%\includegraphics[width=9cm]{fig-inf-rectilinear-polygon}
%\foreach \y in {-1, 0, 1, 2, ..., 7, 8, 9}
%{
%\draw (-2,\y) -- (7.5,\y);
%\node at (-0.5,\y){\y};
%};

%\foreach \x in {-2, -1, 0, 1,2,...,7}
%{
%\draw (\x,-1) -- (\x,9.5);
%\node at (\x,-0.5){\x};
%};
\node (A) at (6.9, 3.5) {$Q$};
\node (B) at (3.8, 1.3) {$x$};
\node (C) at (3.8, 2.5) {$a_x$};
\node (D) at (3.8, -1) {$b_x$};
\node (E) at (-0.4, 0.5) {$s_p$};
\node (F) at (-1.1, 3.1) {$S_{\epsilon}$};
\node (G) at (-1.2, 2.4) {$\phi_r(z_1)$};
\node (H) at (1.8, 3) {$\phi_r(z_2)$};
\node (I) at (-2.5, 4) {infinite number of turns};

\end{tikzpicture}

\caption{$\phi_r$ is a one-to-one function: $r$ is a terminating ray}
\label{phi-r-one-to-one}
%\end{center}
\end{figure}

\noindent {\bf Proof:} We consider the case when $r$ is a terminating ray, as the case of a
non-terminating ray follows from interior JCT for finite rays.

\begin{enumerate}
\item {\it $\phi_{r}$ is well-defined.} Follows from the definition of function $\phi_r$ above.

\item {\it $\phi_{r}$ is a continuous function.} Let $z \notin I^*$. The continuity of $\phi_r$ at
$z$ follows from the continuity of $\phi_{N'_z}$ at point $z$, where $N'_z$ is the smallest
natural number such that $cl(I_{N'_z}) \cap z = \phi$.

Now consider the case when $z \in I^*$ (note that for a terminating ray $r$, $I^*$ consists of a single point). Let $z_1, z_2, \ldots$ be a sequence of points on the unit circle $\mathbb{S}^1$ such that
$\lim_{i \rightarrow \infty} z_i = z$. Since $I_1 \supset I_2 \supset \cdots$ and $\cup_{j=1}^{\infty} cl(I_j) = I^*$, there exists an infinite subsequence $z_{i_1}, z_{i_2}, \ldots$ ($1 \leq i_1 < i_2 < \cdots$) such that
for all natural numbers $k$, $z_{w} \in I_k$ for all natural numbers $w \geq i_k$. In other words,
for all natural numbers $k$, $\phi(z_{w}) \in \cup_{j=k}^{\infty} cl(H(t_j))$ for all natural numbers $w \geq i_k$. Thus, there exists an infinite subsequence $z_{i_{j_1}}, z_{i_{j_2}}, \ldots$ ($j_1 < j_2 < \cdots$)
such that for each $k \in \mathbb{N}$, $z_{i_{j_k}} \in cl(H(t_{u_k}))$, where $u_1 < u_2 < \cdots$
Thus, $\lim_{k \rightarrow \infty} \phi_r(z_{i_{j_k}}) = s^*_r = \phi(z)$, as $z$ is the unique point in the limiting segment $s^*_r$ of ray $r$. 

Since the above argument works for every convergent subsequence of $\phi(z_1), \phi(z_2), \ldots$, we conclude that $\lim_{i \rightarrow \infty} \phi(z_i)$ exists and is equal to $s^*_r = \phi(z)$ (a bounded sequence converges if and only if every convergent subsequence of the bounded sequence converges to the same limit point).

\item {\it $\phi_{r}$ is a one-to-one function.} Suppose, for the sake of contradiction, that $\phi_r$ is
not a one-to-one function. Then, there exist two distinct points $z_1, z_2 \in \mathbb{S}^1$ such that
$\phi_r(z_1) = \phi_r(z_2)$. We consider the following cases:

\begin{enumerate}
	\item {\it Both $z_1, z_2 \notin I^*$.} Let $N''_{z_1, z_2}$ be the smallest natural number
	such that $\{z_1, z_2\} \cap cl(I_{N''_{z_1, z_2}}) = \phi$. Thus, $\phi_r(z_1) =
	\phi_{N''_{z_1, z_2}}(z_1) \neq \phi_{N''_{z_1, z_2}}(z_2) = \phi_r(z_2)$, a contradiction.
	
	\item {\it $z_1 \notin I^*$ and $z_2 \in I^*$, or vice versa.} (see Figure \ref{phi-r-one-to-one})
	Let $N'_{z_1}$ be a natural number such that $\phi_r(z_1) = \phi_{N'_{z_1}}(z_1)$. Then,
	$\phi_r(z_1)$ is an endpoint of open segment $s_p$, where $p \in H(t_k)$ and $1 \leq k \leq 
	N'_{z_1}$.
	
	Since $z_2 \in I^*$, there exists a rectilinear path $Q$ with {\it an infinite number of turns} from
	a point of $s_p$ to $\phi_r(z_2) = s^*_r$. Let $x$ be a point in the interior of the horizontal
	segment $u$ of $Q$ with one endpoint in $s_p$. 
	
	Let $S_{\epsilon}$ be an axis-parallel square of side length $\epsilon > 0$, with center at
	$s^*_r$, such that $\epsilon < \frac{d(x, s_p)}{2}$. Let $Q'_{\epsilon}$ be the subpath of $Q$
	from $p$ till the first point after $p$ at which $Q$ intersects $S_{\epsilon}$. Let $R_{\epsilon}$
	be the rectilinear polygon, with finite number of turns, formed by $s_p$, $Q'_{\epsilon}$ and
	$S_{\epsilon}$. By JCT for rectilinear polygons, $int(R)$ and $ext(R)$ are well-defined.
	
	Since $S_{\epsilon} \cap s_x = \phi$ and $Q'_{\epsilon}$ cannot contain two points of $s_x$ 
	(since $Q'_{\epsilon}$ corresponds to a finite path in ray $r$), one endpoint $a_x$ of 
	$s_x$ lies in $int(R_{\epsilon})$ and the other endpoint $b_x$ lies in $ext(R_{\epsilon})$. 
	Thus, the Jordan curve $J$ crosses $R_{\epsilon}$ two times at points $p_1^{\epsilon}$ and
	$p_2^{\epsilon}$. Both of these points belong to $S_{\epsilon}$, $p_1^{\epsilon}$ lies on the Jordan
	arc of $J$ from $a_x$ to $b_x$ and $p_2^{\epsilon}$ lies on the Jordan arc of $J$ from $b_x$
	to $a_x$.
	
	As $\epsilon \rightarrow 0$, $\lim_{\epsilon \rightarrow 0} p_1^{\epsilon} = \lim_{\epsilon \rightarrow 
	0} p_2^{\epsilon} = s^*_r$. This contradicts the assumption that $J$ is defined by a one-to-one and
	continuous function from $\mathbb{S}^1$ to $\mathbb{R}^2$.
	
	{\it Note.} We have derived a very weak version of JCT for infinite rectilinear polygons, using the
	squares $S_{\epsilon}$, $\epsilon > 0$. In this paper, these infinite rectilinear polygons have at most
	two limit points i.e., points at which an infinite sequence of diminishing turns converges.

	\item {\it Both $z_1, z_2 \in I^*$.} This case cannot occur since $I^*$ consists of a single point.
\end{enumerate}
\end{enumerate}

$\blacksquare$

We set $int(K_r) = \cup_{i=2}^{\infty} int(E_i)$, $ext(K_r) = \mathbb{R}^2 - \big( K_r \cup int(K_r) \big)$ and conclude the following

\begin{theorem}
\label{thm-int-JCT-ray}
{\bf Interior JCT for an infinite ray.} Let $r = H(t_1), H(t_2), \ldots$ be an infinite ray of horizontal sweeps. Then, there exists a piecewise-vertical Jordan curve $K_{r}$ such that

\begin{enumerate}
    \item $int(K_r)$ is a bounded, connected open set with $bd(int(K_r)) = K_r$ and $ext(K_r)$ is an unbounded open set and
    \item the set $int(K_r)$ is the union of a collection of open segments $s_p$, $p \in I_{K_r} \subset \mathbb{R}^2 - J$. Further, 
    for any two open segments $s_{p_1}$ and $s_{p_2}$ ($p_1, p_2 \in I_{K_r}$), there
    is a rectilinear path with a finite number of turns from a point $w_1 \in s_{p_1}$ to a point $w_2 \in s_{p_2}$.
\end{enumerate}
\end{theorem}

%% file: sweepline-jordan-001.tex
\subsection{Sweepline algorithm for a Jordan curve}
\label{sec-sweepline-algorithm}
We now define sweepline algorithms with a finite as well as countably infinite number of steps and prove interior JCT for them.

\begin{definition}
A sweepline algorithm with a countably infinite number of steps is a sequence $H(t_1), H(t_2), \ldots$ of horizontal sweeps such that

\begin{enumerate}
    \item $H(t_1)$ can be extended by $H(t_2)$. Let $E_2 = (H(t_1), H(t_2))$. Let $K_2$ be the piecewise-vertical Jordan curve bounding $int(E_2)$.

    \item $K_2$ can be extended by $H(t_3)$. Let $K_3$ be the piecewise-vertical Jordan curve bounding $E_3 = (H(t_1), H(t_2), H(t_3))$.
    
    \item $K_3$ can be extended by $H(t_4)$. Let $K_4$ be the piecewise-vertical Jordan curve bounding $(H(t_1), H(t_2), H(t_3), H(t_4))$.

    \item and so on till infinity.
\end{enumerate}
\end{definition}

{\it Note.} The existence of $K_2, K_3, \ldots$ as well as the truth of interior JCT for them, follows inductively from Lemma \ref{JCT-horizontal-sweep} (base case on $H(t_1)$) and Lemma \ref{piecewise-vertical-JCT} (the inductive step).\\

{\bf Construction of recursion tree for the sweepline algorithm $\mathcal{S}$.} We now describe
the recursion tree $T(\mathcal{S})$ for the sweepline algorithm $\mathcal{S}$. Each node of
$T(\mathcal{S})$ corresponds to a unique horizontal sweep of $\mathcal{S}$. The root of 
$T(\mathcal{S})$ is the initial horizontal sweep $H(t_1)$. For natural numbers $i, j$ ($i < j$), the node corresponding to the horizontal sweep $H(t_j)$ is the child of the node corresponding 
to the horizontal sweep $H(t_i)$ if and only if the piecewise-vertical Jordan curve $K_{j-1}$ is extended by $H(t_j)$ from the portion of a vertical segment belonging to $bd(H(t_i))$.
Although $\mathcal{S}$ has countably infinite horizontal sweeps, it can have 
countably infinite non-terminating rays and an uncountable number of terminating rays.
Let $\rho(\mathcal{S})$ be the set of all maximal rays starting from the root of the recursion tree $T(\mathcal{S})$ (a ray $r$ is maximal if and only if there is no ray $r'$ such that $r$ is a prefix
of $r'$).

{\bf Construction of $K_{\mathcal{S}}$.} Let $\phi_1 = \phi_{t_1}$, $\phi_2$, $\ldots$ be
the one-to-one and continuous functions from $\mathbb{S}^1$ to $\mathbb{R}^2$, corresponding
to Jordan curves $K_{t_1}$, $K_2$, $\ldots$. For each $j \geq 2$, $K_j$ is obtained
from $K_{j-1}$ by the extension of a piecewise-vertical Jordan curve (see the construction
of $K'$ in the previous section). Thus, for each $j \geq 2$, there exists intervals
$I_j, I'_j \in \mathbb{S}^1$ and a continuous bijection $\omega_j: I_j \rightarrow I'_j$ such that (i) for $z \in I_j$, $\phi_j(z) = \phi_{t_j}(\omega_j(z))$ and (ii) for $z \notin I_j$,
$\phi_j(z) = \phi_{j-1}(z)$.

We first make the following observation, which follows from the definition of extension of a piecewise-vertical Jordan curve:

\begin{observation}
For each $j \geq 2$, $I_j \subset I_k$, for some $1 \leq k < j$. 
\end{observation}

$K_{\mathcal{S}}$ (given by $\phi_{\mathcal{S}}$) is defined as the limit of $K_1, K_2, \ldots$, using the rays in $\rho(\mathcal{S})$, as follows. Let $z$ be any point in $\mathbb{S}^1$. There are two cases:

\begin{enumerate}
\item {\bf Case I.} There exists a natural number $N_z > 0$ such that for all natural numbers $i_1, i_2 \geq N_z$, $\phi_{i_1}(z) = \phi_{i_2}(z)$. In this case, we set $\phi_{\mathcal{S}}(z) = \phi_{N_z}(z)$.

\item {\bf Case II.} There does not exist a natural number $N_z > 0$ such that for all natural numbers $i_1, i_2 \geq N_z$, $\phi_{i_1}(z) = \phi_{i_2}(z)$. In this case, there exists a 
ray $r \in \rho(\mathcal{S})$ such that (i) $r = (H(t_1), H(t_{j_2}), H(t_{j_3}), \ldots)$ ($1 < j_2 < j_3 < \cdots$) and (ii) $z \in I^*_r$ where $I_{j_1} \supset I_{j_2} \supset \cdots$ and $I^*_r = \cap_{k=1}^{\infty} cl(I_{j_k})$ (we take $j_1 = 1$).

If $r$ is a terminating ray, we define $\phi_{\mathcal{S}}(z) = s^*_r$,
as the limiting segment of $r$ consists of a single point of the Jordan curve $J$.
Now consider the case when $r$ is a non-terminating ray. Let $a^*_r$ and $b^*_r$ be the two endpoints of interval $I^*_r$ and let $a_r$ and $b_r$ be the two endpoints of $s^*_r$. Without loss of generality, assume that $\lim_{i \rightarrow \infty} \phi_i(a^*_r) = a_r$ and $\lim_{i \rightarrow \infty} \phi_i(b^*_r) = b_r$.
We define $\phi_{\mathcal{S}}$ so that it maps the closed interval $I^*_r$ to the limiting segment $cl(s^*_r)$ using any linear function between these two sets, such
that $a^*_r$ is mapped to $a_r$ and $b^*_r$ is mapped to $b_r$.

\end{enumerate}

We now prove the following:

\begin{observation}
\label{obs-phi-S}
(i) $\phi_{\mathcal{S}}$ is well-defined, (ii) $\phi_{\mathcal{S}}$ is a continuous function, and
(iii) $\phi_{\mathcal{S}}$ is a one-to-one function. Thus, $K_{\mathcal{S}}$ is a Jordan curve.
\end{observation}

\begin{figure}
%\begin{center}
\begin{tikzpicture}
\node at (3,3) {\includegraphics[width=9cm]{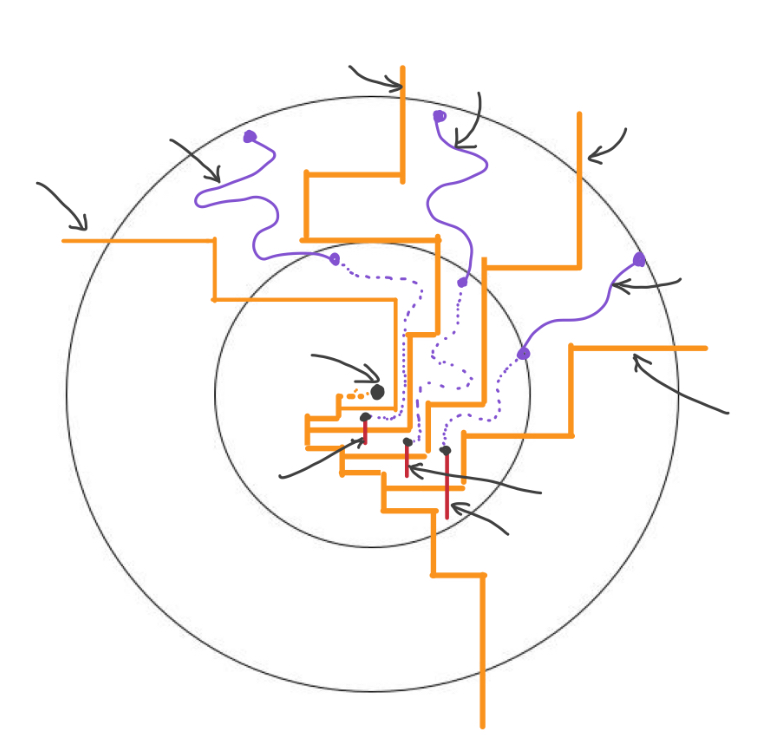}};
%\includegraphics[width=9cm]{fig-inf-rectilinear-polygon}
%\foreach \y in {-1, 0, 1, 2, ..., 7, 8, 9}
%{
%\draw (-2,\y) -- (7.5,\y);
%\node at (-0.5,\y){\y};
%};

%\foreach \x in {-2, -1, 0, 1,2,...,7}
%{
%\draw (\x,-1) -- (\x,9.5);
%\node at (\x,-0.5){\x};
%};
\node (A) at (7.3, 2.5) {$Q_{i_N}$};
\node (B) at (5.8, 6.1) {$Q_{i_{N+1}}$};
\node (C) at (2.3, 6.8) {$Q_{i_{N+2}}$};
\node (D) at (-1.2, 5.5) {$Q_{i_{N+3}}$};
%\node (E) at (-0.4, 0.5) {$Q_{i_{N+4}}$};
\node (F) at (6.5, 4.4) {$\gamma_{i_{N}}$};
\node (G) at (4.2, 6.5) {$\gamma_{i_{N+1}}$};
\node (H) at (0.4, 5.8) {$\gamma_{i_{N+2}}$};

\node (I) at (1.8, 3.2) {$s^*_r$};
\node (J) at (4.6, 1) {$s_{p_N}$};
\node (K) at (5.2, 1.5) {$s_{p_{N+1}}$};
\node (L) at (1.5, 1.7) {$s_{p_{N+2}}$};
\node (M) at (0.7, 3) {$C_1$};
\node (N) at (-1, 3) {$C_2$};
%\node (O) at (4.6, -1) {$Q$};
%\node (I) at (1.8, 3) {$s^*_r$};
%\node (I) at (1.8, 3) {$s^*_r$};

\end{tikzpicture}
\caption{$\phi_S$ is a continuous function. $Q_{i_N}, Q_{i_{N+1}}, \ldots$ are rectilinear paths traced by the sweeps of rays $r_{i_N}, r_{i_{N+1}}, \ldots$ respectively, once they separate from ray $r$. $s_{p_N}, s_{p_{N+1}}, \ldots$ are open segments such that one of their endpoints lies in the different regions in which $Q_{i_N}, Q_{i_{N+1}}, \ldots$ partition the
circle $C_2$.}
\label{fig-phi-S-cont}
%\end{center}
\end{figure}
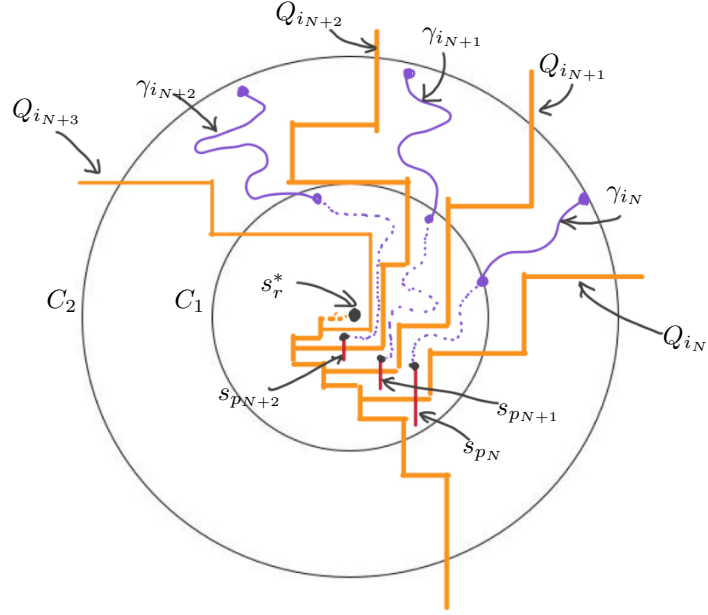

\noindent {\bf Proof:} (i) follows from the definition of $\phi_{\mathcal{S}}$ above.

We now prove (iii) i.e., that $\phi_{\mathcal{S}}$ is one-to-one. Let $\rho(T(\mathcal{S}))$ denote the
set of all maximal rays of the recursion tree $T(\mathcal{S})$, starting from its root. 
Let $A$ be the set of all intervals $I^*_r \subset \mathbb{S}^1$, where $r \in \rho(T(\mathcal{S}))$.
Let $B$ be the set of all limiting segments $s^*_r \subset \mathbb{R}^2$, where $r \in \rho(T(\mathcal{S}))$.

We first prove that no two intervals in $A$ intersect. Let $T_d(\mathcal{S})$ denote the
subtree of $T(\mathcal{S})$ formed by all nodes at depth at most $d$. As this property is true for all
rays in $\rho(T_d(\mathcal{S}))$ and the intervals for rays in $\rho(T_{d+1}(\mathcal{S})$ form
a refinement of the intervals for rays in $\rho(T_d(\mathcal{S}))$, we derive the non-intersection property
of intervals in $A$.

We next prove that no two limiting segments in $B$ intersect.
Suppose there exist two limiting segments $s^*_{r_1}$ and $s^*_{r_2}$ that have a common point. First, assume both limiting segments correspond to non-terminating rays. By the previous section, there exist finite
rays $r'_1 \sim r_1$ and $r'_2 \sim r_2$. Consider a sweepline algorithm $\mathcal{S}'$ with finite number of horizontal
sweeps over $r'_1 \cup r'_2$. We conclude that $\phi_{\mathcal{S}'}$ is not one-to-one, which contradicts
the interior JCT for extension of a piecewise-vertical Jordan curve proved above.
Now, assume that both limiting segments correspond to terminating rays. Then, there exists a rectilinear
path $Q_1 \subset \mathbb{R}^2 - J$ with countably infinite turns from $p^*$ to $s^*_{r_1}$ and a rectilinear
path $Q_2 \subset \mathbb{R}^2 - J$ with countably infinite turns from $p^*$ to $s^*_{r_2}$. Let $R$ 
be the infinite rectilinear polygon obtained by combining $Q_1$ and $Q_2$. The proof now proceeds in a
manner similar to that for extension of a horizontal sweep. There exists a point $x \in R$ such that
one endpoint $a$ of $s_x$ is in $int(R)$ and the other endpoint $b$ of $s_x$ is in $ext(R)$ (the interior
and exterior exist because of the truth of weak JCT for infinite rectilinear polygons; $R$ has at most two limit points). Then, $J$ must cross
$R$ at least two times. We arrive at a contradiction since $J$ can cross $R$ only at point $s^*_{r_1}$ (note that $s^*_{r_1} = s^*_{r_2}$). 

The case where one limiting segment corresponds to a terminating ray and the other limiting segment
corresponds to a non-terminating ray can be resolved in a similar manner.

Now, suppose for the sake of contradiction, that there exist two distinct points $z_1, z_2 \in \mathbb{S}^1$ such that $\phi_{\mathcal{S}}(z_1) = \phi_{\mathcal{S}}(z_2)$. If both $z_1, z_2 \in A$, the one-to-one
property follows from non-intersection property of $A$ and $B$ and the fact that $\phi_{S}$ maps unique
intervals in $A$ to unique limiting segments in $B$. If both $z_1, z_2 \notin A$, both $z_1$ and $z_2$ occur
at finite depth in the recursion tree $T(\mathcal{S})$. Let $d_{max}$ be the maximum of the depths of
$z_1$ and $z_2$. The one-to-one property follows by interior JCT for the finite sweepline algorithm given
by $T_{d_{max}}(\mathcal{S})$. If $z_1 \in A$ and $z_2 \notin A$ (or vice versa), one-to-one property can be
proved by the application of weak JCT for infinite rectilinear polygons (as used above to prove that $B$ is non-intersecting).

We now prove property (ii) i.e., that $\phi_{\mathcal{S}}$ is a continuous function. We describe the central case here; other cases are simpler. Let $r$ be a non-terminating ray. Let $z_1, z_2, \ldots$ be a sequence of points on the unit circle $\mathbb{S}^1$ such that (i)
for each $i \in \mathbb{N}$, $\phi(z_i)$ belongs to ray $r_i$, and (ii) $\lim_{i \rightarrow \infty} z_i = I^*_r$.
We prove that $\lim_{i \rightarrow \infty} \phi(z_i) = s^*_r$. 

Suppose this is not the case. Let $d_i$ be the maximum depth till which ray $r$
and ray $r_i$ have common prefix from the root.
We next assume that $\lim_{i \rightarrow \infty} d_i = \infty$ (the case where the sequence $d_i$, $i \in \mathbb{N}$, is bounded is simpler).
Thus, there exists a subsequence $z_{i_1}, z_{i_2}, \ldots$ ($i_1 < i_2 < \cdots$) such that
(i) $d_{i_1} < d_{i_2} < \cdots$,  (ii) $\lim_{k \rightarrow \infty} \phi(z_{i_k})$ exists and is not equal to
$s^*_r$. Let $z^* =  \lim_{k \rightarrow \infty} \phi(z_{i_k})$. Let $\epsilon = d(s^*_r, z^*)$.

Let $N$ be a natural number such that all open segments generated by horizontal sweeps of ray
$r$ at depth at least $N$ are within distance $\frac{\epsilon}{4}$ of $s^*_r$ and for each
$k \geq N$, $d(z^*, \phi(z_{i_k})) \leq \frac{\epsilon}{4}$. Let $C_1$ and $C_2$ be circles with center
at $s^*_r$ and radii
$\frac{\epsilon}{4}$ and $\frac{\epsilon}{2}$ respectively.

The resulting scenario is depicted in Figure \ref{fig-phi-S-cont} (we use the Jordan curve theorem for 
rectilinear polygons; even if ray $r_{i_k}$, $k \geq N$, is terminating ray, the portion within the
circle of radius $\frac{\epsilon}{2}$ will have a finite number of turns). Thus, we are able to obtain
disjoint Jordan arcs $\gamma_{i_N}, \gamma_{i_{N+1}}, \ldots$ of $J$ such that one of their
two endpoints converges to $s^*_r$ and the other endpoints lie on the outer circle $C_2$, as depicted in Figure \ref{fig-phi-S-cont}.
We can find a subsequence of these Jordan arcs such that the other endpoint also converges. 
Applying Observation \ref{obs-jordan-arc-limit} leads to a contradiction.
The case where $r$ is a finite ray or a non-terminating ray can be resolved using similar ideas.
$\blacksquare$

Since the intervals $I^*_r$, $r \in \rho(\mathcal{S})$ are non-intersecting, we conclude that

\begin{observation}
An infinite tree of horizontal sweeps has a countable number of non-terminating rays.
\end{observation}

We set $int(K_{\mathcal{S}}) = \cup_{i=2}^{\infty} int(E_i)$, $ext(K_{\mathcal{S}}) = \mathbb{R}^2 - \big( K_{\mathcal{S}} \cup int(K_{\mathcal{S}}) \big)$ and conclude the following

\begin{theorem}
\label{thm-int-JCT-sweepline}
{\bf Interior JCT for sweepline algorithm with countably infinite steps.} Let $\mathcal{S} = H(t_1), H(t_2), \ldots$ be a sweepline algorithm with a countably infinite number of horizontal sweeps. 
Then, there exists a piecewise-vertical Jordan curve $K_{\mathcal{S}}$ such that

\begin{enumerate}
\item $int(K_{\mathcal{S}})$ is a bounded, connected open set with $bd(int(K_{\mathcal{S}})) = K_{\mathcal{S}}$ and $ext(K_{\mathcal{S}})$ is an unbounded open set and
    
\item the set $int(K_{\mathcal{S}})$ is the union of a collection of open segments $s_p$, $p \in I_{K_{\mathcal{S}}} \subset \mathbb{R}^2 - J$. Further, for any two open segments $s_{p_1}$ and $s_{p_2}$ ($p_1, p_2 \in I_{K_\mathcal{S}}$), there is a rectilinear path with a finite number of turns from a point $w_1 \in s_{p_1}$ to a point $w_2 \in s_{p_2}$.
\end{enumerate}
\end{theorem}

%% file: zorn-lemma-001.tex
\subsection{Applying Zorn's lemma: a partial order on countably infinite trees}
\label{sec-zorn-lemma}

\begin{definition}
{\bf Equivalence of two sweepline algorithms.} Let $\mathcal{S} = H(t_1), H(t_2), \ldots$
and $\mathcal{S}' = H(t'_1), H(t'_2), \ldots$ be two sweepline algorithms with countably
infinite horizontal sweeps. Then, $\mathcal{S} \sim \mathcal{S}'$ if and only if 
$int(K_{\mathcal{S}}) = int(K_{\mathcal{S}'})$.
\end{definition}

One can verify that $\sim$ is an equivalence relation. Let $\mathcal{T}$ be the set of equivalence classes of sweepline algorithms (with countably infinite steps), according to 
relation $\sim$.

\begin{lemma}
\label{lem-canonical}
{\bf Canonical form of a sweepline algorithm.} Let $\mathcal{S}$ be a sweepline algorithm. Then,
there exists a sweepline algorithm $\mathcal{S}'$ such that 

\begin{enumerate}
	\item $\mathcal{S'}$ is equivalent to $\mathcal{S}$ i.e., $\mathcal{S} \sim \mathcal{S}'$, and
	\item all vertical segments of $K_{\mathcal{S}'}$ are generated by nodes (horizontal sweeps) at finite depth in the recursion tree $T(\mathcal{S}')$. (Equivalently, there are no
	non-terminating infinite rays starting from the root of $T(\mathcal{S}')$.)
\end{enumerate}

We call $\mathcal{S}'$ a canonical form of sweepline algorithm $\mathcal{S}$.

{\it Note.} There can be more than one canonical forms of the sweepline algorithm $\mathcal{S}$.
\end{lemma}

\noindent {\bf Proof:} Let $T(\mathcal{S})$ be the recursion tree of $\mathcal{S}$. $T(\mathcal{S})$ has at most a 
countable number of non-terminating rays. Let $r_1, r_2, \ldots$ be an enumeration of the non-terminating infinite rays of $T(\mathcal{S})$. Let $T'$ be the subtree of $T(\mathcal{S})$ obtained by taking the union of rays $r_1, r_2, \ldots$.

We now construct a recursion tree $T''$, equivalent to $T'$, as follows. Initialize $T''$ to the single horizontal sweep corresponding to root of $T'$. For each $i=1,2, \ldots$, we repeat the following steps. Suppose, if we run the sweepline algorithm corresponding to ray $r_i$, it crosses the current boundary $bd(T'')$ at vertical segment $u_i$. Let $v_i$ be the node
of $T''$ which generates $u_i$. Let $r'_i$ be a finite ray which sweeps the same region as the region swept by ray $r_i$ {\it after it crosses $bd(T'')$ at $u_i$} (since $r_i$ is a non-terminating ray, such a finite ray $r'_i$ must exist). We add ray $r'_i$ as a child of node $v_i$.

We construct a final recursion tree $T_{final}$, from recursion tree $T''$ constructed above.
When we remove all nodes of $T'$ from $T(\mathcal{S})$, we are left with a countable number of subtrees
$T_1, T_2, \ldots$. For each $i \in \mathbb{N}$, let $w_i$ be the root node of $T_i$. Suppose
we execute the sweepline algorithm $\mathcal{S}$ again.
For each $i=1, 2, \ldots$, let $u'_i$ be the vertical segment of $bd(T'')$ such that $\mathcal{S}$ executes the horizontal sweep corresponding to node $w_i$, immediately after crossing $u'_i$.

To construct $T_{final}$, for each $i \in \mathbb{N}$, we attach $T_i$ as a child of node $v'_i$, where $v'_i$ is the node of $T''$ which generates $u'_i$.
The sweepline algorithm $\mathcal{S}'$ corresponding to $T_{final}$ is the canonical form of $\mathcal{S}$.
$\blacksquare$

We now define a partial order on $\mathcal{T}$:

\begin{definition}
{\bf The partial order $<$.} Let $\mathcal{S}_1$ and $\mathcal{S}_2$ be two sweepline algorithms.
Then, $\mathcal{S}_1 < \mathcal{S}_2$ if and only if {\bf for every} sweepline algorithm $\mathcal{S}'_1$ such that $\mathcal{S}'_1$ is the canonical form of $\mathcal{S}_1$, there exists a sweepline algorithm $\mathcal{S}'_2$ such that

\begin{enumerate} 
\item $\mathcal{S}'_2$ is the canonical form of $\mathcal{S}_2$, and
\item the recursion tree $T(\mathcal{S}'_1)$ is a {\it proper} subgraph of the recursion tree 
$T(\mathcal{S}'_2)$ and the root of $T(\mathcal{S}'_1)$ is the same as the root of $T(\mathcal{S}'_2)$.

\end{enumerate}
\end{definition}

One can prove that $<$ is a partial order.

\begin{observation}
{\bf Properties of $<$.}
\begin{enumerate}
\item If $\mathcal{S}_1 \sim \mathcal{S}'_1$, $\mathcal{S}_2 \sim \mathcal{S}'_2$ and $\mathcal{S}_1 < \mathcal{S}_2$, then
$\mathcal{S}'_1 < \mathcal{S}'_2$.
\item {\bf Anti-symmetry.} If $\mathcal{S}_1 < \mathcal{S}_2$ and $\mathcal{S}_2 < \mathcal{S}_1$, then $\mathcal{S}_1 \sim \mathcal{S}_2$.
\item {\bf Transitivity.} If $\mathcal{S}_1 < \mathcal{S}_2$ and $\mathcal{S}_2 < \mathcal{S}_3$, then
$\mathcal{S}_1 < \mathcal{S}_3$.
\end{enumerate}

\end{observation}

\noindent {\bf Proof:} $(1)$ Follows from definition of $<$.

$(2)$ Since $\mathcal{S}_1 < \mathcal{S}_2$, $int(K_{\mathcal{S}_1}) \subsetneq int(K_{\mathcal{S}_2})$. Further, since $\mathcal{S}_2 < \mathcal{S}_1$, $int(K_{\mathcal{S}_2}) \subsetneq int(K_{\mathcal{S}_1})$. We arrive at a contradiction.

$(3)$ Let $\mathcal{S}'_1$ be a canonical form of $\mathcal{S}_1$. Since $\mathcal{S}_1 < \mathcal{S}_2$, there exists a canonical form $\mathcal{S}'_2$ of $\mathcal{S}_2$ such that $T(\mathcal{S}'_1)$
is a proper subgraph of $T(\mathcal{S}'_2)$ (with the same root node). Since $\mathcal{S}_2 < \mathcal{S}_3$, there exists a canonical form $\mathcal{S}'_3$ of $\mathcal{S}_3$ such that $T(\mathcal{S}'_2)$
is a proper subgraph of $T(\mathcal{S}'_3)$ (with the same root node). Thus, $T(\mathcal{S}'_1)$ is a proper subgraph of $T(\mathcal{S}'_3)$ and hence $\mathcal{S}_1 < \mathcal{S}_3$.
$\blacksquare$

\begin{lemma}
Every chain $\mathcal{C}$ in $(\mathcal{T}, <)$ has a maximal element.
\end{lemma}

\noindent {\bf Proof:} We assume that $\mathcal{C}$ has an infinite number of sweepline algorithms.
The case where $\mathcal{C}$ has a finite number of sweepline algorithms, is a special case.

Note that the interior of any horizontal sweep $H(t)$, where the horizontal segment $t$ has positive length, contains
an open ball of positive radius. Let $\mu_{*}(int(H(t))$ be the inner measure of $int(H(t))$ (see \cite{kolmogorov-fomin} for a definition of inner measure). Thus, $\mu_{*}(int(H(t))$ is a positive real number. 

Let $\zeta: \mathcal{T} \rightarrow \mathbb{R}$ be the map which maps a sweepline algorithm $\mathcal{S} \in \mathcal{T}$ to its inner measure $\mu_{*}(\mathcal{S})$. Since all sweepline algorithms in $\mathcal{T}$ consist
of only finite open segments, they stay within the bounding box $\mathcal{B}$, and hence $\mu_{*}(\mathcal{S}) < \mu(\mathcal{B})$, where $\mu$ is the Lebesgue measure. Thus, the range of $\zeta$ lies within the finite open interval
$I = (0, \mu(\mathcal{B}))$.

Further, if $\mathcal{S}_1 < \mathcal{S}_2$, let $H(t)$ be an additional horizontal sweep in $\mathcal{S}_2$ that is
not in $\mathcal{S}_1$. Then, $\mu_{*}(\mathcal{S}_2) \geq \mu_{*}(\mathcal{S}_1) + \mu_{*}(int(H(t)) > \mu_{*}(\mathcal{S}_1)$.

Thus, $\zeta$ is an increasing function from chain $\mathcal{C}$ to interval $I$.
Thus, there exists a countably infinite chain $\mathcal{S}_1 < \mathcal{S}_2 < \cdots$ of sweepline algorithms from
chain $\mathcal{C}$ such that for every sweepline algorithm $\mathcal{S} \in \mathcal{C}$, there exists
a natural number $i$ such that $\mathcal{S}_i \leq \mathcal{S} < \mathcal{S}_{i+1}$. (Note that this property is not true for chains of any arbitrary partial order over an uncountable set, see for example, Aronszajn line \cite{aronszajn}.)

The maximal element $\mathcal{S}^*$ of chain $\mathcal{C}$ is constructed inductively as follows:

\begin{enumerate}
\item Initialize $T(\mathcal{S}^*)$ to the recursion tree $T(\mathcal{S}'_1)$, where $\mathcal{S}'_1$ is any canonical form of $\mathcal{S}_1$.

\item For $i = 2, 3, \ldots$:
	\begin{enumerate}
		\item Since $\mathcal{S}_{i-1} < \mathcal{S}_i$ and current $\mathcal{S}^*$ is a canonical 
		form of 
		$\mathcal{S}_{i-1}$, there exists a canonical form $\mathcal{S}'_i$ of $\mathcal{S}_i$ such 
		that $T(\mathcal{S}^*)$ is a proper subgraph of $T(\mathcal{S}'_i)$.
		\item Set $T(\mathcal{S}^*)$ to $T(\mathcal{S}'_i)$.
		
	\end{enumerate}
\end{enumerate}

Since the infinite loop will run a countable number of times, the final recursion tree $T(\mathcal{S}^*)$
will have countably infinite nodes (or, horizontal sweeps). Hence, $\mathcal{S}^*$ is also a
sweepline algorithm. 

Let $\mathcal{S} \in \mathcal{C}$, be any sweepline algorithm in chain $\mathcal{C}$. By the definition of $\mathcal{S}_i$'s, $i \in \mathbb{N}$, there exists a natural number $i_0$ such that
$\mathcal{S}_{i_0} \leq \mathcal{S} < \mathcal{S}_{i_0+1}$. 
Run the above algorithm, with the sequence $\mathcal{S}_1 < \mathcal{S}_2 < \cdots < \mathcal{S}_{i_0} \leq \mathcal{S} < \mathcal{S}_{i_0+1} < \mathcal{S}_{i_0+2} < \cdots$. Let $\mathcal{S}'$ be the final sweepline algorithm obtained at the end of the algorithm. By construction, $\mathcal{S} < \mathcal{S}'$. Since $\mathcal{S}' \sim \mathcal{S}^*$, we conclude that $\mathcal{S} < \mathcal{S}^*$. 

Hence, $\mathcal{S} < \mathcal{S}^*$, for each $\mathcal{S} \in \mathcal{C}$. Thus, 
$\mathcal{S}^*$ is a maximal element of chain $\mathcal{C}$ in $\mathcal{T}$.
$\blacksquare$

\begin{lemma}
There is a maximal element $\mathcal{S}_{max}$ in $\mathcal{T}$.
\end{lemma}

\noindent {\bf Proof:} Apply Zorn's lemma \cite{kolmogorov-fomin} on $\mathcal{T}$. $\blacksquare$

\begin{lemma}
$int(K_{\mathcal{S}_{max}})$ is a bounded, connected open set with $bd(K_{\mathcal{S}_{max}}) = J$. 
\end{lemma}

\noindent {\bf Proof:} Suppose, for the sake of contradiction, that $bd(K_{\mathcal{S}_{max}}) \neq J$.
Then, $T(\mathcal{S}_{max})$ has at least one non-terminating ray $r$. Let $I^*_r$ be the limiting
segment of ray $r$. Let $\mathcal{S}'$ be any canonical form of the sweepline algorithm
$\mathcal{S}_{max}$. Then, $I^*_r$ is a vertical segment of $K_{\mathcal{S}'} = K_{\mathcal{S}_{max}}$. Further, $I^*_r$ is the limiting segment of a {\bf finite} ray $r'$ of $\mathcal{S}'$. 

We can extend $\mathcal{S}'$ by a horizontal sweep $H(t)$, such that one endpoint of $t$
is in $int(I^*_r)$ and the rest of the segment $t$ lies in $ext(K_{\mathcal{S}'})$. 

Let $\mathcal{S}''$ be the extended sweepline algorithm.
Then, by definition, $\mathcal{S}' < \mathcal{S}''$. We conclude that $\mathcal{S}_{max} < \mathcal{S}''$. This contradicts the maximality
of $\mathcal{S}_{max}$, and we arrive at a contradiction.

Hence, the statement of the lemma is true.
$\blacksquare$

%% file: root-segment-001.tex
\subsection{Finding the root segment $s_{p^*}$}
\label{sec-root-segment}

Let $p_{min}$ be a point with minimum $x$-coordinate among all points of $J$. Let $p_{max}$ be a point with maximum $x$-coordinate
among all points of $J$. Since $J$ is a closed set, $p_{min}$ and
$p_{max}$ exist. Further, $p_{min}$ has $x$-coordinate strictly
less than $p_{max}$. For otherwise, $J$ will be a subset of the vertical line through $p_{min}$, a contradiction.

Let $l$ be a line in the middle of the vertical lines through $p_{min}$ and $p_{max}$. Let $J_1$ be the arc of $J$ from $p_{min}$
to $p_{max}$ and let $J_2$ be the remaining arc of $J$ from $p_{max}$
to $p_{min}$.

\begin{lemma}
\label{lem-root-seg}
There exists a point $p^* \in l - l \cap J$ such that the upper and lower endpoints of $s_{p^*}$ belong to different arcs in the set $\{ J_1, J_2\}$.
\end{lemma}

\noindent {\bf Proof:} Note that $l \cap J$ is a closed set. Let $M_1$ and $M_2$ be the points with highest and lowest
$y$-coordinates in $l \cap J$. Since $l \cap J$ is a closed
set and $J$ lies completely within the bounding box $\mathcal{B}$,
$M_1$ and $M_2$ exist and $M_1 \neq M_2$.

If $M_1$ and $M_2$ belong to different arcs in $\{J_1, J_2\}$, we
set $x_1 = M_1$ and $y_1=M_2$. Suppose both belong to arc $J_1$.
Then, since $J_2$ intersects line $l$ (by Jordan curve theorem
for a vertical line), $\overline{x_1y_1} \cap J_2 \neq \phi$.
Further, since $d(J_1 \cap l, J_2 \cap l) > 0$, there is a point
$w \in J_2$ in the interior of $\overline{x_1y_1}$. We set $x_1 = M_1$ and $y_1 = w$. The case when both $M_1$ and $M_2$ belong to arc $J_2$
is symmetric.

By the above construction, $x_1$ and $y_1$ belong to different arcs in $\{J_1, J_2\}$. Let $\overline{x'_1y'_1}$ be
the middle half of segment $\overline{x_1y_1}$. 
Let $p_1$ be any point in $\overline{x_1y_1} - J$. Since $J$ has no vertical line segments, such a point must exist. Consider the open segment $s_{p_1}$. Let $a_1$ (upper) and $b_1$ (lower) be its two endpoints. If $a_1$ and $b_1$ belong to different arcs in $\{J_1, J_2\}$, we are done. Now there are two cases:

\begin{enumerate}
    \item Both $a_1, b_1 \in J_1$. If $x_1 \in J_2$, set $x_2 = x_1$
    and $y_2 = a_1$. If $y_1 \in J_2$, set $x_2 = b_1$ and $y_2 = y_1$. Note that $x_2$ and $y_2$ belong to different arcs.

    \item Both $a_1, b_2 \in J_2$. This case is symmetric to the above case.
\end{enumerate}

Repeating the above process, we get a sequence of segments
$\overline{x_1y_1}, \overline{x_2y_2}, \overline{x_3y_3}, \ldots$
such that for each $i \in \mathbb{N}$, (i) $x_i$ and $y_i$ belong to different arcs in $\{J_1, J_2\}$ and (ii) length of $\overline{x_{i+1}y_{i+1}}$ is at most $\frac{3}{4}$ times the length of $\overline{x_iy_i}$. 

Thus, there are arbitrarily close pairs of points on line $l$, belonging to two different arcs $J_1$ and $J_2$. We arrive at a contradiction, and hence the lemma is proved.
$\blacksquare$

\begin{lemma}
\label{lem-correct-root-seg}
Let $Q$ be any rectilinear path with a finite number of turns starting at $p^*$ such that $Q \subset \mathbb{R}^2 - J$. Let $q$ be the other endpoint of $Q$. Then, the open segment $s_q$ is of finite length.
\end{lemma}
\begin{figure}
%\begin{center}
\begin{tikzpicture}
\node at (3,3) {\includegraphics[width=9cm]{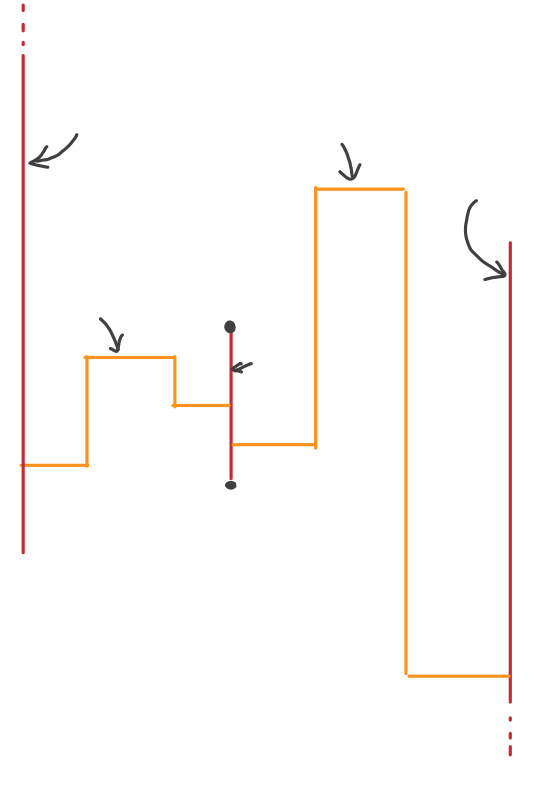}};
%\foreach \y in {-1, 0, 1, 2, ..., 7, 8, 9}
%{
%\draw (-2,\y) -- (7.5,\y);
%\node at (-0.5,\y){\y};
%};

%\foreach \x in {-2, -1, 0, 1,2,...,7}
%{
%\draw (\x,-1) -- (\x,9.5);
%\node at (\x,-0.5){\x};
%};
\node (A) at (2.8, 3.5) {$s_{p^*}$};
\node (B) at (6, 6.3) {$s_{q_2} \in R$};
\node (C) at (-0.1, 7.5) {$s_{q_1} \in L$};
\node (D) at (-0.1, 4.4) {$Q_1$};
\node (E) at (4, 7.4) {$Q_2$};
\node (F) at (2.2, 4.4) {$a \in J_1$};
\node (G) at (2.2, 1.3) {$b \in J_2$};
\node (H) at (-0.8, 9.7) {$+\infty$};
\node (I) at (6.6, -2.9) {$-\infty$};

\end{tikzpicture}
\caption{Ruling out Case $1$. The arc of $J$ from $a$ to $b$ must cross the rectilinear
polygon formed by $s_{p^*}$, $Q_1$, $Q_2$, and the two infinite open segments $s_{q_1}, s_{q_2}$.}
\label{fig-case-1}
%\end{center}
\end{figure}
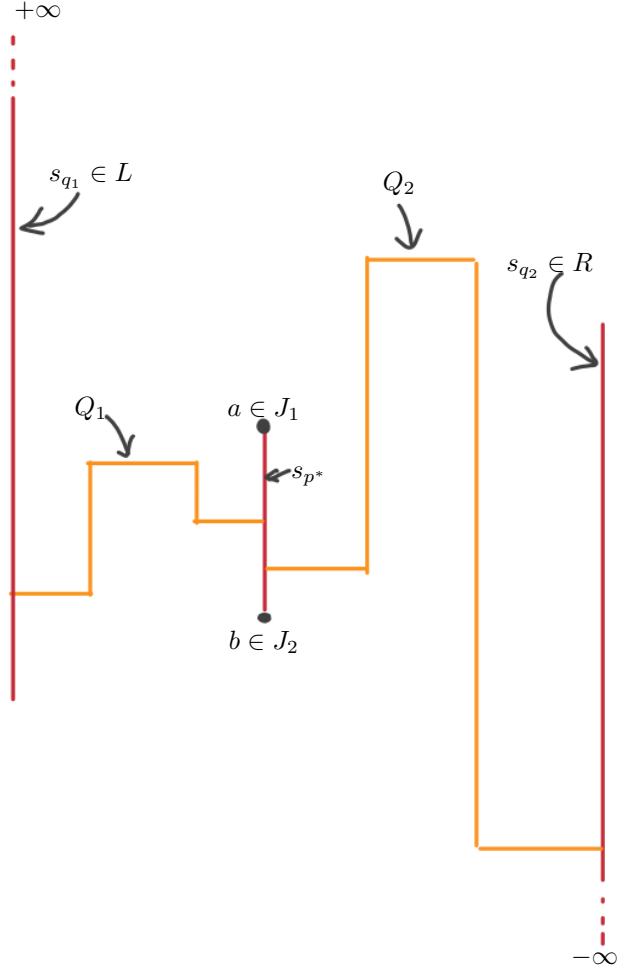

\begin{figure}
%\begin{center}
\begin{tikzpicture}
\node at (3,3) {\includegraphics[width=9cm]{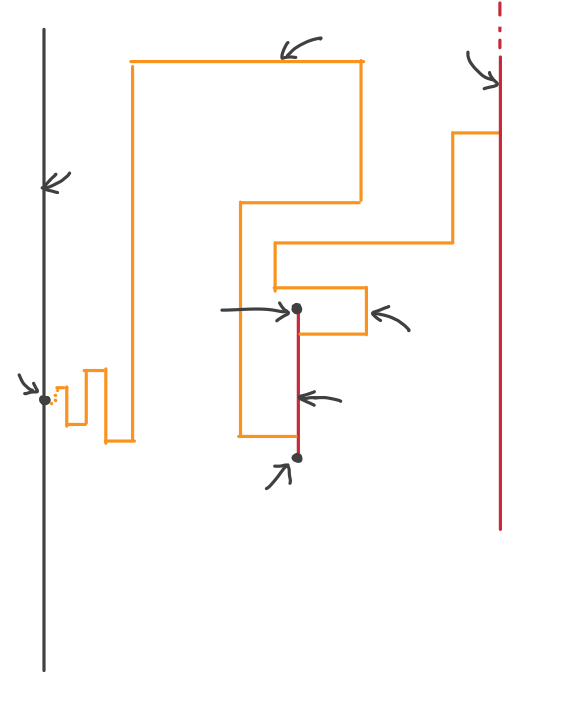}};
%\foreach \y in {-1, 0, 1, 2, ..., 7, 8, 9}
%{
%\draw (-2,\y) -- (7.5,\y);
%\node at (-0.5,\y){\y};
%};

%\foreach \x in {-2, -1, 0, 1,2,...,7}
%{
%\draw (\x,-1) -- (\x,9.5);
%\node at (\x,-0.5){\x};
%};
\node (A) at (4.2, 2.3) {$s_{p^*}$};
\node (B) at (5.8, 7.9) {$s_{q_2} \in R$};
%\node (C) at (3.5, 8.2) {$s_{q_1} \in L$};
\node (D) at (3.5, 8.2) {$Q_1$};
\node (E) at (5.2, 3.5) {$Q_2$};
\node (F) at (1.6, 3.8) {$a \in J_1$};
\node (G) at (2.5, 0.7) {$b \in J_2$};
\node (H) at (6.5, 8.8) {$+\infty$};
%\node (I) at (6.6, -2.9) {$-\infty$};
\node (J) at (-1.4, 2.8) {$p_{min}$};
\node (K) at (-0.3, 6) {$h_1$};
\end{tikzpicture}
\caption{Ruling out Case $2$. $J$ must cross $Q_1 \cup Q_2 \cup \{s_{p^*} \}$ at least once to go
from $a$ to $b$ and at least once to go back from $b$ to $a$. It cannot go to the left of the
vertical line $h_1$ or to the right of the infinite open segment $s_{q_2} \in R$. By the (very weak version)
of the JCT for infinite rectilinear polygons proved above, it can cross both times only at point $p_{min}$, a contradiction.}
\label{fig-case-2}
%\end{center}
\end{figure}

\noindent {\bf Proof:} Let $[s_{p^*}]$ denote the equivalence class of open segment $s_{p^*}$ i.e., the set of all open segments $s_{q}$ such that there is a rectilinear path $Q \subset \mathbb{R}^2 - J$ with finite number of turns from a point of $int(s_{p^*})$ to a point of $int(s_q)$.

Let $L$ be the set of all open segments $s_q \in [s_{p^*}]$ such that there exists a rectilinear
path $Q \subset \mathbb{R}^2 - J$ with a finite number of turns from $int(s_{p^*})$ to $int(s_q)$ such that the first horizontal edge of $Q$ goes to the left of the vertical line $l_{p^*}$. The set $R$ of 
open segments $s_q \in [s_{p^*}]$ consists of all segments $s_q$ such that there exists a rectilinear
path $Q \subset \mathbb{R}^2 - J$ with a finite number of turns from $int(s_{p^*})$ to $int(s_q)$ such that the first horizontal edge of $Q$ goes to the right of the vertical line $l_{p^*}$.
Note that $L \cap R$ can be non-empty.

There are three cases:

\begin{enumerate}
\item Each of  the two sets $L$ and $R$ contains at least one open segment of infinite length.
\item $L$ contains an open segment of infinite length, but $R$ contains only open segments of finite length.
\item $R$ contains an open segment of infinite length, but $L$ contains only open segments of finite length.
\item Both $L$ and $R$ contain only open segments of finite length. 
\end{enumerate}

Case $1$ cannot occur, as depicted in Figure \ref{fig-case-1}. 

Cases $2$ and $3$ are symmetric. Hence, we only consider Case $2$ here.
Let $J_1^l$ be the portion of $J_1$ from $p_{min}$ to the endpoint of $s_{p^*}$ in $J_1$.
Let $J_2^l$ be the portion of $J_2$ from the endpoint of $s_{p^*}$ in $J_2$ to $p_{min}$.
Then, the concatenation of Jordan arcs $J_1^l$, $s_{p^*}$ and $J_2^l$, in this order, form
a Jordan curve $J'$. 

The maximal sweepline algorithm, among all sweepline algorithm which start from
a horizontal sweep with one (right) vertical boundary at $s_{p^*}$, will sweep a finite, connected, open region $int(J')$ with boundary $J'$. This follows from the interior JCT for Jordan curve proved above, under the assumption that all reachable open segments are finite. In this case, this is true since
$L$ has been assumed to have only finite open segments.

In particular, in Case $2$, there exists a rectilinear path $Q'$ such that $int(Q') \subset \mathbb{R}^2 - J$, with at most countably infinite number of turns, from $p^*$ to $p_{min}$. We arrive at a contradiction, as depicted in Figure \ref{fig-case-2}. (It may be the case that, instead of $p_{min}$ we visit $p_{max}$ and the boundary of the region swept is $J - J'$. The proof in this case is same as above.)

Thus, only Case $4$ can cccur, and hence the statement of the lemma is proved. $\blacksquare$\\

%% file: exterior-jct-001.tex
\section{Proving the existence of exterior}
In this section, we prove the following theorem:

\begin{theorem}
\label{thm-ext-jct}
{\bf Exterior Jordan Curve Theorem.} There exists an unbounded, open and connected set $ext'(J) \subseteq \mathbb{R}^2 - ( int(J) \cup J)$ such that (i) $ext'(J)$ contains all infinite open segments and (ii) $bd(ext'(J)) = J$.

\end{theorem}

\noindent {\bf Proof:} The existence of $int(J)$ was proved above. We apply a translation $\tau$ so that the point $p^*$ moves to the origin of the Euclidean plane. Next, we apply the inversion map
$i: \mathbb{R}^2 \rightarrow \mathbb{R}^2$ to the Euclidean plane ($i$ sends circles of radius $r$ to circles of radius $\frac{1}{r}$ for $r \geq 0$).

First, note that the image $i(J)$ of the Jordan curve is also a Jordan curve. This is because
the map $\psi': \mathbb{S}^1 \rightarrow \mathbb{R}^2 - \{ {\bf 0} \}$ defined as the composition $\psi' = i \circ \tau \circ \psi$ is both one-to-one and continuous, since $\psi: \mathbb{S}^1 \rightarrow \mathbb{R}^2 - \{p^*\}$, $\tau: \mathbb{R}^2 - \{ p^* \} \rightarrow \mathbb{R}^2 - \{ {\bf 0} \}$ and $i: \mathbb{R}^2 - \{ {\bf 0} \} \rightarrow \mathbb{R}^2 - \{ {\bf 0} \}$ are all one-to-one and continuous maps. 

By applying the interior Jordan curve theorem to $i(J)$, we obtain a open, bounded and connected
region $int(i(J))$. 

We define $ext'(J) = i^{-1} ( int(i(J)) - \{ {\bf 0} \})$. $ext'(J)$ is an open set since $int(i(J))$ is an open set and $i^{-1}$ is a continuous map.

Since ${\bf 0} \in int(i(J))$, there exists an $\epsilon > 0$ such that the open ball $B(p, \epsilon) \subset int(i(J))$. Thus, the set $ext'(J)$ contains points with arbitrarily large $x$- and $y$-coordinates. Hence, $ext'(J)$ is an unbounded set.

Note that $int(i(J))$ is path connected. Take any two points $x, y \in int(i(J)) - \{ 0 \}$. Let $\epsilon_1 > 0$ be a real number such that $\epsilon_1 < \frac{\epsilon}{2}$ and $\min ( d(x, {\bf 0}), d(y, {\bf 0}) ) > \epsilon_1$. Let $Q \subset \mathbb{R}^2 - i(J)$ be a rectilinear path with a finite number of turns between points $x$ and $y$. If $Q$ does not intersect $bd(B({\bf 0}, \epsilon_1))$, we set $Q' = Q$. If $Q$ intersects $bd(B({\bf 0}, \epsilon_1))$, let $a$ and $b$ be the first and last points of intersection (since $cl(B({\bf 0}, \epsilon_1))$ is a closed set, these points exist). Let $Q'$ be the path 
obtained by following $Q$ from start to $a$, following the circle $bd(B({\bf 0}, \epsilon_1))$ from
$a$ to $b$, and then following the portion of $Q$ from $b$ to the end. Note that $Q' \subset \mathbb{R}^2 - i(J) \cup \{ {\bf 0} \}$. Thus, $i^{-1}(Q')$ is a continuous path from $i^{-1}(x)$ to $i^{-1}(y)$ and
$i^{-1}(Q') \subset \mathbb{R}^2 - J$. Thus, $ext'(J)$ is path connected.

Hence, the theorem is proved.
$\blacksquare$

%% file: jct-001.tex
\section{The Jordan curve theorem}
\label{sec-JCT}

\begin{theorem}
\label{thm-sweep-jct}
{\bf Jordan curve theorem (using sweepline algorithm).} $\mathbb{R}^2 - J$ is the disjoint union of two open and connected sets $int(J)$ and $ext(J)$, where $int(J)$ is a bounded set and $ext(J)$ is an unbounded set.
\end{theorem}

\begin{figure}[h]
%\begin{center}
\begin{tikzpicture}
\node at (3,3) {\includegraphics[width=9cm]{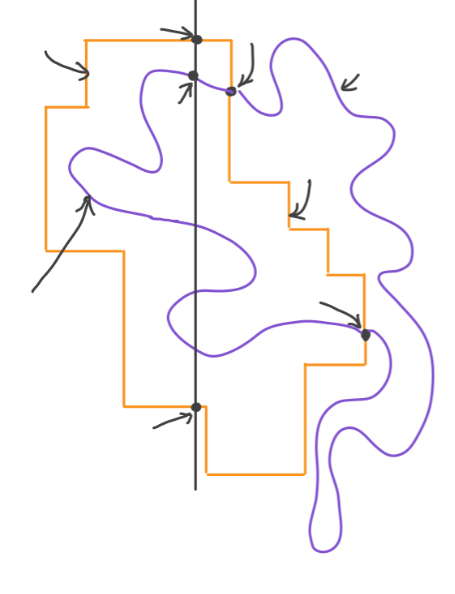}};
%\foreach \y in {-1, 0, 1, 2, ..., 7, 8, 9}
%{
%\draw (-2,\y) -- (7.5,\y);
%\node at (-0.5,\y){\y};
%};

%\foreach \x in {-2, -1, 0, 1,2,...,7}
%{
%\draw (\x,-1) -- (\x,9.5);
%\node at (\x,-0.5){\x};
%};
\node (A) at (3.5, 8.2) {$a$};
\node (B) at (4.6, 3.1) {$b$};
\node (D) at (4.5, 5.6) {$Q_1$};
\node (E) at (-0.4, 8.1) {$Q_2$};
\node (F) at (-1.1, 3) {$J_1$};
\node (G) at (5.7, 7.5) {$J_2$};
\node (H) at (1.4, 8.2) {$\beta_1$};
\node (J) at (1.4, 0.4) {$\beta_2$};
\node (K) at (2, 6.5) {$x$};
\end{tikzpicture}
\caption{Only two components. An open segment with one endpoint at $x$ is finite since both $s_{\beta_1}$ and 
$s_{\beta_2}$ are finite open segments and they lie on opposite sides of $x$, on the vertical
line through $x$.}
\label{fig-only-two}
%\end{center}
\end{figure}
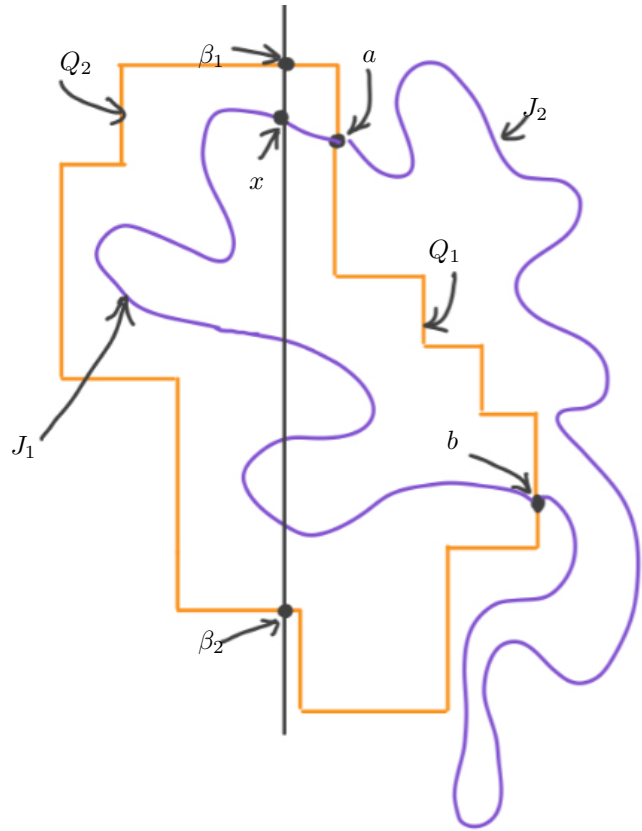

\noindent {\bf Proof:} The existence of $int(J)$ and $ext'(J)$ follows from Theorems \ref{thm-int-jct} and \ref{thm-ext-jct} (both in main paper). 

We rule out the existence of any other open and connected set $A$ with $bd(A) = J$, which can be obtained as the result of a maximal sweepline algorithm starting from some point $q \in \mathbb{R}^2 - J$. First, note that $A$ is a bounded set. This is because $ext'(J)$ has all infinite open segments, since all infinite open segments are in the same equivalence class (they connect together at the circle at infinity).

Thus, we conclude that there are two bounded, connected and open set $int(J)$ and $A$ with 
common boundary $J$.

Let $a$ and $b$ be two distinct points on the Jordan curve $J$. There exists a rectilinear
path $Q_1$ in $int(J)$ between $a$ and $b$. Further, there exists a rectilinear path $Q_2$
in $A$ between $a$ and $b$. These two rectilinear paths can each have countably infinite
turns. 

Let $R$ be the rectilinear polygon with infinite number of edges formed by the union of paths $Q_1$
and $Q_2$ (see Figure \ref{fig-only-two}). By the Jordan curve theorem for infinite rectilinear polygons, $R$ partitions the plane into two regions, the interior $int(R)$ and the exterior $ext(R)$. 

The Jordan curve $J$ can intersect $R$ only at the two points $a$ and $b$. Thus, $J$ is
partitioned into two Jordan arcs $J_1$ and $J_2$, where $J_1 - \{a, b\} \subset int(R)$
and $J_2 - \{a, b\} \subset ext(R)$. Thus, for any point $x \in J_1 - \{a, b\}$, any open segment
with an endpoint at $x$ is a finite segment. 

Note that $a$ and $b$ were arbitrary points of $J$. Hence, for all points $x \in J$, except a single
exceptional point $x^* \in J$, any open segment with an endpoint at $x$ is a finite segment.

Suppose there are at least two exceptional points. Let $x_1$ and $x_2$ be two distinct
exceptional points. Let $(a_1, b_1), (a_2, b_2), \ldots$ be a sequence of nested intervals of $\mathbb{S}^1$ with length tending to $0$, such that each interval contains $x_1$ in its interior.
For each $i \in \mathbb{N}$, the arc $\alpha_i$ of $\mathbb{S}^1$ with endpoints at 
$a_i$ and $b_i$ and not containing $x_1$, corresponds to the Jordan arc $J_1$ defined above in
$\mathbb{R}^2$ (with
$a_i$ and $b_i$ in place of $a$ and $b$, respectively).
Note that $\cup_{i=1}^{\infty} \alpha_i = \mathbb{S}^1 - \{ x_1 \}$ and we arrive at a contradiction.

This contradicts the conclusion above that $bd(ext'(J)) = J$.
Hence, there are only two open and connected components of $\mathbb{R}^2 - J$: $int(J)$ and
$ext(J)=ext'(J)$. This completes the proof of the Jordan curve theorem using the sweepline algorithm.
$\blacksquare$